\def\jacm{0}

\ifnum\jacm=1

\else
\def\draft{1}
\fi 

\def\merge{1}

\ifnum\jacm=1

\documentclass[acmsmall,screen,review]{acmart}

\usepackage{amsthm}
\usepackage{mathtools,amsmath}

\usepackage{bm}
\usepackage{enumitem}
\usepackage{dsfont}
\usepackage{float}
\usepackage{thm-restate}
\usepackage{bbm}

\AtBeginDocument{%
  \providecommand\BibTeX{{%
    \normalfont B\kern-0.5em{\scshape i\kern-0.25em b}\kern-0.8em\TeX}}}

\setcopyright{acmcopyright}
\copyrightyear{2022}
\acmYear{2022}
\acmDOI{XXXXXXX.XXXXXXX}

\acmJournal{JACM}
\acmVolume{1}
\acmNumber{1}
\acmArticle{1}
\acmMonth{1}

\newcommand{\mnote}[1]{\ifnum\draft=1 {\color{red} \textbf{Madhu's note:} #1}\fi}
\newcommand{\cnote}[1]{\ifnum\draft=1 {\color{brown} \textbf{Chi-Ning's note:} #1}\fi}
\newcommand{\snote}[1]{\ifnum\draft=1 {\color{orange} \textbf{Sasha's note:} #1}\fi}
\newcommand{\vnote}[1]{\ifnum\draft=1 {\color{green!40!black} \textbf{Santhoshini's note:} #1}\fi}

\newcommand{\One}{\mathds{1}}
\newcommand{\GOOD }{\textsf{GOOD}}
\renewcommand{\epsilon}{\varepsilon}
\newcommand{\eps}{\varepsilon}

\usepackage[breakable]{tcolorbox}
\newtcolorbox{reduction}[2][]
{
  breakable,
  colframe = gray!50,
  colback  = gray!10,
  coltitle = gray!10!black,
  before skip = 10pt,
  after skip = 10pt,
  title    = \textbf{#2},
  #1,
}

\newtcolorbox{examplebox}[2][]
{
  breakable,
  colframe = gray!50,
  colback  = gray!10,
  coltitle = gray!10!black,
  before skip = 10pt,
  after skip = 10pt,
  title    = \textbf{#2},
  #1,
}

\usepackage{algorithmicx}
\usepackage{algorithm}
\usepackage[noend]{algpseudocode}
\newcounter{algsubstate}
\renewcommand{\thealgsubstate}{\alph{algsubstate}}
\newenvironment{algsubstates}
  {\setcounter{algsubstate}{0}%
   \renewcommand{\State}{%
     \stepcounter{algsubstate}%
     \Statex {\footnotesize\thealgsubstate:}\space}}
  {}
\algnewcommand\algorithmicinput{\textbf{Input:}}
\algnewcommand\Input{\item[\algorithmicinput]}

\algnewcommand\algorithmicoutput{\textbf{Output:}}
\algnewcommand\Output{\item[\algorithmicoutput]}

\algnewcommand\algorithmicgoal{\textbf{Goal:}}
\algnewcommand\Goal{\item[\algorithmicgoal]}

\usepackage{amsthm}
\usepackage{thmtools,thm-restate}

\numberwithin{equation}{section}

\ifnum\jacm=1 
\usepackage[nameinlink,capitalise]{cleveref}
\else 
\usepackage[colorlinks=true, allcolors=blue]{hyperref}
\usepackage[nameinlink,capitalise]{cleveref}
\hypersetup{
    citecolor={violet}
}
\fi 

\declaretheoremstyle[bodyfont=\it,qed=\qedsymbol]{noproofstyle}

\declaretheorem[name=Observation,numbered=no]{observation*}

\declaretheorem[numberlike=equation]{theorem}

\declaretheorem[name=Theorem,numbered=no]{theorem*}

\declaretheorem[numberlike=equation]{lemma}
\declaretheorem[name=Lemma,numbered=no]{lemma*}

\declaretheorem[numberlike=equation]{corollary}
\declaretheorem[name=Corollary,numbered=no]{corollary*}

\declaretheorem[numberlike=equation]{proposition}
\declaretheorem[name=Proposition,numbered=no]{proposition*}

\declaretheorem[numberlike=equation]{claim}
\declaretheorem[name=Claim,numbered=no]{claim*}

\declaretheorem[name=Conjecture,numbered=no]{conjecture*}

\declaretheorem[name=Question,numbered=no]{question*}

\declaretheoremstyle[bodyfont=\it]{defstyle} 

\declaretheorem[numberlike=equation,style=defstyle]{definition}
\declaretheorem[unnumbered,name=Definition,style=defstyle]{definition*}

\declaretheorem[unnumbered,name=Notation=defstyle]{notation*}

\declaretheorem[unnumbered,name=Construction,style=defstyle]{construction*}

\declaretheoremstyle[]{rmkstyle} 

\declaretheorem[numberlike=equation,style=rmkstyle]{remark}

\declaretheorem[unnumbered,name=Example,style=rmkstyle]{example*}

\newcommand{\opt}{\mathop{\mathrm{opt}}}
\newcommand{\gbmaxf}{{(\gamma,\beta)\textrm{-}\maxf}}

\newcommand{\C}{\mathcal{C}}
\newcommand{\Exp}{\mathop{\mathbb{E}}}
\newcommand{\cP}{\mathcal{P}}
\newcommand{\cA}{\mathcal{A}}
\newcommand{\cB}{\mathcal{B}}
\newcommand{\cD}{\mathcal{D}}
\newcommand{\cF}{\mathcal{F}}
\newcommand{\cI}{\mathcal{I}}
\newcommand{\poly}{\mathsf{poly}}
\newcommand{\cY}{\mathcal{Y}}
\newcommand{\cN}{\mathcal{N}}
\newcommand{\N}{\mathbb{N}}
\newcommand{\F}{\mathbb{F}}
\newcommand{\R}{\mathbb{R}}
\newcommand{\Z}{\mathbb{Z}}
\newcommand{\bern}{\mathsf{Bern}}
\newcommand{\unif}{\mathsf{Unif}}
\newcommand{\maxf}{\textsf{Max-CSP}(f)}
\newcommand{\maxF}{\textsf{Max-CSP}(\cF)}
\newcommand{\gbmaxF}{{(\gamma,\beta)\textrm{-}\maxF}}
\newcommand{\bias}{\textsf{bias}}
\newcommand{\val}{\textsf{val}}
\newcommand{\sign}{\textsf{sign}}
\newcommand{\ALG}{\mathbf{ALG}}
\newcommand{\supp}{\textsf{supp}}
\newcommand{\yes}{\textbf{YES}}
\newcommand{\no}{\textbf{NO}}
\newcommand{\mdcut}{\textsf{Max-DICUT}}
\newcommand{\mcut}{\textsf{Max-CUT}}
\newcommand{\sgyf}{S_\gamma^Y(\cF)}
\newcommand{\sbnf}{S_\beta^N(\cF)}
\newcommand{\kgyf}{K_\gamma^Y(\cF)}
\newcommand{\kbnf}{K_\beta^N(\cF)}
\newcommand{\Disj}{\textsf{Disj}}
\newcommand{\GHD}{\textsf{GHD}}
\newcommand{\kand}{k\textsf{AND}}
\newcommand{\maxkand}{\textsf{Max-}\kand}
\newcommand{\Th}{\mathsf{Th}}
\newcommand{\maxthreeand}{\textsf{Max-}3\textsf{AND}}
\newcommand{\maxtwoand}{\textsf{Max-}2\textsf{AND}}

\newcommand{\veca}{\mathbf{a}}
\newcommand{\vecb}{\mathbf{b}}
\newcommand{\vecc}{\mathbf{c}}
\newcommand{\vecd}{\mathbf{d}}

\newcommand{\vecj}{\mathbf{j}}
\newcommand{\vecJ}{\mathbf{J}}
\newcommand{\vecs}{\mathbf{s}}
\newcommand{\vecu}{\mathbf{u}}
\newcommand{\vecv}{\mathbf{v}}
\newcommand{\vecw}{\mathbf{w}}
\newcommand{\vecx}{\mathbf{x}}
\newcommand{\vecy}{\mathbf{y}}
\newcommand{\vecz}{\mathbf{z}}

\newcommand{\vecsigma}{\boldsymbol{\sigma}}
\newcommand{\veclambda}{\boldsymbol{\lambda}}
\newcommand{\vecmu}{\boldsymbol{\mu}}
\newcommand{\vecphi}{\boldsymbol{\phi}}

\newcommand{\veczero}{\mathbf{0}}
\newcommand{\gd}{\mathcal{Q}}
\newcommand{\vecA}{\mathbf{A}}
\newcommand{\cfk}{\cF(\gd)}
\newcommand{\vecnu}{\boldsymbol{\nu}}

\newcommand{\pstrm}{\text{pad-stream}}
\newcommand{\SD}{\textsf{SD}}

\newcommand{\psSD}{\textsf{padded-streaming-SD}}
\newcommand{\strm}{\textrm{stream}}
\newcommand{\simul}{\textrm{simul}}
\newcommand{\COMP}{\textsf{SKETCH}}
\newcommand{\COMB}{\textsf{COMB}}

\newcommand{\cC}{\mathcal{C}}
\newcommand{\canon}{\mathbb{I}}

\else 

\documentclass[11pt]{article}

\usepackage[utf8]{inputenc}
\usepackage{fullpage}
\usepackage{amsthm}
\usepackage{mathtools,amsmath,amssymb}

\usepackage{bm}
\usepackage{enumitem}
\usepackage{mathtools}
\usepackage{dsfont}
\usepackage{float}
\usepackage{thm-restate}
\usepackage{bbm}

\newcommand{\mnote}[1]{\ifnum\draft=1 {\color{red} \textbf{Madhu's note:} #1}\fi}
\newcommand{\cnote}[1]{\ifnum\draft=1 {\color{brown} \textbf{Chi-Ning's note:} #1}\fi}
\newcommand{\snote}[1]{\ifnum\draft=1 {\color{orange} \textbf{Sasha's note:} #1}\fi}
\newcommand{\vnote}[1]{\ifnum\draft=1 {\color{green!40!black} \textbf{Santhoshini's note:} #1}\fi}

\newcommand{\One}{\mathds{1}}
\newcommand{\GOOD }{\textsf{GOOD}}
\renewcommand{\epsilon}{\varepsilon}
\newcommand{\eps}{\varepsilon}

\usepackage[breakable]{tcolorbox}
\newtcolorbox{reduction}[2][]
{
  breakable,
  colframe = gray!50,
  colback  = gray!10,
  coltitle = gray!10!black,
  before skip = 10pt,
  after skip = 10pt,
  title    = \textbf{#2},
  #1,
}

\newtcolorbox{examplebox}[2][]
{
  breakable,
  colframe = gray!50,
  colback  = gray!10,
  coltitle = gray!10!black,
  before skip = 10pt,
  after skip = 10pt,
  title    = \textbf{#2},
  #1,
}

\usepackage{algorithmicx}
\usepackage{algorithm}
\usepackage[noend]{algpseudocode}
\newcounter{algsubstate}
\renewcommand{\thealgsubstate}{\alph{algsubstate}}

\algnewcommand\algorithmicinput{\textbf{Input:}}
\algnewcommand\Input{\item[\algorithmicinput]}

\algnewcommand\algorithmicoutput{\textbf{Output:}}
\algnewcommand\Output{\item[\algorithmicoutput]}

\algnewcommand\algorithmicgoal{\textbf{Goal:}}
\algnewcommand\Goal{\item[\algorithmicgoal]}

\fi

\begin{document}

\newcommand{\abstracttext}{
A constraint satisfaction problem (CSP), $\maxF$, is specified by a finite set of constraints $\cF \subseteq \{[q]^k \to \{0,1\}\}$ for positive integers $q$ and $k$. An instance of the problem on $n$ variables is given by $m$ applications of constraints from $\cF$ to subsequences of the $n$ variables, and the goal is to find an assignment to the variables that satisfies the maximum number of constraints. In the $(\gamma,\beta)$-approximation version of the problem for parameters $0 \leq \beta < \gamma \leq 1$, the goal is to distinguish instances where at least $\gamma$ fraction of the constraints can be satisfied from instances where at most $\beta$ fraction of the constraints can be satisfied. 

In this work, we consider the approximability of this problem in the context of sketching algorithms and give a dichotomy result. Specifically,  for every family $\cF$ and every $\beta < \gamma$,  we show that either a linear sketching algorithm solves the problem in polylogarithmic space, or the problem is not solvable by any sketching algorithm in $o(\sqrt{n})$ space. In particular, we give non-trivial approximation algorithms using polylogarithmic space for infinitely many constraint satisfaction problems.

We also extend previously known lower bounds for general streaming algorithms
to a wide variety of problems, and in particular the case of $q=k=2$, where we get a dichotomy, and the case  when the satisfying assignments of the constraints of $\cF$ support a distribution on $[q]^k$ with uniform marginals.

Prior to this work, other than sporadic examples, the only systematic classes of CSPs that were analyzed considered the setting of Boolean variables $q=2$, binary constraints $k=2$, singleton families $|\cF|=1$ and only considered the setting where constraints are placed on literals rather than variables.

Our positive results show wide applicability of bias-based algorithms used previously by \cite{GVV17} and \cite{CGV20}, which we extend to include richer norm estimation algorithms, by giving a systematic way to discover biases. Our negative results combine the Fourier analytic methods of \cite{KKS}, which we extend to a wider class of CSPs, with a rich collection of reductions among communication complexity problems that lie at the heart of the negative results. 
In particular, previous works used Fourier analysis over the Boolean cube to initiate their results and the results seemed particularly tailored to functions on Boolean literals (i.e., with negations). Our techniques surprisingly allow us to get to general $q$-ary CSPs without negations by appealing to the same Fourier analytic starting point over Boolean hypercubes.}

\ifnum\jacm=0

\title{Sketching approximability of all finite CSPs\ifnum\merge=1\footnote{This paper subsumes~\cite{CGSV21} which in turn replaced the withdrawn paper \cite{CGSV20}.}\fi}
\author{
Chi-Ning Chou\thanks{School of Engineering and Applied Sciences, Harvard University, Cambridge, Massachusetts, USA. Supported by NSF awards CCF 1565264 and CNS 1618026. Email: \texttt{chiningchou@g.harvard.edu}.}
\and Alexander Golovnev\thanks{Department of Computer Science, Georgetown University. Email: \texttt{alexgolovnev@gmail.com}.}
\and Madhu Sudan\thanks{School of Engineering and Applied Sciences, Harvard University, Cambridge, Massachusetts, USA. Supported in part by a Simons 
Investigator Award and NSF Awards CCF 1715187 and CCF 2152413. Email: \texttt{madhu@cs.harvard.edu}.}
\and Santhoshini Velusamy\thanks{Toyota Technological Institute, Chicago, Illinois, USA. Supported in part by a Google PhD fellowship, a Simons 
Investigator Award to Madhu Sudan, and NSF Awards CCF 1715187 and CCF 2152413. Email: \texttt{santhoshini@ttic.edu}.}
}

\date{}
\sloppy
\maketitle

\begin{abstract}

\abstracttext

\end{abstract}

\newpage

\tableofcontents

\newpage

\else 

\title{Sketching approximability of all finite CSPs}

\author{Chi-Ning Chou}
\authornote{Supported by NSF awards CCF 1565264 and CNS 1618026.}
\email{chiningchou@g.harvard.edu}
\affiliation{%
  \institution{School of Engineering and Applied Sciences, Harvard University}
  \city{Cambridge}
  \state{Massachusetts}
  \country{USA}
}

\author{Alexander Golovnev}
\email{alexgolovnev@gmail.com}
\affiliation{%
  \institution{Department of Computer Science, Georgetown University}
  \country{USA}
}

\author{Madhu Sudan}
\authornote{Supported in part by a Simons 
Investigator Award and NSF Awards CCF 1715187 and CCF 2152413.}
\email{madhu@cs.harvard.edu}
\affiliation{%
\institution{School of Engineering and Applied Sciences, Harvard University}
  \city{Cambridge}
  \state{Massachusetts}
  \country{USA}
}

\author{Santhoshini Velusamy}
\authornote{Supported in part by a Google Ph.D. Fellowship, a Simons 
Investigator Award to Madhu Sudan, and NSF Awards CCF 1715187 and CCF 2152413.}
\email{santhoshini@ttic.edu}
\affiliation{%
\institution{Toyota Technological Institute}
  \city{Chicago}
  \state{Illinois}
  \country{USA}
}



\begin{abstract}

\abstracttext

\end{abstract}

\ccsdesc[500]{Theory of computation~Sketching and sampling}
\ccsdesc[500]{Theory of computation~Communication complexity}
\ccsdesc[500]{Theory of computation~Approximation algorithms analysis}

\keywords{streaming algorithms, communication lower bound, inapproximability, constraint satisfaction problem}

\maketitle

\fi 




\section{Introduction}\label{sec:intro}

In this paper we give a complete characterization of the approximability of constraint satisfaction problems (CSPs) by sketching algorithms. We describe the exact class of problems below, and give a brief history of previous work before giving our results.

\subsection{CSPs}

For positive integers $q$ and $k$, a $q$-ary {\em constraint satisfaction problem} (CSP) is given by a (finite) set of constraints 
$\cF \subseteq \{f:[q]^k \to \{0,1\}\}$. A constraint $C$ on $x_1,\ldots,x_n$ is given by a pair $(f,\vecj)$, with $f \in \cF$ and $\vecj = (j_1,\ldots,j_k) \in [n]^k$ where the coordinates of $\vecj$ are all distinct.\footnote{To allow repeated variables in a constraint, note that one can turn $\cF$ into $\cF'$ by introducing new functions corresponding to all the possible replications of variables of functions in $\cF$.} An assignment $\vecb \in [q]^n$ satisfies $C = (f,\vecj)$ if $f(b_{j_1},\ldots,b_{j_k})=1$. 
To every finite set $\cF$, we associate a maximization problem $\maxF$ that is defined as follows: An instance $\Psi$ of $\maxF$ consists of $m$ constraints $C_1,\ldots,C_m$ applied to $n$ variables $x_1,x_2,\dots,x_n$ along with $m$ non-negative integer weights $w_1,\ldots,w_m$. The value of an assignment $\vecb \in [q]^{n}$ on an instance $\Psi = (C_1,\ldots,C_m; w_1,\ldots,w_m)$, denoted $\val_\Psi(\vecb)$, is the fraction of weight of constraints satisfied by $\vecb$. The goal of the {\em exact} problem is to compute the maximum, over all assignments, of the value of the assignment on the input instance, i.e., to compute, given $\Psi$, the quantity $\val_\Psi = \max_{\vecb \in [q]^{n}}\{\val_\Psi(\vecb)\}$. 

In this work we consider the approximation version of $\maxF$, which we study in terms of the ``gapped promise problems''. Specifically given $0 \leq \beta < \gamma \leq 1$, the $(\gamma,\beta)$-approximation version of $\maxF$, abbreviated $(\gamma,\beta)\textrm{-}\maxF$, is the task of distinguishing between instances from $\Gamma = \{\Psi | \opt(\Psi) \geq \gamma\}$ and instances from $B = \{\Psi | \opt(\Psi) \leq \beta\}$. It is well-known that this distinguishability problem is a refinement of the usual study of approximation which usually studies the ratio of $\gamma/\beta$ for tractable versions of $\gbmaxF$. See \cref{prop:approx-equivalence} for a formal statement in the context of streaming approximability of $\maxF$ problems.

\subsection{Streaming algorithms}

We study the complexity of $(\gamma,\beta)$-$\maxF$ in the setting of randomized streaming algorithms. Here, an instance $\Psi = (C_1,\ldots,C_m)$ is presented as a stream $\sigma_1,\sigma_2,\ldots,\sigma_m$ with $\sigma_i = (f(i),\vecj(i))$ representing the $i$th constraint.
We study the space required to solve the $(\gamma,\beta)$-approximation version of $\maxF$. Specifically we consider algorithms that are allowed to use internal randomness and $s$~bits of space. The algorithms output a single bit at the end. They are said to solve the $(\gamma,\beta)$-approximation problem correctly if they output the correct answer with probability at least $2/3$ (i.e., they err with probability at most $1/3$).

A sketching algorithm is a special class of a streaming algorithm, where the algorithm's output is determined by a small sketch it produces of the input stream, and the sketch itself has the property that the sketch of the concatenation of two streams can be computed from the sketches of the two component streams. (See \cref{def:sketching alg} for a formal definition.) 

For over a decade now, there has been active research on designing streaming and sketching algorithms for combinatorial optimization problems in various settings. See for example, 
\begin{itemize}
    \item \cite{GKK12,AKL16,AKLY16,AKL17,GVV17,KKSV17,KK19,SPass21,SPass22,SPass22a,SPass22b,SPass22C,SpAss22d,SPass22e,SPass23,SPass23a} for results in the \emph{single-pass} setting where the algorithm is allowed only a single pass through the stream,
    \item \cite{MPass20,assadi2020multi,MPass21,AN21,MPass21a,MPass22,MPass22b,MPass23,Mpass23a} for results on \emph{multi-pass} streaming algorithms which are allowed a constant number of passes through the stream, and
    \item \cite{Rand14,KKS,Rand21,Rand21a,Rand22,Rand23} for results in the \emph{random-ordering} setting where the input is randomly shuffled in the stream.
\end{itemize}

We primarily focus on single-pass streaming algorithms and our main dividing line is between algorithms that work with space $\poly(\log n)$, versus algorithms that require space at least $n^\epsilon$ for some $\epsilon > 0$. In informal usage we refer to a streaming problem as ``easy'' if it can be solved with polylogarithmic space (the former setting) and ``hard'' if it requires polynomial space for sketching algorithms. 
We note that all the positive results (algorithms) given in this paper are linear sketching algorithms which are more restrictive than general sketching algorithms.
We also note that many of our lower bounds work against general streaming algorithms and we elaborate on this in \cref{ssec:results}.

\subsection{Past work}

To our knowledge, streaming algorithms for CSPs have not been investigated extensively. Here we cover the few results we are aware of, all of which consider only the Boolean ($q=2$) setting. 
On the positive side, it may be surprising that there exists any non-trivial algorithm at all. (Briefly, we say that an algorithm that outputs a constant value independent of the input is ``trivial''.)

It turns out that there do exist some non-trivial approximation algorithms for Boolean CSPs.  This was established by the work of Guruswami, Velingker, and Velusamy~\cite{GVV17} who, in our notation, gave an algorithm for the $(\gamma,2\gamma/5-\epsilon)$-approximation version of $\textsf{Max-2AND}$, for every $\gamma \in [0,1]$ ($\maxtwoand$ is the $\maxF$ problem corresponding to $\cF = \{f_{c,d} | c,d \in \{0,1\}\}$ where $f_{c,d}(a,b) = 1$ if $a =c$ and $b = d$ and $f_{c,d}(a,b) = 0$ otherwise). A central ingredient in their algorithm is the ability of streaming algorithms to approximate the $\ell_1$ norm of a vector in the turnstile setting, which allows them to estimate the ``bias'' of $n$ variables (how often they occur positively in constraints, as opposed to negatively). Subsequently, the work of Chou, Golovnev, and Velusamy~\cite{CGV20} further established the utility of such algorithms, which we refer to as bias-based algorithms, by giving optimal algorithms for all Boolean CSPs on $2$ variables. In particular they give a better (optimal!) analysis of bias-based algorithms for $\textsf{Max-2AND}$, and show that $\textsf{Max-2SAT}$ also has an optimal algorithm based on bias. 

On the negative side, the problem that has been explored the most is $\textsf{Max-CUT}$, or in our language $\textsf{Max-2XOR}$, which corresponds to $\cF = \{f\}$ and $f(x,y) = x \oplus y$.
Kapralov, Khanna, and Sudan~\cite{KKS} showed that $\textsf{Max-2XOR}$ does not have a $(1,1/2 + \epsilon)$-approximation algorithm using $o(\sqrt{n})$-space, for any $\epsilon > 0$.  This was subsequently improved upon by Kapralov, Khanna, Sudan, and Velingker~\cite{KKSV17}, and Kapralov and Krachun~\cite{KK19}. The final paper~\cite{KK19} completely resolves $\textsf{Max-CUT}$ showing that $(1,1/2 + \epsilon)$-approximation for these problems requires $\Omega(n)$ space. Turning to other problems, the work by \cite{GVV17} notices that the $(1,1/2+\epsilon)$-inapproximability of $\textsf{Max-2XOR}$ immediately yields $(1,1/2+\epsilon)$-inapproximability of $\maxtwoand$ as well. In \cite{CGV20} more sophisticated reductions are used to improve the inapproximability result for $\maxtwoand$ to a $(\gamma,4\gamma/9 + \epsilon)$-inapproximability for some positive $\gamma$, which turns out to be the optimal ratio by their algorithm and analysis. As noted earlier their work gives algorithms for $\maxF$ for all $\cF \subseteq \{f:\{0,1\}^2 \to \{0,1\}\}$,\footnote{Note that when $q=2$ we switch to using $\{0,1\}$ or $\{-1,1\}$ as the domain (as opposed to $\{1,2\}$) depending on convenience.} which are optimal if $\cF$ is closed under literals (i.e., if $f(x,y) \in \cF$ then so are the functions $f(\neg x, y)$ and $f(\neg x, \neg y)$).

\subsection{Results}\label{ssec:results}

Our main theorem is a decidable dichotomy theorem for $(\gamma,\beta)$-$\maxF$ with sketching algorithms.

\begin{theorem}[Succinct version]\label{thm:informal}
For every $q,k\in\N$, $0 \leq \beta < \gamma \leq 1$ and $\cF \subseteq \{f:[q]^k \to \{0,1\}\}$, one of the following two conditions holds: Either $(\gamma,\beta)$-$\maxF$ can be solved with $O(\log^3 n)$ space by linear sketches, or for every $\epsilon > 0$, every sketching algorithm for $(\gamma-\epsilon,\beta+\epsilon)$-$\maxF$ requires $\Omega(\sqrt{n})$-space. Furthermore there is a polynomial space algorithm that decides which of the two conditions holds, given $\gamma,\beta$ and $\cF$. 
\end{theorem}

\cref{thm:informal} combines the more detailed \cref{thm:main-detailed-dynamic} with the polynomial space decidability coming from \cref{thm:pspace-dec}.

The first order message of the theorem statement is that the known non-trivial approximation algorithms for streaming CSPs (i.e., the algorithms for $\textsf{Max-2AND}$ and $\textsf{Max-2SAT}$ from \cite{CGSV20}) can {\em potentially} be extended to infinitely many problems.  To confirm this potential, one needs to be able to identify an infinite subclass of CSPs for which the decidability condition for non-trivial $(\gamma, \beta)$ pairs can be analytically shown to be ``solvable in polylog space''. While we do not find such explicit families in this paper, subsequent work has succeeded in getting such an analysis~\cite{BHPSV21,CGSSV-mon}. We elaborate further on this in \cref{ssec:subseq} but note that the subsequent work~\cite{BHPSV21} shows that 
 $\textsf{Max-kAND}$ (the generalization of $\textsf{Max-2AND}$ to $k$ literals) for every $k \in \N$ has non-trivial approximation algorithms thereby confirming this potential! We believe this in itself may be a surprising result to some given that the bias-based algorithms and their analysis did appear tailored to the structure of $\textsf{Max-2AND}$ and $\textsf{Max-2SAT}$.

The next main message is that when the class of algorithms we use cannot be used to solve a $(\gamma,\beta)$-approximation problem then there is an inherent hurdle and no sketching based algorithm can work. Indeed in many cases our results rule out completely general streaming algorithms, though we do not get a dichotomy for general streaming. 

Finally we highlight some of the descriptive strengths of the class of problems captured by \cref{thm:informal} above, we note that previous works 
could only handle the special case 
where (1) $\cF$ contains a single function $f$, (2) $q = 2$, (3) Constraints are placed on ``literals'' rather than variables and (4) They only capture a single parameter approximation problem not the more refined two parameter (``gapped'') version considered in this work. The difference in expressivity due to conditions (1)-(3) is significant: To capture a problem such as \textsf{Max-3SAT} one needs to go beyond restriction (1) to allow different constraints for clauses of length $1$, $2$, and $3$. This is a quantitatively significant restriction in that the approximability in this case is ``smaller'' than that of $\maxf$ for any of the constituent functions. So hard instances do involve a mix of constraints! The lack of expressiveness induced by the second restriction of Boolean variables is perhaps more obvious. Natural examples of CSPs that require larger alphabets are \textsf{Max-$q$-Coloring} and \textsf{Unique Games}.
Next we turn to restriction (3) --- the inability to capture CSP problems over variables. This restriction prevents previous works from capturing some very basic problems including \textsf{Max-CUT} and \textsf{Max-DICUT}. Furthermore, the notion of ``literals'' is natural only in the setting of Boolean variables --- so overcoming this restriction seems crucial to eliminating the restriction of Booleanity of the variables. Notice that while for families with a single function $\cF = \{f\}$, going from constraints on literals to constraints on variables does not lead to greater expressivity, once we study $\maxF$ for all sets $\cF$, the study does get formally richer. Finally the two parameter versions allow us to also understand the approximability of satisfiable and nearly-satisfiable instances of \textsf{Max-CSP}, a quest that is quite common in the literature. (See for instance the works on robust satisfiability~\cite{DalmauK13,BartoK12,KunOTYZ12}.) 

In particular \cref{thm:informal} allows us also to capture the extreme case of hard problems where no ``non-trivial'' algorithms exist. Such problems are usually referred to as approximation-resistant problems. In the study of Boolean CSPs, with constraints placed on literals, ``non-triviality'' is defined as ``beating a random assignment'' and approximation resistance in the setting of polynomial time algorithms is a well-studied topic~\cite{Has01,guruswami2011beating,AustrinMossel}. Extending the definition to the setting where constraints are placed on variables rather than literals, requires some thought. We propose a definition in this paper (see \cref{def:approx-res}) which uses the notion that algorithms outputting a constant value are trivial, and a problem is approximation resistant if beating this trivial algorithm is hard. Specifically, $\cF$ is said to be approximation resistant if for every $\beta < \gamma$ either $\gbmaxF$ is solved by a ``constant function'' or it requires $n^{\Omega(1)}$ space. We then show how \cref{thm:informal} (or its more detailed version \cref{thm:main-detailed-dynamic}) leads to a characterization of approximation-resistance in the streaming setting as well. (See \cref{cor:approx-res}.)

As mentioned earlier, the results above (and in particular the negative results) apply only to sketching algorithms for streaming CSPs. For general streaming algorithm, we get some partial results. To describe our next result, we define the notion of a function supporting a one-wise independent distribution.
We say that $f$ supports {\em one-wise independence} if there exists a distribution $\cD$ supported on $f^{-1}(1)$ whose marginals are uniform on $[q]$.
We say that $\cF$ supports one-wise independence if every $f \in \cF$ supports one-wise independence.

\begin{theorem}[Informal]\label{thm:one-wise-informal}
If $\cF\subseteq \{f:[q]^k \to \{0,1\}\}$ supports one-wise independence then it is approximation-resistant in the streaming setting.
\end{theorem}

\cref{thm:one-wise-informal} is formalized as \cref{thm:one-wise} in \cref{sec:lb-detail-insert}.
We also give theorems capturing hardness in the streaming setting beyond the 1-wise independent case. Stating the full theorem requires more notions (see \cref{sec:lb-detail-insert}), but as a consequence we get the following extension of theorem of \cite{CGV20}.

\begin{theorem}\label{thm:main-intro-k=q=2} Let $q = k = 2$. Then, for every family $\cF\subseteq \{f:[q]^2 \to \{0,1\}\}$, and for every $0 \leq \beta < \gamma \leq 1$, at least one of the following always holds:
\begin{enumerate}
    \item $(\gamma,\beta)$-$\maxF$ has an $O(\log^3 n)$-space linear sketching algorithm.
    \item For every $\epsilon > 0$, every streaming algorithm that solves $(\gamma-\epsilon,\beta+\epsilon)$-$\maxF$ requires $\Omega(\sqrt{n})$ space. If $\gamma = 1$, then $(1,\beta+\epsilon)$-$\maxF$ requires $\Omega(\sqrt{n})$ space. 
\end{enumerate}
Furthermore, for every $\ell\in\N$, there is an algorithm using space $\poly(\ell)$ that decides which of the two conditions holds  given the truth-tables of functions in  $\cF$, and $\gamma$ and $\beta$ as $\ell$-bit rationals.
\end{theorem}

\cref{thm:main-intro-k=q=2} is proved in \cref{sec:lb-detail-insert}. 
\cite{CGV20} study the setting where constraints are applied to literals, $\cF$ contains a single function and get a tight characterization of the approximability of $\maxF$\footnote{By approximability of $\maxF$ we refer to the quantity $\inf_\beta \sup_\gamma\{\{\beta/\gamma\}\}$ over polylog space solvable $(\gamma,\beta)$-$\maxF$ problems.}. 
 
Our work extends theirs by allowing constraints to be applied only to variables, and by allowing families of constraint functions, and by determining the complexity of every $\gbmaxF$ (and not just studying the optimal ratio of $\beta/\gamma$).  

For the sake of completeness we also give a simple characterization of the $\maxF$ problems that are solvable \emph{exactly} in polylogarithmic space. 

\begin{theorem}[Succinct version]\label{thm:informal2}
For every $q,k\in\N$ and $\cF \subseteq \{f:[q]^k \to \{0,1\}\}$, the $\maxF$ problem is solvable exactly in deterministic logarithmic space if and only if there is a constant $\sigma \in [q]$ such that every satisfiable function in $\cF$ is satisfied by the all $\sigma$-assignment. All remaining families $\cF$ require $\Omega(n)$ space to solve exactly. 
\end{theorem}

The proof of this theorem is by elementary reductions from standard communication complexity problems and is included in \cref{sec:exact}.

\paragraph{This version:} This version of the paper subsumes the works \cite{CGSV20,CGSV21,CGSV21-general-arxiv}. The paper~\cite{CGSV20}, now withdrawn, claimed a restriction of \cref{thm:informal} in the streaming setting, but that version had an error and the status of Theorem 1.1 in \cite{CGSV20} is currently open. \cite{CGSV21} proves the results of this paper for the special cases of $\cF = \{f\}$, $q=2$ and constraints being applied to literals rather than variables. \cite{CGSV21-general-arxiv}~essentially contains the same results as this paper, but builds upon \cite{CGSV21}. The conference version of \cite{CGSV21-general-arxiv} appeared in the proceedings of FOCS~2021~\cite{CGSV21-general-conference}. This paper combines \cite{CGSV21} and \cite{CGSV21-general-arxiv}.

\subsection{Contrast with dichotomies in the polynomial time setting}

The literature on polynomial-time dichotomies of $\maxf$ problems is vast. One broad family of results here \cite{Schaefer,Bulatov,Zhuk} 
considers the exact satisfiability problems (corresponding to distinguishing between instances from $\{\Psi | \opt(\Psi) = 1\}$ and instances from $\{\Psi | \opt(\Psi) < 1\}$).
Another family of results~\cite{raghavendra2008optimal,AustrinMossel,KhotTW} considers the approximation versions of $\maxf$ and gets ``near dichotomies'' along the lines of this paper --- i.e., they either show that the $(\gamma,\beta)$-approximation is easy (in polynomial time), or for every $\epsilon > 0$ the  $(\gamma - \epsilon,\beta+\epsilon)$-approximation version is hard (in some appropriate sense). Our work resembles the latter series of works both in terms of the nature of results obtained, the kinds of characterizations used to describe the ``easy'' and ``hard'' classes, and also in the proof approaches (though of course the sketching setting is much easier to analyze, allowing for simpler proofs overall and unconditional results). We summarize their results giving comparisons to our theorem and then describe a principal contrast. 

In a seminal work, Raghavendra~\cite{raghavendra2008optimal} gave a characterization of the polynomial time approximability of the $\maxf$ problems based on the unique games conjecture~\cite{Kho02}. Our~\cref{thm:informal} is analogous to his theorem.
A characterization of approximation resistant functions is given by Khot, Tulsiani and Worah~\cite{KhotTW}. Our~\cref{thm:one-wise-informal} is analogous to this. Austrin and Mossel~\cite{AustrinMossel} show that all functions supporting a pairwise independent distribution are approximation-resistant. Our~\cref{thm:one-wise} is analogous to this theorem. 

While our results run in parallel to the work on polynomial time approximability our characterizations are not immediately comparable. Indeed there are some significant differences which we highlight below. Of course there is the obvious difference that our negative results are unconditional (and not predicated on a complexity theoretic assumption like the unique games conjecture or \textsf{P}$\ne$\textsf{NP}). But more significantly our characterization is a bit more ``explicit'' than those of \cite{raghavendra2008optimal} and \cite{KhotTW}. 
In particular the former only shows decidability of the problem which takes $\epsilon$ as an input (in addition to $\gamma,\beta$ and $f$) and distinguishes  $(\gamma,\beta)$-approximable problems from $(\gamma-\epsilon,\beta+\epsilon)$-inapproximable problems. The running time of their decision procedure grows with $1/\epsilon$. In contrast our distinguishability is sharper and separates $(\gamma,\beta)$-approximability from ``$\forall \epsilon > 0$, $(\gamma-\epsilon,\beta+\epsilon)$-inapproximability'' --- so our algorithm does not require $\epsilon$ as an input --- it merely takes $\gamma,\beta$ and $f$ as input.
Indeed this difference is key to the understanding of approximation resistance. Due to the stronger form of our main theorem (\cref{thm:informal}), our characterization of streaming-approximation-resistance is explicit (decidable in \textsf{PSPACE}), whereas a decidable characterization of approximation-resistance in the polynomial time setting seems to be still open. 

Our characterizations also seem to differ from the previous versions in terms of the features being exploited to distinguish the two classes. This leads to some strange gaps in our knowledge. For instance, it would be natural to suspect that (conditional) inapproximability in the polynomial time setting should also lead to (unconditional) inapproximability in the streaming setting. But we do not have a formal theorem proving this. (Of course, if this were false, it would be a breakthrough result giving a quasi-polynomial time (even polylog space) algorithm for the unique games!)

\subsection{Overview of our analysis}
At the heart of our characterization is a family of linear sketching algorithms for $\maxF$. We will describe this family soon, but the main idea of our proof is that if no algorithm in this family solves $(\gamma,\beta)$-$\maxF$, then we can extract a pair of instances, roughly a family of $\gamma$-satisfiable ``yes'' instances and a family of at most $\beta$-satisfiable ``no'' instances, that certify this inability. We then show how this pair of instances can be exploited as gadgets in a negative result. Up to this part, our approach resembles that in \cite{raghavendra2008optimal} (though of course, all the steps are quite different). The main difference is that we are able to use the structure of the algorithm and the lower bound construction to show that we can afford to consider only instances on $k$ variables. (This step involves a non-trivial choice of definitions that we elaborate on shortly.) 
This bound on the number of variables allows us to get a ``decidable'' separation between approximable and inapproximable problems. Specifically, we show that the distinction between the approximable setting and the inapproximable one can be expressed by a quantified formula over the reals with a constant number of quantifiers over $2^k$ variables and equations --- a problem that is known to be solvable in \textsf{PSPACE}. We give more details below. To simplify the discussion we consider a singleton function family $\cF = \{f\}$. Extending to multiple functions is not much harder (though as stressed by the \textsf{Max-3SAT} example, this is not trivial either). We start by giving some intuition into our framework before actually describing the framework. We remark that while this intuition may be helpful, it is not necessary for any of our proofs.

\paragraph{Intuition.} 
Our starting point is the belief that streaming algorithms working with polylogarithmic space can essentially extract the ``bias profile'' of an instance, while algorithms with much more (specifically $o(\sqrt{n})$) space can not do much more. Here, by bias profile of an instance $\Phi$ on $n$ variables we mean the $n \times k$ matrix $B = B(\Phi)$ with $B_{i,j}$ representing the fraction of constraints of $\Phi$ that have $x_i$ as the $j$th variable. If our belief were to be true then the only obstacle to deciding $\gbmaxf$ in $o(\sqrt{n})$ space would be two instances $\Phi_Y$ and $\Phi_N$ on the same set of variables with $\val(\Phi_Y) \geq \gamma$ and $\val(\Phi_N) \leq \beta$ while the instances have the same bias profile, i.e. $B(\Phi_Y) = B(\Phi_N)$. 

To convert our belief into a proof of \cref{thm:informal}, we need to do three things: (1) Given $\gamma,\beta$ and $f$, show that the existence of such a pair of instances $\Phi_Y$ and $\Phi_N$ can be decided (in finite time); (2) Show that if no pair of such instances exist then $\gbmaxf$ can be decided by a polylogarithmic space sketching algorithm; and (3) If such a pair of instances exist then no $o(\sqrt{n})$ space sketching algorithm can solve $\gbmaxf$.

While step (3) ends up taking most of the technical work in this paper, it is also perhaps the most believable. Roughly hard instances of arbitrary length can be extracted from $\Phi_Y$ and $\Phi_N$ by doing ``random lifts'', i.e., creating many copies of each variable in $\Phi_Y$ and applying constraints randomly among these copies according to $\Phi_Y$ or $\Phi_N$ roughly preserves the values; and the fact that the bias profiles match can be converted into a hardness result for sketching algorithms using communication complexity based arguments. We expand on this more below.

The less believable steps (in our estimate) are steps (1) and (2) and it turns out that understanding the challenge in (1) better leads to a solution to both steps. The challenge behind (1) is of course the fact that a priori the number of variables in $\Phi_Y$ or $\Phi_N$ can not be bounded and so there is no finite upper bound on the time it would take to decide their existence. The key to resolving this is the fact (that we will argue below) that the information contained in $\Phi_Y$ and $\Phi_N$ can be compressed into smaller instances on $kq$ variables.

To establish this, let us suppose (without loss of generality) that $\Phi_Y$ and $\Phi_N$ are instances on $n \times q$ variables $\{X_{i,\sigma}\}_{i \in [n], \sigma \in [q]}$. Further suppose the assignment that establishes $\val(\Phi_Y) \geq \gamma$ is the assignment $a_{i,\sigma} = \sigma$. For permutations $\pi_1,\ldots,\pi_q:[n]\to [n]$, let $\Phi_Y^{\pi_1,\ldots,\pi_q}$ be a copy of $\Phi_Y$ with variables renamed to $\{X_{\pi_\sigma(i),\sigma}\}$. Similarly define $\Phi_N^{\pi_1,\ldots,\pi_q}$. Note that renaming the variables preserves the values and the bias profiles still match, and furthermore the assignment that yields a value of $\gamma$ to $\Phi_Y^{\pi_1,\ldots,\pi_q}$ is still $a_{i,\sigma} = \sigma$. Thus if we now consider the instances $\widetilde{\Phi_Y}$ obtained by concatenating all the constraints of $\Phi_Y^{\pi_1,\ldots,\pi_q}$ over all choices of 
${\pi_1,\ldots,\pi_q}$, and similarly define $\widetilde{\Phi_N}$, then the resulting instances still have matching bias profiles and they still satisfy  $\val(\widetilde{\Phi_Y}) \geq \gamma$ and $\val(\widetilde{\Phi_Y}) \leq \beta$. The gain with all these transformations is that $\widetilde{\Phi_Y}$ and $\widetilde{\Phi_N}$ are very symmetric instances with only $q$ equivalence classes of variables (as opposed to $n$ general variables). And a random constraint just picks a uniform variable from an equivalence class, conditioned on picking a variable from that class, in any given position. (Recall that by our assumption, every constraint is applied on $k$ \emph{distinct} variables.) Thus the instances $\widetilde{\Phi_Y}$ and $\widetilde{\Phi_N}$ are effectively given by a distribution supported on $[q]^k$ where the probability of $(\sigma_1,\ldots,\sigma_k)$ measures the frequency of constraints on $k$-tuples of variables of the form $(X_{*,\sigma_1},\ldots,X_{*,\sigma_k})$. 

Thus the instances revealing the gap between $\gamma$ and $\beta$ are finitely specified (or at least are distributions over a finite space), but it is still unclear how to search for (specifications of) such instances of value at least $\gamma$ or at most $\beta$. To address this challenge one may try to reduce 
the entire instance $\widetilde{\Phi_Y}$ into an ``equivalent'' instance on just $q$ variables, (by replacing all variables $X_{i,\sigma}$ for $i \in [n]$ with a single variable $Z_i$) but this may result in constraints where all variables are not distinct. To exclude this possibility we replace the collection of variables $X_{i,\sigma}$ with $k$ variables $Z_{\ell,\sigma}$ for $\ell \in [k]$; and now compress $\widetilde{\Phi_Y}$ by replacing all occurrences of $X_{i,\sigma}$ as the $\ell$th variable in a constraint, with $Z_{\ell,\sigma}$. This leads to a compressed instance $\Phi'_Y$ on just $kq$ variables. We can do a similar reduction with $\widetilde{\Phi_N}$ to get an instance $\Phi'_N$. These resulting instances also have matching bias profiles. The reduction in the variables ensures $\val(\Phi'_Y) \geq \gamma$ since the assignment $Z_{\ell,\sigma} = \sigma$ still satisfies a $\gamma$ fraction of the constraints. However, it is no longer true that $\val(\Phi'_N)\leq \beta$. This is so since the assignment to a variable $Y_{i,\sigma}$ might depend on $i$ which was not a possibility considered when bounding $\val(\widetilde{\Phi_N})$. What we would like at this stage is a succinct way to capture the fact that if we try to reverse engineer $\widetilde{\Phi_N}$ from $\Phi'_N$ then we would have $\val(\widetilde{\Phi_N})\leq \beta$. It turns out one succinct way to capture this is to consider only those distributions on assignments to the variables $Z_{\ell,\sigma}$ that are independent across variables and furthermore the distributions of $Z_{\ell,\sigma}$ and $Z_{\ell',\sigma}$ are identical. If we require that $\Phi'_N$ has value at most $\beta$ in expectation over all such distributions of assignments to its variables, then we effectively capture the constraint $\val(\widetilde{\Phi_N})\leq \beta$. 


Thus the search for instances $\Phi_Y$ and $\Phi_N$ can be reduced to a search for instances $\Phi'_Y$ and $\Phi'_N$ on just $kq$ variables whose bias profiles must match and whose values satisfy some constraints. Since the marginals of distributions supported on $[q]^k$ are captured by vectors in $[0,1]^{kq} \subseteq \R^{kq}$
we get that the space of marginals of all yes instances (of the special type we care about) is given by a subset of points in $\R^{kq}$, which we denote $\kgyf$. Similarly, the space of the marginals of the no instances is also a subset of $\R^{kq}$ denoted $\kbnf$. It turns out these sets are bounded, closed, and convex and actually described by some polynomial conditions. Thus solving step (1) reduces to the task of determining if $\kgyf$ and $\kbnf$ intersect. And when they do not intersect, the separating hyperplane gives us a clue on how to solve the problem from step (2), i.e., how to solve $\gbmaxF$ with polylogarithmic space.

To show that this framework works, we need to explain what our algorithms are, why they lead to these special instances when they fail, and how to use the failure of the algorithms (or equivalently the intersection of $\kgyf$ and $\kbnf$) to get the hardness of $\gbmaxF$. We attempt to explain this below.

\paragraph{Bias-based algorithms.} 
The class of algorithms we use are what we call ``bias-based algorithms,'' which extend algorithms used for \textsf{Max-DICUT} and other problems in \cite{GVV17,CGV20}. Roughly, these algorithms work by inspecting constraints one at a time and (linearly) updating the ``preference/bias'' of variables involved in the constraint for a given assignment. This update depends on the location of the variable within the constraint (and if there are multiple functions in the family, also on the function itself). Thus implicitly these algorithms maintain an $n$-dimensional bias vector and at the end use some property of this vector to estimate a lower bound on the value of the instance. If this property is computable efficiently in the turnstile streaming model, then this leads to a space-efficient streaming algorithm. 

The key questions for us are: (1) How to update the bias? and (2) What property of the vector yields a lower bound. When dealing with specific functions as in previous papers, there are some natural candidates for bias and the most natural one turns out to be both useful and computable efficiently using $\ell_1$ norm estimators. For the property, one has to devise a ``rounding scheme'' that takes the bias vector and uses it to create an assignment that achieves a large value (or value related to the property being estimated). 

In our case, obviously ``inspection'' of natural candidates will not work for item (1) --- we have infinitely many problems to inspect. But it turns out that the convex set framework, somewhat surprisingly, completely solves both parts (1) and (2) for us. If $\kgyf$ and $\kbnf$ do not intersect then there is a linear separator in $\R^{kq}$ separating the two sets and the coefficients of this separator are interpretable as giving $kq$ ``biases'' --- for $i\in[k]$ and $\sigma\in [k]$ the $(i,\sigma)$-th coefficient can be viewed as the bias/preference of the $i$-th variable in a constraint for taking the  assignment $\sigma \in [q]$. This gives us an $n \times q$ bias matrix at the end that captures all the biases of variables from the whole instance. Turning to (2), a natural property to consider at this stage is the one-infinity norm of this matrix (i.e., the $\ell_1$ norm of the $n$-dimensional vector whose coordinates are the $\ell_\infty$ norms of the rows of the bias matrix). Informally, this corresponds to each variable acting independently according to its bias. It turns out this norm is one of many that is known to be computable with small space in the turnstile streaming setting and in particular we use a result of Andoni, Krauthgamer, and Onak~\cite{AKO} to compute this. Finally, we need a relationship between this property and a lower bound on the value, and once again the fact that the bias came from a separating hyperplane (and the exact definition of the sets in the convex set framework) allows us to distinguish instances with value at least $\gamma$ from instances of value at most $\beta$. (Note that these constants are already baked into our sets and hence the separating hyperplane.) We remark that we do not give an explicit rounding procedure for our approximation algorithm, though one can probably be extracted from the definitions of the convex sets and analyses of the correctness of our algorithms.

\paragraph{Lower bounds.}
Finally, we turn to the lower bounds. Once again we restrict our overview to the setting of $|\cF| = 1$ for simplicity. Both our lower bounds for sketching algorithms and for general streaming algorithms have a common starting point. Recall we are given that there are two distributions $\cD_Y$ and $\cD_N$ on constraints that have the same one-wise marginals and these can be viewed as distributions on $[q]^k$.
 
For every pair of such distributions $\cD_Y$ and $\cD_N$ in $[q]^k$ we define a two player communication problem we call $(\cD_Y,\cD_N)$-signal detection (SD). (So in effect these are infinitely many different communication problems, roughly corresponding to the infinitely many different $\maxF$ problems we wish to analyze.) We show that if $\cD_Y$ and $\cD_N$ have the same marginals, then the communication problem requires $\Omega(\sqrt{n})$ communication. We give further details below, but now explain the path from this communication lower bound to the streaming lower bounds.
To get these lower bounds, we convert our SD lower bound into lower bounds on some $T$-players games, for all large constants $T$. Instances of the $T$-player games immediately correspond to instances of $\maxF$ and furthermore the properties of the sets $\kgyf$ and $\kbnf$ translate into the value of these $\maxF$ instances. 

Turning to the $T$-player games: 
In the lower bound for sketching algorithms, we first convert the SD lower bound into a lower bound on a $T$-player {\em simultaneous} communication game.
This conversion is relatively standard in the streaming literature \cite{Kapralov13,Konrad15,HRVZ15,AKLY16}: reduce the two-player communication game to the $T$-player communication game by letting Bob play the role of one of the players and Alice play the role of the remaining $T-1$ players.
By turning a sketching algorithm into a protocol for the communication game we can get a space $\sqrt{n}$ lower bound for every $(\gamma,\beta)$-$\maxF$ against any sketching algorithms, whenever the corresponding $K^Y$ and $K^N$ intersect. (See~\cref{thm:main-negative-dynamic}.) For the hardness result in the streaming setting, the lower bound on the simultaneous communication problem no longer suffices. So here we craft our own reduction to a $T$-player {\em one-way} communication problem which reduces in turn to $(\gamma,\beta)$-$\maxF$ in the streaming setting. (This step follows the same path as \cite{KKS,CGV20}.) Unfortunately, this step works only in some restricted cases (for instance if $\cD_N$ is the uniform distribution on $[q]^k$) and this yields our lower bound (\cref{thm:one-wise}) in the streaming setting. 

We now turn to our family of communication problems (SD), which is a distributional one-way communication problem. In the $(\cD_Y,\cD_N)$-\textsf{SD} problem with length parameter $n$, Alice gets a random string $\vecx^* \in [q]^n$ and Bob gets a hypermatching $\vecJ = (\vecj(1),\ldots,\vecj(m))$ with $m=\alpha n$ edges (where $\alpha > 0$ is a constant of our choice independent of $n$). In other words $\vecj(i)$ is a sequence of $k$ distinct elements of $[n]$ and furthermore $\vecj(i)$ and $\vecj(i')$ are disjoint for every $i\ne i' \in [m]$. In addition, Bob also gets $m$ bits $\vecz = (z(1),\ldots,z(m))$, where 
$z(i)$ is obtained by sampling $\vecb(i) \sim \cD_Y$ in the \yes\ case (and $\vecb(i) \sim \cD_N$ in the \no\ case) independently for $i \in [m]$ and letting $z(i) = 1$ iff $\vecx^*|_{\vecj(i)} = \vecb(i)$. The goal of the communication problem is for Alice to send a message to Bob that allows Bob to guess whether this is a \yes\ instance or a \no\ instance. The minimum length (over all protocols solving SD) of Alice's message is the complexity of the $(\cD_Y,\cD_N)$-\textsf{SD}. It is straightforward from the definition to get a $O_{\cD_Y,\cD_N,\alpha}(1)$-bit communication protocol achieving constant advantage if $\cD_Y$ and $\cD_N$ do not have the same marginals. Our lower bound shows that whenever the marginals match, the communication is at least $\Omega(\sqrt{n})$. (It is again straightforward to show distributions with matching marginals where $O(\sqrt{n})$ bits of communication suffice to distinguish the two cases.)

Before giving some details on our lower bound proof of the SD problem, we briefly give some context to the problem itself.
We note that our communication game is different from those in previous works: Specifically the problem studied in \cite{GKKRW,KKS} is called the \textit{Boolean Hidden Matching (BHM)} problem from~\cite{GKKRW} and the works \cite{KKSV17,KK19} study a variant called the \textit{Implicit Hidden Partition} problem. While these problems are similar, they are less expressive than our formulation, and specifically do not seem to capture all $\maxF$ problems. We note that the BHM problem is essentially well suited only for the setting $k=q=2$. In particular, the definition and analysis of BHM relies on the Fourier analysis over $\F_q$. Increasing $k$ leads to several possible extensions that seem more naturally suited to CSPs on literals rather than variables. And increasing $q$ leads to further complications since we do not have a natural field to work with. Thus the choice of SD is made carefully to allow both expressibility (we need to capture all $\maxF$s) and the ability to prove lower bounds. 

Turning to our lower bound, it comes in two major steps. In the first step we resort to a different communication problem that we call the ``Randomized Mask Detection Problem with advice'' (\textsf{Advice-RMD}). In this problem, defined only for $q=2$, Alice and Bob are given more information than in SD. Specifically Alice is given as ``advice'' a partition of $[n]$ into $k$ parts with the promise that the $\ell$-th variable in every constraint is from the $\ell$-th part for every $\ell \in [k]$. And Bob is given the vectors $(\vecz(1),\ldots,\vecz(m))$ where $\vecz(i) = \vecx^*|_{\vecj(i)} \oplus \vecb(i)$ for $i \in [m]$. This problem is closest both in definition and analyzability to the previous problems. Indeed we are able to extend previous Fourier-analytic lower bounds, in the special case where the marginals of $\cD_Y$ and $\cD_N$ over $\{-1,1\}$ are {\em uniform}, to give an $\Omega(\sqrt{n})$ lower bound on the communication complexity of this problem. (See ~\cref{thm:rmd}.) This immediately yields a hardness of the SD problem when $\cD_Y$ and $\cD_N$ are distributions over $\{-1,1\}^k$ with uniform marginals, but we need more.

To extend the lower bound to all $q$ and to non-uniform marginals, we use more combinatorial methods. Specifically we show that we can move $\cD_Y$ to $\cD_N$ in a series of steps $\cD_Y = \cD_1,\ldots,\cD_L = \cD_N$ where for every $i$, the difference between $\cD_i$ and $\cD_{i+1}$ is ``captured'' (in a sense we do not elaborate here) by two distributions with uniform marginals over $\{a,b\}^k$ for some $a,b \in [q]$. We refer to each of these $L$ steps as a ``polarization step''. Showing that $L$, the number of polarization steps, is finite leads to an interesting problem we solve in \cref{ssec:finite}. (The bound depends on $q$ and $k$, but not $\cD_Y,\cD_N, \alpha$ or $n$. We remark that any dependence on the first three would have been fine for our application.) Finally we show that the lower bound on the \textsf{Advice-RMD} mentioned above, in the Boolean uniform marginal setting, suffices to show that the $(\cD_i,\cD_{i+1})$-\textsf{SD} problem also requires $\Omega(\sqrt{n})$ communication. (See \cref{thm:kpart,thm:polarization step hardness}.) By a triangle inequality it follows that $(\cD_Y,\cD_N)$-\textsf{SD} requires $\Omega(\sqrt{n})$ communication. (See \cref{thm:communication lb matching moments}). 

\subsection{Subsequent results}\label{ssec:subseq}


\newcommand{\mon}{\mathsf{monarchy}}

Subsequent to the first announcement of this work several followup results have extended and strengthened the results of this paper. We report on some of these below.

\paragraph{Explicit Families of Easy and Hard Problems.} One of the main drawbacks of our result in \cref{thm:informal} is that the decision criterion is not completely explicit. This is of course natural given the richness of the class of problems, but it is still natural to ask are there some clean families of problems that can be shown to be non-trivially approximable, or not, by further analyzing the tractability condition. Two subsequent works have addressed this question for infinite classes of problems and we report on these below. 

One class of works by the authors with Shahrasbi~\cite{CGSSV-mon} explores the ``monarchy'' and ``weak monarchy'' predicates. The monarchy predicate is the function $f_{\mon}: \{-1,1\}^k \to \{0,1\}$ given by $f(x_1,\ldots,x_k) = \sign((k-2)x_1 + \sum_{i=2}^k x_i)$. In other words $f_{\mon}(\vecx) = 1$ if $x_1 = 1$ and at least one other $x_i$ is $1$, or if $x_2 = \cdots = x_k = 1$. The monarchy family $\cF_{\mon}$ is given by applying the monarchy predicate to literals, i.e., $\cF_{\mon} = \{f_{\mon}^{\vecb} | \vecb \in \{-1,1\}^k\}$ where $f_{\mon}^{\vecb}(\vecx) = f_{\mon}(\vecx \odot \vecb)$. The  monarchy CSP ($\maxF_{\mon}$) is known to be approximable in the polynomial time setting for every $k$~\cite{Potechin}. In contrast, their work \cite{CGSSV-mon} shows that for $k\geq 5$, the monarchy CSP  is approximation-resistant in the sketching setting. This is of particular interest since this is a family that is not one-wise independent but remains approximation-resistant in the sketching setting. The approximation resistance of this class for general streaming algorithms remains open. \cite{CGSSV-mon} also explores weak monarchy CSPs, i.e., CSPs on functions of the form $f_{k,j}(\vecx) = \sign(j x_1 + \sum_{i=2}^k x_i)$ applied to literals. They show that for every $j$ for all sufficiently large $k$ the weak monarchy CSP based on $f_{k,j}(\vecx)$ is non-trivially approximable in the sketching setting.

Another work deriving explicit bounds for infinite families is due to Boyland, Hwang, Prasad, Singer, and Velusamy \cite{BHPSV21}. They derive the exact form of the optimal sketching approximation ratios for several symmetric Boolean CSPs including $\maxkand$ and $\Th_k^{k-1}$ (the ``weight-at-least-$(k-1)$'' threshold function on $k$ variables). In both cases they show that there are non-trivial approximation algorithms thus establishing infinitely many problems for which the exact approximation ratio can be determined using (and further analyzing) our framework. (As an example they show that the approximation ratio for $\maxkand$ is exactly $2^{-(k-1)} (1-k^{-2})^{(k-1)/2}$ for odd $k \geq 3$ for sketching algorithms.) 
Their work further analyzes our streaming lower bound in \cref{thm:main-negative} and shows that for the threshold function $\Th^3_4$, our streaming and sketching lower bounds match. (This is analogous to our result for $\mdcut$ in \cref{Example:Max-DICUT}.)

\paragraph{$o(n)$-space algorithms.} In a work of the authors with Velingker~\cite{CGS+21}, the space lower bound in \cref{thm:one-wise} is improved to $\Omega(n)$ for a subclass of function families that support one-wise independence. In particular, they show that the subclass they consider is approximation resistant with respect to $o(n)$-space streaming algorithms. We do not describe the exact subclass here but mention that it suffices for them to get an ``approximate'' classification of all approximation problems, Namely for every given $\gamma$, $\beta$ and $\cF$ over a $q$-ary alphabet they show that either $\gbmaxF$ is trivial or $(\gamma/q,\beta)$-$\maxF$ requires $\Omega(n)$ space to solve. Their work suggests some inherent barriers in extending the full classification of the problems considered in the current paper to $o(n)$-space algorithms.
This was later confirmed in a work of Saxena, Singer, Sudan, and Velusamy \cite{SSSV23b} where they give an $\tilde{O}(\sqrt{n})$ space algorithm for $\mdcut$ that beats the best $o(\sqrt{n})$ space algorithm. Singer \cite{Singer23} partially extends this result to obtain an $O(n^{1-1/k})$ space algorithm for $\maxkand$ that beats the optimal $o(\sqrt{n})$ space algorithm on ``bounded-degree'' instances.

\paragraph{Random-ordering streaming setting} While Kapralov, Khanna, and Sudan~\cite{KKS} show that $\mcut$ is inapproximable by $o(\sqrt{n})$ space streaming algorithms even in the random-ordering setting, Saxena, Singer, Sudan, and Velusamy~\cite{SSSV23a} give an $O(\log n)$ space streaming algorithm in this setting that beats the optimal $o(\sqrt{n})$ space algorithm for $\mdcut$ in the adversarial-ordering setting. Singer \cite{Singer23} extends this result to obtain $O(\log n)$ space random-order streaming algorithms that beat the best $o(\sqrt{n})$ space adversarial-order algorithms for $\maxkand$, for all $k$!

\paragraph{Multi-pass streaming setting}
The random-order streaming algorithms in \cite{SSSV23a, Singer23} can be trivially extended to obtain $O(\log n)$ space two-pass adversarial-order streaming algorithms with the same approximation ratio.  A recent result due to Kol, Saxena, Paramonov, and Yu \cite{Mpass23a} gives a complete characterization for the exact computability of every Boolean $\maxf$ in the  multi-pass streaming setting and subsumes our \cref{thm:informal2} for this family. In particular, for every Boolean predicate $f$, they give an $\tilde{O}(n^{\deg(f)})$ space single-pass streaming algorithm that solves $\maxf$ exactly, where $\deg(f)$ is the degree of $f$ when viewed as multilinear polynomial, and show that any \emph{constant}-pass streaming algorithm requires at least $\Omega(n^{\deg(f)})$ space.

\paragraph{Variations of CSPs.} It turns out that our work on CSPs also is helpful in analyzing some variations of CSPs. In particular Singer, Sudan, and Velusamy~\cite{SSV21} consider the space of ``ordering CSPs'' where the challenge is to find an ordering of $n$ variables that satisfy some specified ordering constraints. An example is the Maximum Acyclic Subraph (MAS) problem where the goal is to find an ordering of $n$ variables $x_1,\ldots,x_n$ that, given many constraints of the form $x_i < x_j$, satisfies as many constraints as possible. Prior to the work of \cite{SSV21} no problem (including MAS) was tightly analyzed. \cite{SSV21} show that no ordering CSP has a non-trivial streaming algorithm with $o(\sqrt{n})$ space. Their work crucially relies on the framework from this paper and uses the approximation resistance of some CSPs considered in this paper. (See \cref{ssec:examples} for further details.) Since the problems needed in their work fall within the subclass of problems considered in \cite{CGS+21}, their streaming lower bound actually improves to $\Omega(n)$-space.

\subsection{Structure of rest of the paper}
\cref{sec:preliminaries} contains some of the preliminary background used in the rest of the paper. In \cref{sec:results}, we describe our results in detail. In particular, we build our convex set framework and give  an explicit criterion to distinguish the easy and hard $\maxF$ problems. We also describe sufficient conditions for the hardness of some streaming problems in the streaming setting. 
In \cref{sec:alg}, we describe and analyze our algorithm that yields our easiness result. In \cref{sec:lower-bound}, we define the ``Signal Detection'' problem and show how the communication complexity of this problem leads to the streaming space lower bounds claimed in \cref{sec:results}. 
In \cref{sec:kpart}, we introduce and analyze the \textsf{Advice-RMD} problem. 
In \cref{sec:polar} we prove our general lower bound for SD assuming that a single polarization step is hard. In \cref{sec:spl} we complete this remaining step by using the \textsf{Advice-RMD} lower bound to show hardness of a single polarization step, thus concluding our main lower bound.
Finally, in \cref{sec:exact} we give the dichotomy for the exact computability of $\maxF$. 

\section{Preliminaries}\label{sec:preliminaries}

In this section we introduce notations, definitions and some standard tools that will be used in the rest of this paper. Specifically we define constraint satisfactions problems and some promise problems related to their approximation (\cref{ssec:def-apx-csp}). Then we formally describe the streaming and sketching models of computation along with some variants and background material (\cref{ssec:def-stream}). In \cref{ssec:def-apx-equiv} we explain the folklore relationship between the promise problems defined in \cref{ssec:def-apx-csp} with the standard single parameter version of approximations, in the context of streaming algorithms. \cref{ssec:def-prob} has some basic notions from probability and some tools we will use. \cref{ssec:def-fourier} recalls notions from Fourier analysis and mentions the tools used from this area. Finally, \cref{ssec:def-reals} defines notions and results from the quantified theory of reals. We start with some notation.

We let $\N$ denote the set of positive integers. 
We let $[n]$ denote the set $\{1,\ldots,n\}$. 
For a finite set $\Omega$, let $\Delta(\Omega)$ denote the space of all probability distributions over $\Omega$, i.e., 
$$\Delta(\Omega) = \left\{\cD:\Omega \to \R^{\geq 0}\ |\ \sum_{\omega \in \Omega} \cD(\omega) = 1\right\}.$$
We view $\Delta(\Omega)$ as being contained in $\R^{|\Omega|}$. We use $X\sim\cD$ to denote a random variable drawn from the distribution $\cD$.
By default, a Boolean variable in this paper takes value in $\{-1,1\}$.
For every $p\in[0,1]$, $\textsf{Bern}(p)$ denotes the Bernoulli distribution that takes value $1$ with probability~$p$ and takes value $-1$ with probability $1-p$.

We will follow the convention that $n$ denotes the number of variables in CSPs, 
 $m$ denotes the number of constraints, and $k$ denotes the arity of the CSP.  

For variables of a vector form, we write them in boldface, \textit{e.g.,} $\vecx\in[q]^n$, and its $i$-th entry is written without boldface, \textit{e.g.,} $x_i$. For variable being a vector of vectors, we write it, for example, as $\vecb=(\vecb(1),\vecb(2),\dots,\vecb(m))$ where $\vecb(i)\in[q]^k$. The $j$-th entry of the $i$-th vector of $\vecb$ is then written as $\vecb(i)_j$. Let $\vecx$ and $\vecy$ be two vectors of the same length, $\vecx\odot\vecy$ denotes the entry-wise product of them.

\subsection{Approximate Constraint Satisfaction}\label{ssec:def-apx-csp}

$\maxF$ is specified by a family of constraints $\cF$, where each constraint function $f\in \cF$ is such that $f:[q]^k\rightarrow \{0,1\}$, for a fixed positive integer $k$. Given $n$ variables $x_1,x_2,\dots,x_n$, an application of the constraint function $f$ to these variables, which we term simply a {\em constraint}, is given by a $k$-tuple $\vecj = (j_1,\ldots,j_k) \in [n]^k$  where the $j_i$'s are distinct and  represent the application of the constraint function $f$ to the variables $x_{j_1},\ldots,x_{j_k}$. We use $\C_{\cF,n}$ to denote the set of all constraints of $\maxF$ on $n$ variables. (Note that $\C_{\cF,n}$ is a finite set.) 
Specifically an assignment $\vecb\in [q]^n$ satisfies a constraint given by $(f,\vecj)$ if $f(b_{j_1},\ldots,b_{j_k})=1$.

An instance $\Psi$ of $\maxF$ consists of $m$ constraints $C_1,\ldots,C_m$ with non-negative weights $w_1,\ldots,w_m$ where $C_i=(f_i,\vecj(i)) \in \C_{\cF,n}$ and $w_i \in \R$ for each $i\in[m]$. For an assignment $\vecb \in [q]^{n}$, the value $\val_\Psi(\vecb)$ of $\vecb$ on $\Psi$ is the fraction of weight of constraints satisfied by $\vecb$, i.e., $\val_\Psi(\vecb)=\tfrac{1}{W}\sum_{i\in[m]}w_i \cdot f_i(\vecb|_{\vecj(i)})$, where $W = \sum_{i=1}^m w_i$. The optimal value of $\Psi$ is defined as $\val_\Psi = \max_{\vecb \in [q]^{n}}\{\val_\Psi(\vecb)\}$.  The approximation version of $\maxF$ is defined as follows.

Throughout this paper we will only consider the case of $\maxF$ instances with integer weights bounded by a polynomial in $n$.

\begin{definition}[$(\gamma,\beta)$-$\maxF$]
Let $\cF$ be a constraint family and $0\leq\beta<\gamma\leq1$. For each $m\in\N$, let $\Gamma_m=\{\Psi=(C_1,\dots,C_m;w_1,\ldots,w_m)\, |\, \val_\Psi\geq\gamma\}$ and $B_m=\{\Psi=(C_1,\dots,C_m;w_1,\ldots,w_m)\, |\, \val_\Psi\leq\beta\}$.

The task of $(\gamma,\beta)$-$\maxF$ is to distinguish between instances from $\Gamma=\cup_{m\leq\textsf{poly}(n)}\Gamma_m$ and instances from $B=\cup_{m\leq\textsf{poly}(n)}B_m$. Specifically we desire algorithms that output $1$ w.p. at least $2/3$ on inputs from $\Gamma$ and output $1$ w.p. at most $1/3$ on inputs from $B$. 
\end{definition}

\subsection{Streaming and Sketching Algorithms}\label{ssec:def-stream}

\newcommand{\wtS}{\widetilde{S}}

We now define streaming and sketching algorithms in the context of $\maxF$. Note that the input to both algorithms are sequences of weighted constraints. Rather than explicitly including the weight we will simply allow the sequence to repeat constraints (not necessarily successively). The implied weight of a constraint will thus be the number of times it is repeated. (Note that we only consider integer polynomially bounded weights. Thus this representation only blows up the input by a polynomial factor.) A stream is thus an element of $(\C_{\cF,n})^*$ and we use $\lambda$ to denote the empty stream. 

\begin{definition}[Streaming algorithm]\label{def:streaming alg}
A deterministic space $s$ streaming algorithm $\ALG$ for $\maxF$ on $n$ variables is given by a (state-evolution) function $S: \{0,1\}^s \times \C_{\cF,n} \to \{0,1\}^s$ and a (output) function $v:\{0,1\}^s \to [0,1]$.
Let  $\wtS: (\C_{\cF,n})^*\to \{0,1\}^s$ given by $\wtS(\lambda) = 0^s$ and
$\wtS(\sigma_1,\ldots,\sigma_m) = S(\wtS(\sigma_1,\ldots,\sigma_{m-1}),\sigma_m)$ denote the iterated state-evolution map.
Then the output of $\ALG$ on input $\sigma = (\sigma_1,\ldots,\sigma_m)$ is $v(\wtS(\sigma))$.

In a {\em uniform randomized} space $s$ streaming algorithm the evolution map is given by $S:\{0,1\}^s \times \C_{\cF,n} \times \{0,1\}^r \to \{0,1\}^s$ for some $r \leq s$ and its iterate evolution map is a random variable given by $\wtS(\sigma_1,\ldots,\sigma_m) = S(\wtS(\sigma_1,\ldots,\sigma_{m-1}),\sigma_m,R_m)$ where $R_m \sim \textsf{Unif}(\{0,1\}^r)$ is independent of $\sigma_1,\ldots,\sigma_m$ and $R_1,\ldots,R_{m-1}$. 

A {\em non-uniform randomized} space $s$ streaming algorithm is simply a distribution on deterministic space $s$ streaming algorithms. 
\end{definition}

We note that non-uniform randomized algorithms can simulate uniform ones but may be much stronger since they allow algorithms to ``remember'' all previous random coins without being charged for the memory. All our upper bounds are in the uniform randomized model. Our lower bounds are in the non-uniform randomized model (and use this extra power in the reductions). 

Sketching algorithms are a special class of streaming algorithms that have been widely used in both upper bounds and lower bounds. For the definition of sketching algorithms below, we adopt Definition~5.21 in~\cite{C20}.

\begin{definition}[Sketching algorithms]\label{def:sketching alg}
A deterministic space $s$ streaming algorithm $\ALG=(S,v)$ is a sketching algorithm if there exists a compression function $\COMP:(\C_{\cF,n})^* \to \{0,1\}^s$ and a combination function $\COMB:\{0,1\}^s\times\{0,1\}^s \to \{0,1\}^s$ such that the following hold:
\begin{itemize} 
\item $S(z,C) = \COMB(z,\COMP(C))$ for every $z \in \{0,1\}^s$ and $C \in \C_{\cF,n}$.
\item For every pair of streams $\sigma,\tau\in (\C_{\cF,n})^*$, we have
\[
\COMB(\COMP(\sigma),\COMP(\tau))=\COMP(\sigma\circ\tau)
\]
where $\sigma\circ\tau$ represents the concatenation of the streams $\sigma$ and $\tau$. 

A uniform randomized sketching algorithm is similarly defined with $\COMB:\{0,1\}^s\times\{0,1\}^s\times\{0,1\}^r \to \{0,1\}^s$ and $S(z,C,R) = \COMB(z,\COMP(C),R)$ for every $z,C,R$, where $r\leq s$.  A randomized algorithm $\ALG$ is a non-uniform  randomized sketching algorithm if it is a distribution over deterministic sketching algorithms.
\end{itemize}

\end{definition}

    We remark that there can be several variants to the streaming problem above involving the possibility of weighted constraints, deletion of constraints and the length of the input stream. 
    \begin{enumerate}
        \item Dynamic streams: In this setting constraints may be inserted, even multiple times, and later deleted. In this setting algorithms are required to be correct on the final instance, under the promise that constraints were deleted fewer times than they were inserted at all intermediate stages of the streaming process. The input stream can be unboundedly large in this setting even while maintaining polynomially bounded integer weights (e.g., by inserting and deleting the same constraint an arbitrary number of times). Thus, algorithms may have restrictions on the length of input streams or have complexity growing with the length of the stream. 

        All our lower bounds work in the insertion only setting. Our upper bounds work on dynamic streams provided they have length polynomial in $n$.
        
        \item Weighted instances: Variations of $\maxF$ allow constraints to have non-negative real weights. We do not explicitly consider this setting in this paper, but standard techniques (involving rounding weights to nearby rationals) allow algorithms for polynomially bounded integer weights to be extended to apply to this setting also. 
        \item Linear Sketching: An instance $\Psi$ of $\maxF$ can be viewed as a vector in $\R^{\C_{\cF,n}}$ with the $C$th coordinate representing the weight of the constraint $C$ in $\Psi$. A linear sketching algorithm is one whose state is a linear function of this representation of the instance. Note that in this representation, the stream can be viewed as a sequence of linear updates. Thus if the state is a linear function, the updates to the state can be computed knowing only the previous state and the update to $\Psi$ thus leading to a natural streaming algorithm. Furthermore it can be seen that this streaming algorithm also satisfies the notion of sketchability. 

        The space complexity of such a sketching algorithm deserves special mention. The space requirement of linear sketching is the space needed to represent $t$ real numbers, where $t$ is the rank of the linear map used to sketch the inputs. When the weights are integers bounded by a polynomial in $n$, this can be used to show that the real numbers arising in the sketch can be represented by $O(\log n)$ bit rationals and so this translates to a small space sketch. This possibility goes away if the input is not polynomially bounded.

        All our algorithms are linear sketching algorithms as defined above.
        
    \end{enumerate}

\begin{remark}
    We note that \cite{LNW14,AHLW16} have shown that algorithms that work on dynamic streams are also linear sketching algorithms. Thus the assertion above that our algorithms are linear sketching algorithm (Item~3) seems redundant in view of the claim that they work in the dynamic setting (Item~1). However the results in \cite{LNW14,AHLW16} only apply to the case where the input streams are superpolynomially long (even requiring doubly exponential length). This is even necessary as proved by \cite{KP20}. Our results, on the other hand, only hold for polynomial length streams.  Thus in our setting, dynamic streams and linear sketching are not equivalent. 
\end{remark}

\subsubsection{Relation to single parameter approximability}\label{ssec:def-apx-equiv}

The traditional study of approximation algorithms typically focuses on a single parameter problem. Specifically, for $\alpha \in [0,1]$, $\maxF$ is said to be {\em $\alpha$-approximable} in space $s$ in the streaming setting if there is a space $s$ algorithm that on input a stream representing instance $\Psi$ of $\maxF$ outputs a number in $[\alpha \cdot \val_\Psi,\val_\Psi]$.
The connection between this single parameter approximability and the gapped problems we study is folklore. For the sake of completeness we describe the algorithmic implication below. 

\newcommand{\easy}{\textsc{Easy}}
\newcommand{\hard}{\textsc{Hard}}

\begin{proposition}\label{prop:approx-equivalence} 
Given $\cF \subseteq \{f:[q]^k \to \{0,1\}\}$, a space complexity measure $s : \N \to \N$, and sets $\easy,\hard \subseteq [0,1]\times [0,1]$ such that
for every $(\gamma,\beta) \in \easy$, $\gbmaxF$ is solvable in $s(n)$-space in the sketching model, and for every $(\gamma,\beta)\in\hard$, $\gbmaxF$ is not solvable in $s(n)$-space in the sketching model. Then for 
\[\alpha = \inf_{\beta \in [0,1]} \left \{ \sup_{\gamma\in(\beta,1] \rm{~s.t~} (\gamma,\beta)\in \easy} \{\beta/\gamma\}\right\},\]
and for every $\epsilon > 0$, there is an $(\alpha-\epsilon)$-approximation algorithm for $\maxF$ that uses $O_{k,q,\epsilon}(s(n))$ space in the sketching model. Conversely 
for 
\[\alpha = \inf_{\beta \in [0,1]} \left \{ \sup_{\gamma\in(\beta,1] \rm{~s.t~} (\gamma,\beta)\notin \hard} \{\beta/\gamma\}\right\},\]
and every $\eps>0$, every $(\alpha+\epsilon)$-approximation sketching algorithm for $\maxF$ requires $s(n)$ space.
\end{proposition}

\begin{proof}
The negative result is simple. We prove it in the contrapositive form by showing that if $\maxF$ has an $(\alpha+\epsilon)$-approximation algorithm using $s(n)$ space, then for every $(\gamma,\beta)$ with $\beta \leq \alpha \gamma$, $\gbmaxF$ is solvable in $s(n)$ space (and so $(\gamma,\beta) \notin \hard$). Suppose $\maxF$ has an $(\alpha+\epsilon)$ approximation algorithm $A$ using $s(n)$-space in the sketching model. Given $\gamma,\beta$ with $\beta/\gamma \geq \alpha$, we can use $A$  to solve the $(\gamma,\beta)$-$\maxF$ on input $\Psi$ as follows: Compute $A(\Psi)$ and output YES if $A(\Psi) \geq \beta$ and NO otherwise. Since $\beta \leq \alpha \gamma < (\alpha+\epsilon)\gamma$, it follows that if $\val(\Psi) \geq \gamma$ then $A(\Psi)$ will output some number greater that $\beta$ and our algorithm will output YES. If $\val(\Psi) \leq \beta$ then $A(\Psi)$ will output some number less than or equal to $\beta$ and our algorithm outputs NO. This yields the negative result.

For the positive result, we assume that $\easy$ is monotone in the following sense: If $(\gamma,\beta) \in \easy$ and $\beta' \leq \beta$ then $(\gamma,\beta') \in \easy$. (Note that we can assume this since an algorithm solving the  $(\gamma,\beta)$-$\maxF$ problem also solves the  $(\gamma,\beta')$-$\maxF$ problem.)
We also assume that every constraint in $\cF$ has at least one satisfying assignment.  (If not we can simply remove unsatisfiable constraints from $\cF$ and ignore them in the input stream.) Due to this assumption we have that a random assignment satisfies at least $\rho 
\triangleq q^{-k}$ fraction of the constraints. 
Let $\tau \triangleq \epsilon\cdot \rho/2$ and let 
$$A_\tau = \{(i\tau,j\tau) \in [0,1]^2 ~\vert~ i,j \in \mathbb{Z}^{\geq 0}, (i\tau,j\tau) \in \easy \}.$$
Thus for every $(\gamma,\delta) \in A_\tau$ there is a
$s(n)$-space algorithm  for $(\gamma,\beta)$-$\maxF$ with error probability $1/3$. By repeating this algorithm $O(\log(1/\tau))$ times and taking majority, we may assume the error probability is at most $1/(10\tau^2)$. We refer to this amplified algorithm as the $(\gamma,\beta)$-distinguisher below. In the following we consider the case where all $O(\tau^{-2})$ distinguishers output correct answers, which happens with probability at least $2/3$.

Our $O_{\tau}(s(n))$ space $(\alpha-\epsilon)$-approximation algorithm for $\maxF$ is the following: On input $\Psi$, run in parallel all the $(\gamma,\beta)$-distinguishers on $\Psi$, for every $(\gamma,\beta) \in A_\tau$. Let
\[
\beta_0 = \arg \max_{\beta} [\exists \gamma \text{ such that the }(\gamma,\beta)\text{-distinguisher outputs YES on }\Psi] \, .
\]
Output $\beta' = \max\{\rho,\beta_0\}$. 

We now prove that this is an $(\alpha-\epsilon)$-approximation algorithm. First note that by the correctness of the distinguisher we have $\beta' \leq \val_\Psi$. 
Let $\gamma_0$ be the smallest multiple of $\tau$ satisfying $\gamma_0 \geq (\beta_0 + \tau)/\alpha$.
By the definition of $\alpha$, we have that $(\gamma_0,\alpha\gamma_0)\in \easy$ and so by the monotonicity assumption on $\easy$ we have $(\gamma_0,\beta_0 + \tau) \in \easy$. So $(\gamma_0,\beta_0+\tau) \in A_\tau$ and so the $(\gamma_0,\beta_0+\tau)$-distinguisher must have output NO on $\Psi$ (by the maximality of $\beta_0$). By the correctness of this distinguisher we  conclude $\val_\Psi \leq \gamma_0 \leq (\beta_0 + \tau)/\alpha + \tau \leq (\beta'+\tau)/\alpha + \tau$. We now verify that $(\beta' + \tau)/\alpha + \tau \leq \beta'/(\alpha-\epsilon)$ and this gives us the desired approximation guarantee. We have 
\[(\beta' + \tau)/\alpha + \tau \leq (\beta'+2\tau)/\alpha \leq (\beta'/\alpha) \cdot (1 + 2\tau/\rho) = (\beta'/\alpha)(1+\epsilon) \leq (\beta'/(\alpha(1-\epsilon))),
\]
where the first inequality uses $\alpha \leq 1$, the second uses $\beta' \geq \rho$, the equality comes from the definition of $\tau$ and the final inequality uses $(1+\epsilon)(1- \epsilon) \leq 1$. This concludes the positive result.

\end{proof}

\subsection{Probabilistic notions and tools}\label{ssec:def-prob} 

We recall some standard notions from probability theory and mention some results we will use. 

\subsubsection{Total variation distance}
The total variation distance between probability distributions plays an important role in our analysis.

\begin{definition}[Total variation distance of discrete random variables]
Let $\Omega$ be a finite probability space and $X,Y$ be random variables with support $\Omega$. The total variation distance between $X$ and $Y$ is defined as follows.
\[
\|X-Y\|_{tvd} :=\frac{1}{2}
\sum_{\omega\in\Omega}\left|\Pr[X=\omega]-\Pr[Y=\omega]\right| \, .
\]
\end{definition}
We will use the triangle and data processing inequalities for the total variation distance.
\begin{proposition}[E.g.,{\cite[Claim~6.5]{KKS}}]\label{prop:tvd properties}
For random variables $X, Y$ and $W$:
\begin{itemize}
\item (Triangle inequality) $\|X-Y\|_{tvd}\geq\|X-W\|_{tvd}-\|Y-W\|_{tvd}$.
\item (Data processing inequality) If $W$ is independent of both $X$ and $Y$, and $f$ is a function, then  $\|f(X,W)-f(Y,W)\|_{tvd}\leq\|X-Y\|_{tvd}$.
\end{itemize}
\end{proposition}

\subsubsection{A concentration inequality}

We will use the following concentration inequality which is essentially an Azuma-Hoeffding style inequality for submartingales. The form we use is based on  \cite[Lemma~2.5]{KK19}, and allows for variables with different expectations. The analysis is a very slight modification of theirs.

\begin{lemma}\label{lem:our-azuma}
Let $X=\sum_{i\in[N]}X_i$ where $X_i$ are Bernoulli random variables such that for every $k\in[N]$, $\Exp[X_k \, |\, X_1,\dots,X_{k-1}]\leq p_k$ for some $p_k\in(0,1)$. Let $\mu=\sum_{k=1}^N p_k$. For every $\Delta>0$, we have:
\[
\Pr\left[X\geq\mu+\Delta\right]\leq\exp\left(-\frac{\Delta^2}{2\mu+2\Delta}\right) \, .
\]
\end{lemma}

\begin{proof}
Let $v = \Delta/(\mu+\Delta)$ and $u = \ln (1+v)$. We have
\[ \Exp[e^{uX}] = \Exp[\prod_{k=1}^N e^{uX_k}] \leq (1+p_N(e^u-1))\cdot \Exp[\prod_{k=1}^{N-1} e^{uX_k}] \leq \prod_{i=1}^N (1 +p_k(e^u-1))
= \prod_{i=1}^N (1 +p_kv) \leq e^{v\mu},
\] 
where the final inequality uses $1+x \leq e^x$ for every $x$ (and the definition of $\mu$). 
Applying Markov's inequality to the above, we have:
\[\Pr\left[X\geq\mu+\Delta\right]=\Pr\left[e^{uX}\geq e^{u(\mu+\Delta)}\right]\leq \Exp[e^{uX}]/ e^{u(\mu+\Delta)} \leq e^{v\mu-u\mu - u\Delta}.
\]
From the inequality $e^{v-v^2/2} \leq 1+v$ we infer $u \geq v-v^2/2$ and so the final expression above can be bounded as:
\[
\Pr\left[X\geq\mu+\Delta\right]\leq e^{v\mu-u\mu - u\Delta} \leq e^{\frac{v^2}2 (\mu+\Delta) - v \Delta} = e^{-\frac{\Delta^2}{2(\mu+\Delta)}},
\]
where the final equality comes from our choice of $v$.
\end{proof}

\subsection{Fourier analysis}\label{sec:fourier}\label{ssec:def-fourier} 
We will need the following basic notions from Fourier analysis over the Boolean hypercube (see, for instance,~\cite{o2014analysis}).
For a Boolean function $f: \{-1,1\}^k \to \R$ its Fourier coefficients are defined by $\widehat{f}(\vecv) = \Exp_{\veca\in\{-1,1\}^k}[f(\veca)\cdot(-1)^{\vecv^\top\veca}]$, where $\vecv\in\{0,1\}^k$. We need the following two important tools.

\begin{lemma}[Parseval's identity]\label{prop:parseval}
For every function $f\{-1,1\}^k\to\R$, 
\[
\|f\|_2^2=\frac{1}{2^k}\sum_{\veca\in\{-1,1\}^k}f(\veca)^2=\sum_{\vecv\in\{0,1\}^k}\widehat{f}(\vecv)^2 \, .
\]
\end{lemma}

Note that for every distribution $f$ on $\{-1,1\}^k$, $\widehat{f}(0^k)=2^{-k}$. For the uniform distribution $U$ on $\{-1,1\}^k$, $\widehat{U}(\vecv)=0$ for every $\vecv\neq0^k$. Thus, by \cref{prop:parseval}, for any distribution $f$ on $\{-1,1\}^k$:
\begin{align}\label{eq:dist}
\|f-U\|_2^2=\sum_{\vecv\in\{0,1\}^k}\left(\widehat{f}(\vecv)-\widehat{U}(\vecv)\right)^2=\sum_{\vecv\in\{0,1\}^k\backslash\{0^k\}}\widehat{f}(\vecv)^2 \, .
\end{align}

Next, we will use the following consequence of hypercontractivity for Boolean functions as given in \cite[Lemma 6]{GKKRW} which in turns relies on a lemma from \cite{KKL88}. 
\begin{lemma}\label{lem:kkl}
Let $f:\{-1,1\}^n\rightarrow\{-1,0,1\}$ and $A=\{\veca\in\{-1,1\}^n\, |\, f(\veca)\neq0\}$. If $|A|\geq2^{n-c}$ for some $c\in\N$, then for every $\ell\in\{1,\dots,4c\}$, we have
\[
\frac{2^{2n}}{|A|^2}\sum_{\substack{\vecv\in\{0,1\}^n\\\|\vecv\|_1=\ell}}\widehat{f}(\vecv)^2\leq\left(\frac{4\sqrt{2}c}{\ell}\right)^\ell \, .
\]
\end{lemma}

\subsection{Quantified theory of reals}\label{ssec:def-reals}

The decidability of several characterizations in this paper follows from the decidability of the ``quantified theory of the reals''. We describe the main problem and result here.

\begin{definition}[Quantified Polynomial Sentence]
    A quantified polynomial sentence over $K$ variables, $S$ polynomials of degree $D$ of quantifier width $w$ is given by (1) A Boolean formula $\Psi(Y_1,\ldots,Y_S)$ on $S$ Boolean variables;  (2) A set $\cP$ of $S$ polynomials $\cP=\{P_i(X_1,\ldots,X_K) \mid i \in [S]\}$, with each $P_i$ being a polynomial with real coefficients and of degree at most $D$ in $K$ variables; and (3) a partition $\Pi = (X_{[1]},\ldots,X_{[w]})$ of the set $\{X_1,\ldots,X_K\}$ and $w$ quantifiers $Q = (Q_1,\ldots,Q_w)$ with $Q_j \in \{\exists,\forall\}$ for every $j \in w$. The sentence $(\Psi,\cP,\Pi,Q)$ is defined to be TRUE if $Q_1 X_{[1]} Q_2 X_{[2]} \ldots Q_w X_{[w]} \Psi(Y_1(X_1,\ldots,X_K),\ldots, Y_S(X_1,\ldots,X_K))$ is true where $Y_i(X_1,\ldots,X_K) =$ TRUE if and only if $P_i(X_1,\ldots,X_S) \leq 0$.  
\end{definition}

Note that the syntax is rich enough to express conditions such as $P(X) \geq 0$ and $P(X) < 0$ by use of arithmetic negations ($-P(X) \leq 0$) and logical negations NOT$(P(X) \geq 0)$ where the logical negation is inserted into the Boolean formula $\Psi$. As an example the sentence ``Every positive number can be written as the square of a real number'' can be expressed as the sentence $\forall \alpha \exists \beta (-\alpha \geq 0) \vee ((\alpha - \beta^2) \geq 0) \vee (-(\alpha - \beta^2)) \geq 0$, which is a quantified sentence with 2 quantifiers, 2 variables parititioned into $\{\alpha\}$ and $\{\beta\}$ with quantifiers $Q_1 = \forall$ and $Q_2 = \exists$ and 3 polynomials of degree at most 2. This sentence happens to be TRUE. 

\begin{theorem}[\protect{\cite[Theorem 14.14, see also Remark 13.10]{BasuPR}}]\label{thm:bpr}
The truth of a quantified formula with $w$ quantifiers over $K$ variables and $S$ degree $D$ polynomial (potentially strict) inequalities can be decided in space $K^{O(w)} \log(SD)$ and time $(SD)^{K^{O(w)}}$. 
\end{theorem}

Specifically, Theorem 14.14 in \cite{BasuPR} asserts the time complexity above, and Remark 13.10 yields the space complexity.

\section{Results}\label{sec:results}

In this section we introduce our convex set framework that makes our classification of ``easy'' vs. ``hard'' sketching problems explicit. The sets are introduced in \cref{sec:main_notions}. We then state our main dichotomy theorem and also state its decidability in \cref{ssec:dichotomy}. Other results of this paper, including some strengthenings to the streaming setting, are stated in \cref{ssec:other-results}. We work out some example applications of the dichotomy theorem and strengthenings in \cref{ssec:examples}. Finally in \cref{ssec:short-proofs} we include proofs of all the simple results and corollaries of this section, leaving only the proofs of \cref{thm:main-detailed-dynamic}, \cref{thm:main-negative} and \cref{thm:exact} to later sections.

\subsection{The Convex Set Framework}\label{sec:main_notions}
The main objects that allow us to derive our characterization are the space of distributions on constraints that either allow a large number of constraints to be satisfied, or only a few constraints to be satisfied. To see where the distributions come from, note that
distributions of constraints over $n$ variables can naturally be identified with instances of weighted constraint satisfaction problem (where the weight associated with a constraint is simply its probability).

In this part we consider distributions of constraints over a set of $kq$ variables denoted $\vecx = (x_{i,\sigma}\, |\, i \in [k], \sigma \in [q])$. (We think of the variables as sitting in a $k \times q$ matrix with $i$ indexing the rows and $\sigma$ indexing the columns.) For 
$f \in \cF$ and $\veca \in [q]^k$, let 
$\cC(f,\veca)$ denote the constraint $f(x_{1,a_1},\ldots,x_{k,a_k})$. 
For an assignment  $\vecb = (b_{i,\sigma}\, |\, i \in [k], \sigma \in [q]) \in [q]^{kq}$ we use the notation $\cC(f,\veca)(\vecb)$ to denote the value $f(b_{1,a_1},\ldots,b_{k,a_k})$.
We let $\canon \in [q]^{kq}$ denote the assignment $\canon_{i,\sigma} = \sigma$. (In the following section we will use $\canon$ as our planted assignment.)

We now turn to defining the ``marginals'' of distributions. For $\cD \in \Delta(\cF \times [q]^k)$, we let $\vecmu(\cD) = (\mu_{f,i,\sigma})_{f\in\cF,i \in [k], \sigma\in [q]}$ be given by $\mu_{f,i,\sigma} = \Pr_{(g,\veca) \sim \cD} [ g = f \mbox{ and }a_i = \sigma]$. Thus the marginal $\vecmu(\cD)$ lies in $\R^{|\cF| \times qk}$.

We often reduce our considerations to families $\cF$ containing a single element. In such cases we simplify the notion of a distribution to $\cD \in \Delta([q]^k)$.
For $\cD \in \Delta([q]^k)$, we let $\vecmu(\cD) = (\mu_{i,\sigma})_{i \in [k], \sigma\in [q]}$ be given by $\mu_{i,\sigma} = \Pr_{\veca \sim \cD} [a_i = \sigma]$.

Next we introduce our family of distributions that capture our ``Yes'' and ``No'' instances. ``Yes'' instances are highly satisfied by our planted assignment, while ``No'' instances are not very satisfied by any ``column-symmetric'', independent, probabilistic assignment. The fact that we only consider distributions on $kq$ variables makes this a set in a finite-dimensional space. 

\begin{definition}[Space of YES/NO distributions]\label{def:sysn}
For $q,k \in \N$, $\gamma \in [0,1]$ and $\cF\subseteq\{f:[q]^k \to \{0,1\}\}$,  we let 
$$\sgyf = \left\{\cD \in \Delta(\cF \times [q]^k) ~\vert~ \Exp_{(f,\veca)\sim \cD} [\cC(f,\veca)(\canon)] \geq \gamma\right\}.$$
For $\beta \in [0,1]$ we let 
$$\sbnf = \left\{\cD \in \Delta(\cF \times [q]^k) ~\vert~ \forall (\cP_\sigma \in \Delta([q]))_{\sigma \in [q]}, \Exp_{(f,\veca)\sim \cD} \left[\Exp_{\vecb, b_{i,\sigma} \sim \cP_\sigma} [\cC(f,\veca)(\vecb)]\right] \leq \beta\right\}.$$
\end{definition}

By construction, for $\beta < \gamma$, the sets $\sgyf$ and $\sbnf$ are disjoint. (In particular for any $\cD \in \sgyf$, $\canon$ corresponds to a (deterministic!) column symmetric assignment that satisfies $\gamma > \beta$ fraction of constraints, so $\cD \not\in\sbnf$.) The key to the analysis of low-space sketching algorithms is that they only seem to be able to estimate the marginals of a distribution --- so we turn to exploring the marginals of the sets above. 

\begin{definition}[Marginals of Yes/NO Distributions]\label{def:marginals}
For $\gamma,\beta \in [0,1]$ and $\cF\subseteq\{f:[q]^k \to \{0,1\}\}$,  we let 
$$\kgyf = \{\vecmu(\cD) \in \R^{|\cF| kq} ~\vert~ \cD \in \sgyf \}
\mbox{ and }
\kbnf = \{\vecmu(\cD) \in \R^{|\cF|kq} ~\vert~ \cD \in \sbnf \}.$$
\end{definition}
See~\cref{ssec:examples} for some examples of the sets $\sgyf,\sbnf,\kgyf,\kbnf$.

\subsection{The dichotomy for sketching algorithms} \label{ssec:dichotomy}

The following theorem now formalizes the informal statement that low space sketching algorithms (see \autoref{def:sketching alg}) can only capture the marginals of distributions.

\begin{theorem}[Dichotomy for Sketching Algorithms]\label{thm:main-detailed-dynamic}
For every $q,k \in \N$, every family of functions $\cF \subseteq \{f:[q]^k \to \{0,1\}\}$ and for every $0\leq\beta<\gamma\leq1$, the following hold:
\begin{enumerate}
    \item If $\kgyf \cap \kbnf = \emptyset$, then $(\gamma,\beta)$-$\maxF$ admits a uniform randomized linear sketching algorithm 
    that uses $O(\log^3 n) $ space\footnote{In particular, the space complexity is $O(\log^3 n)$ bits, or $O(\log^2 n)$ cells where each cell is $O(\log n)$ bits long. Crucially while the constant in the $O(\cdot)$ depends on $k$, $\gamma$ and $\beta$, the exponent is a universal constant.} on instances on $n$ variables. \label{thmpart:positive-result-detailed}
    \item If $\kgyf \cap \kbnf \neq \emptyset$, then for every $\epsilon>0$, every (non-uniform randomized) sketching algorithm for the $(\gamma-\eps,\beta+\eps)$-$\maxF$ 
    requires $\Omega(\sqrt{n})$ space\footnote{Again, the constant hidden in the $\Omega$ notation depends on $k$, $\gamma$ and $\beta$.} on instances on $n$ variables. Furthermore, if $\gamma = 1$, then every sketching algorithm for $(1,\beta+\epsilon)$-$\maxF$ requires $\Omega(\sqrt{n})$ space. 
    \label{thmpart:negative-result-detailed}
 \end{enumerate}
\end{theorem}

We remark that Part~\ref{thmpart:positive-result-detailed} of \cref{thm:main-detailed-dynamic} is actually stronger and holds even for dynamic streams where constraints are added and deleted, provided the total length of the stream is polynomial in~$n$. 
\cref{thm:main-detailed-dynamic} is proved in two parts: \cref{thm:main-positive} proves \cref{thm:main-detailed-dynamic}, Part \ref{thmpart:positive-result-detailed} while \cref{thm:main-negative-dynamic} proves \cref{thm:main-detailed-dynamic}, Part \ref{thmpart:negative-result-detailed}.

We now complement \cref{thm:main-detailed-dynamic} by showing that the condition ``$\kgyf \cap \kbnf = \emptyset?$'' can be decided in polynomial space given $\gamma$ and $\beta$ as ratios of $\ell$-bit integers and members of $\cF$ as truth tables. (So the input is of size $O(\ell+ |\cF| \cdot q^k)$ and our algorithm needs space polynomial in this quantity.) 

\begin{theorem}\label{thm:pspace-dec}
For every $k,q \in \N$  $\cF\subseteq \{f:[q]^k\to \{0,1\}\}$, and $\ell$-bit rationals $\beta,\gamma \in [0,1]$ (i.e., $\beta$ and $\gamma$ are expressible as the ratio of two integers in $\{-2^\ell,\ldots,2^\ell\}$), the condition ``$\kgyf\cap\kbnf=\emptyset$?'' can be decided in space
$\poly(|\cF|,q^k,\ell)$ given truth tables of all elements of $\cF$ and $\gamma$ and $\beta$ as $\ell$-bit rationals.
\end{theorem}

We include a proof of \cref{thm:pspace-dec} in \cref{sec:dec_proof}.

\subsection{Other Results}\label{ssec:other-results}

\subsubsection{Approximation resistance of sketching algorithms}

We now turn to the notion of ``approximation resistant'' $\maxF$ problems. We start with a discussion where $\cF = \{f\}$. In the setting where constraints are applied to literals rather than variables, the notion of approximation resistance is used to refer to problems where it is hard to outperform a uniform random assignment. In other words if $\rho(f)$ is defined to be the probability that a random assignment satisfies $f$, then $\maxf$ is defined to be approximation resistant if $(1-\epsilon,\rho(f)+\epsilon)$-$\maxf$ is hard. In our setting however, where constraints are applied to variables, this notion is a bit more nuanced. Here it may be possible to construct functions where a random assignment does poorly and yet every instance has a much higher value.\footnote{Take for instance $f(x_1)=1$ iff $x_1 = 1$. The random assignment satisfies $f$ with probability $1/q$ while every instance is satisfiable!} In our setting, the correct notion is to simply consider the infimum value achieved over instances of $\maxf$. If this quantity is $\rho$ then it is trivial to get a $\rho$-approximation for $\maxf$ --- namely the algorithm that outputs the constant $\rho$ is always correct and gives a $\rho$-approximation. (Equivalently, $\gbmaxf$ can be decided by the algorithm that always outputs YES if $\beta < \rho$.) And if $(1-\epsilon,\rho(f)+\epsilon)$-$\maxf$ is hard for every $\epsilon > 0$ then we can say that $\maxf$ is approximation-resistant. 

The only catch with the above notion of approximation resistant is that $\rho$ may not be computable. To resolve this problem we introduce an alternate definition of this quantity $\rho$ and prove that it is equivalent and computable. We start with the definitions, generalized for all $\cF$. 

\begin{definition}[Approximation resistance for streaming/sketching algorithms]\label{def:approx-res}
For $\cF \subseteq \{f:[q]^k \to \{0,1\}\}$, we define  $$\rho_{\min} (\cF) = \lim\inf_{\Psi \textrm{ instance of }\maxF} \{\val_\Psi\}.$$
We say that $\maxF$ is {\em approximation-resistant} for streaming algorithms (resp. sketching algorithms) 
if for every $\epsilon>0$ there exists $\delta>0$ such that every streaming (resp. sketching) algorithm for $(1-\epsilon,\rho_{\min}(\cF)+\epsilon)$-$\maxF$ 
requires $\Omega(n^{\delta})$ space.
We also define
$$\rho(\cF) = \min_{\cD_{\cF} \in \Delta(\cF)} \left\{ \max_{\cD \in \Delta([q])} \left\{ \Exp_{f \sim \cD_{\cF},  \veca \sim \cD^k} [ f(\veca) ]         \right\} \right\}.$$
\end{definition}

The following proposition asserts the equivalence of $\rho_{\min}(\cF)$ and $\rho(\cF)$. 

\begin{proposition}\label{prop:rho-prop}
For every $q,k\in\N$, $\cF\subseteq \{f:[q]^k \to \{0,1\}\}$ we have $\rho_{\min}(\cF) = \rho(\cF)$.
\end{proposition}

\cref{prop:rho-prop} allows us to show that $\rho(\cF)$ is computable as asserted below. 

\begin{theorem}\label{thm:rho-min-compute}
There is an algorithm $A$ that, on input $\cF \subseteq \{[q]^k \to \{0,1\}\}$ presented as $|\cF|$ truth-tables and $\tau \in \R$ presented as an $\ell$-bit rational, answers the question ``Is $\rho_{\min}(\cF)\leq \tau$?'' in space $\poly(|\cF|,q^k,\ell)$. 
\end{theorem}

\cref{thm:main-detailed-dynamic} immediately yields a decidable characterization of $\maxF$ problems that are approximation resistant with respect to sketching algorithms.
\begin{theorem}[Classification of sketching approximation resistance]\label{cor:approx-res}
For every $q,k\in \N$, for every family $\cF \subseteq \{f:[q]^k \to \{0,1\}\}$, $\maxF$ is approximation resistant with respect to sketching algorithms if and only if $K_1^Y(\cF)\cap K_{\rho(\cF)}^N(\cF) \ne \emptyset$. Furthermore, if $\maxF$ is approximation-resistant with respect to sketching algorithms, then for every $\epsilon > 0$ we have that $(1,\rho(\cF)+\epsilon)$-$\maxF$ requires $\Omega(\sqrt{n})$ space for non-uniform randomized sketching algorithms. 
If $\maxF$ is not approximation-resistant with respect to sketching algorithms, then there exists $\epsilon > 0$ such that $(1-\epsilon,\rho(\cF)+\epsilon)$-$\maxF$ can be solved in polylogarithmic space by a uniform randomized linear sketching algorithm. Finally, given the truth-table of the functions in $\cF$ there is an algorithm running in space $\poly(q^k |\cF|)$ that decides whether or not $\maxF$ is approximation-resistant with respect to sketching algorithms.
\end{theorem}

\cref{prop:rho-prop}, \cref{thm:rho-min-compute}, and \cref{cor:approx-res} are proved  in \cref{sec:ar_proof}.

\subsubsection{Lower bounds in the streaming setting}\label{sec:lb-detail-insert}

We now turn to some special classes of CSPs where we can prove lower bounds in the streaming setting as opposed to only ruling out sketching algorithms. To describe these classes we need some definitions. 

We start by defining the notion of a ``one-wise independent'' distribution $\cD \in \Delta(\cF \times [q]^k)$. (We note that this is somewhat related to, but definitely not the same as the notion of a family $\cF$ that {\em supports} one-wise independence which was defined informally in \cref{sec:intro}. We will recall that notion shortly.)
We also define a broader notion of a ``padded one-wise pair'' of distributions.

\begin{definition}[One-wise independence and Padded one-wise independence of Distributions]\label{def:ow-and-pow}
For $\cD \in \Delta(\cF \times [q]^k)$ we say that $\cD$ is {\em one-wise independent} (or has ``uniform marginals'') if its marginal $\vecmu(\cD) = (\mu_{f,i,\sigma})_{f\in\cF,i \in [k], \sigma\in [q]}$ satisfies $\mu_{f,i,\sigma} = \mu_{f,i,\sigma'}$ for every $f \in \cF$, $i \in [k]$ and $\sigma,\sigma' \in [q]$. (In other words for every $f_0 \in \cF$ and $i \in [k]$, the random variable $a_i$ obtained by sampling $(f,(a_1,\ldots,a_k)) \sim \cD$ conditioned on $f = f_0$ and projecting to $a_i$ is uniformly distributed over $[q]$.) 

We say that a pair of distributions $(\cD_1,\cD_2)$ form a {\em padded one-wise pair} if there exist $\cD_0,\cD'_1,\cD'_2$ and $\tau \in [0,1]$ such that for every $i \in \{1,2\}$ we have $\cD'_i$ is one-wise independent and $\cD_i = \tau \cD_0 + (1-\tau) \cD'_i$. 
\end{definition}

Our main lower bound in the streaming setting asserts that if $\sgyf\times \sbnf$ contains a padded one-wise pair $(\cD_Y,\cD_N)$ then $(\gamma,\beta)$-$\maxF$ requires $\Omega(\sqrt{n})$-space.

\begin{restatable}[Streaming lower bound]{theorem}{restatethmmainnegative}
\label{thm:main-negative}
For every $q,k \in \N$, every family of functions $\cF \subseteq \{f:[q]^k \to \{0,1\}\}$ and for every $0 \leq \beta<\gamma \leq 1$, 
if there exists a padded one-wise pair of distributions $\cD_Y \in \sgyf$ and $\cD_N \in \sbnf$ then, 
for every $\epsilon>0$, every non-uniform randomized streaming algorithm that solves the $(\gamma-\eps,\beta+\eps)$-$\maxF$ problem requires $\Omega(\sqrt{n})$ space.
 Furthermore, if $\gamma = 1$, then $(1,\beta+\epsilon)$-$\maxf$ requires $\Omega(\sqrt{n})$ space.
\end{restatable}

\cref{thm:main-negative} is proved in \cref{sec:lb insert}. As stated above the theorem is more complex to apply than, say, \cref{thm:main-detailed-dynamic}, owing to the fact that the condition for hardness depends on the entire distribution (and the sets $S^Y_\gamma$ and $S^N_\beta$) rather than just marginals (or the sets $K^Y_\gamma$ and $K^N_\beta$). However it can be used to derive some clean results, specifically \cref{thm:one-wise} and \cref{thm:main-intro-k=q=2}, that do depend only on the marginals. We state these below after defining a notion of a function family supporting one-wise independence.

\begin{definition}[(Weakly/Strongly) Supporting One-wise Independence]\label{def:ow-support}
We say that a function $f:[q]^k\to\{0,1\}$ supports {\em one-wise independence} if there exists a distribution $\cD$ supported on $f^{-1}(1)$ whose marginals are uniform on $[q]$. 
We say that a family $\cF$ {\em strongly supports one-wise independence} if every function $f \in \cF$ supports one-wise independence. 
We say that a family $\cF$ {\em weakly supports one-wise independence} if there exists $\cF' \subseteq \cF$ satisfying $\rho(\cF') = \rho(\cF)$ such that  every function $f \in \cF'$ supports one-wise independence. 
\end{definition}

\begin{theorem}\label{thm:one-wise}
For every $q,k\in\N$ and $\cF \subseteq \{f:[q]^k \to \{0,1\}\}$ such that $\cF$ weakly supports one-wise independence, $\maxF$ is approximation resistant with respect to 
streaming algorithms. In particular, for every $\epsilon > 0$, every non-uniform randomized streaming algorithm for $(1,\rho(\cF)+\epsilon)$-$\maxF$ requires $\Omega(\sqrt{n})$ space.
\end{theorem}

\begin{remark}
We note that \cref{thm:one-wise-informal} differs from \cref{thm:one-wise} in that \cref{thm:one-wise-informal} asserted hardness for $\cF$ that strongly supports one-wise independence whereas \cref{thm:one-wise} asserts hardness for $\cF$ that weakly supports one-wise independence. Thus \cref{thm:one-wise} is stronger and implies \cref{thm:one-wise-informal}.
\end{remark}

Finally we turn to \cref{thm:main-intro-k=q=2}. Below we assert a more detailed version of the theorem along the lines of \cref{thm:main-detailed-dynamic} in this case. 
\begin{theorem}\label{thm:main-detailed-k=q=2} For every family $\cF \subseteq \{f:[2]^2 \to \{0,1\}\}$, and for every $0 \leq \beta < \gamma \leq 1$, the following hold:
\begin{enumerate}
    \item If $K_\gamma^Y(\cF) \cap K_\beta^N(\cF) = \emptyset$, then $(\gamma,\beta)$-$\maxF$ admits a uniform randomized linear sketching algorithm that uses $O(\log^3 n) $ space. 
    \item If $K_\gamma^Y(\cF) \cap K_\beta^N(\cF) \neq \emptyset$, then for every $\epsilon>0$,  then $(\gamma-\eps,\beta+\eps)$-$\maxF$ in the streaming setting requires $\Omega(\sqrt{n})$ space\footnote{The constant hidden in the $\Omega$ notation may depend on $k$ and $\epsilon$.}. Furthermore, if $\gamma = 1$, then $(1,\beta+\epsilon)$-$\maxF$ in the streaming setting requires $\Omega(\sqrt{n})$ space for non-uniform randomized streaming algorithms. 
 \end{enumerate}
\end{theorem}

\cref{thm:main-detailed-k=q=2} clearly implies \cref{thm:main-intro-k=q=2}. We prove \cref{thm:one-wise} and \cref{thm:main-detailed-k=q=2} in \cref{sec:streaming_proof}. 

\subsubsection{Classification of exact computability}\label{ssec:exact}

Finally for the sake of completeness we show that all ``non-trivial'' CSPs are hard to solve exactly. ``Trivial'' families are those where all satisfiable constraints are satisfied by a constant assignment, as defined precisely below.
    
\begin{definition}[Constant satisfiable]
For $\sigma \in [q]$ and $\cF \subseteq \{f:[q]^k \to \{0,1\}\}$ we say that 
$\cF$ is $\sigma$-satisfiable if for every $f \in \cF\setminus\{\veczero\}$ we have that $f(\sigma^k) = 1$. 
We say $\cF$ is constant-satisfiable if there exists $\sigma \in [q]$ such that $\cF$ is $\sigma$-satisfiable.
\end{definition}

Our theorem below asserts that constant satisfiable families are the only ones that are solvable exactly. And for additive $\epsilon$ approximations to the maximum fraction of satisfiable constraints, they require space growing polynomially in $\epsilon^{-1}$.

\begin{restatable}{theorem}{exactthm}\label{thm:exact}
For every $q,k \in \N$, every family of functions $\cF \subseteq \{f:[q]^k \to \{0,1\}\}$ the following hold: 
\begin{enumerate}
    \item If $\cF$ is constant satisfiable, then there exists a deterministic linear sketching algorithm that uses $O(\log n)$ space and solves $\maxF$ exactly optimally.
    \item If $\cF$ is not constant satisfiable, then the following hold in the streaming setting: 
    \begin{enumerate}
        \item Every probabilistic algorithm solving $\maxF$ exactly requires $\Omega(n)$ space.
        \item For every $\epsilon=\eps(n)>0$, $(1,1-\epsilon)$-$\maxF$ requires $\Omega(\min\{n,\epsilon^{-1}\})$-space\footnote{\label{note1} The constant hidden in the $\Omega$ depends on $\cF$, but (obviously) not on $\epsilon$.} on sufficiently large inputs. 
        \item For $\rho_{\min}(\cF)$ defined in \cref{def:approx-res}, for every $\rho_{\min}(\cF) < \gamma < 1$ and every $\epsilon=\eps(n)>0$, $(\gamma,\gamma-\epsilon)$-$\maxF$ requires $\Omega(\min\{n,\epsilon^{-2}\})$-space\textsuperscript{\rm{\ref{note1}}} on sufficiently large inputs. 
    \end{enumerate}
 \end{enumerate}
\end{restatable}

\cref{thm:exact} is proved in \cref{sec:exact}.

\subsection{Some Examples}\label{ssec:examples}

We consider three basic examples of general $q$-CSP and illustrate how to apply \cref{thm:main-negative} to determine their approximability.

The first example is \textsf{Max-DICUT} described below.

\begin{examplebox}{Example 1 (\textsf{Max-DICUT}).}\label{Example:Max-DICUT}
Let $f(x,y):[2]^2\to\{0,1\}$ with $f(x,y)=1$ if and only if $x=2$ and $y=1$. Note that $\textsf{Max-DICUT}=\textsf{Max-CSP}(\{f\})$ with $q=k=2$. Observe that for every distribution $\cD\in\Delta([q]^k)$ with probability density vector $\vecphi(\cD)=(\phi_{22},\phi_{21},\phi_{12},\phi_{11})$, we have for every $0\leq\gamma,\beta\leq1$
\[
\sgyf = \{\cD \, |\, \phi_{21}\geq\gamma\}
\]
and
\[
\sbnf = \left\{\cD \, |\, \max_{p,q\in[0,1]} p(1-p)\cdot\phi_{22}+pq\cdot\phi_{21}+(1-q)(1-p)\cdot\phi_{12}+(1-q)q\cdot\phi_{11}\leq\beta\right\} \, .
\]
Also, note that the marginal vector $\vecmu(\cD)=(\mu_{22},\mu_{21},\mu_{12},\mu_{11})$ and $\vecphi(\cD)$ satisfy the following relations:
\[
\left\{\begin{array}{l}
\mu_{22} = \phi_{12} + \phi_{22}  \\
\mu_{21} = \phi_{11} + \phi_{21}  \\
\mu_{12} = \phi_{21} + \phi_{22}  \\
\mu_{11} = \phi_{11} + \phi_{12} \, .  
\end{array}
\right.
\]
Note that for every $\cD\in\Delta([q]^k)$, we have $\cD\in S^N_{1/4}$. In particular, the uniform distribution $\unif([2]^2)\in S^N_{1/4}$. Since the distribution given by the density vector $(\phi_{22}=0,\phi_{21}=1/2,\phi_{12}=1/2,\phi_{11}=0)$ also has uniform marginals and belongs to $S^Y_{1/2}$, we have that for every $\beta\geq1/4$, $K^Y_{1/2}\cap\kbnf\neq\emptyset$. So it suffices to focus on the case where $\gamma\geq1/2$.

Fix $\gamma\geq1/2$, we want to compute the minimum $\beta$ such that $\kgyf\cap\kbnf\neq\emptyset$. The kernel of the mapping from probability density $\vecphi$ to the marginal vector $\vecmu$ is spanned by $(1,-1,-1,1)$. Then simple calculations show that the minimum $\beta$ is achieved when $\vecmu=(1-\gamma,\gamma,\gamma,1-\gamma)$ with  $(0,\gamma,1-\gamma,0)\in\sgyf$ and $(1-\gamma,2\gamma-1,0,1-\gamma)\in\sbnf$. Specifically, 
\begin{align*}
\beta &= \max_{p,q\in[0,1]}(p(1-p)+q(1-q))\cdot(1-\gamma)+pq\cdot(2\gamma-1)\\
&=\max_{p,q\in[0,1]}\frac{(1-\gamma)^2}{3-4\gamma}-\frac{3-4\gamma}{2}\cdot\left(\left(p+\frac{1-\gamma}{4\gamma-3}\right)^2+\left(q+\frac{1-\gamma}{4\gamma-3}\right)^2\right)-\frac{(2\gamma-1)}{2}\cdot(p-q)^2 \, .
\end{align*}
When $\gamma\geq2/3$, the expression is maximized by $p=q=1$ and hence $\beta=2\gamma-1$. When $1/2\leq\gamma\leq2/3$, the expression is maximized by $p=q=(1-\gamma)/(3-4\gamma)$ and hence $\beta=(1-\gamma)^2/(3-4\gamma)$. 

We thus get that the set $H^{\cap} \triangleq \{(\gamma,\beta) \in [0,1]^2 | K^Y_\gamma \cap K^N_\beta \ne \emptyset\}$ (of {\em hard} problems) is given by (see also~\autoref{fig:max dcut Hcap}):
\begin{align*}
H^\cap = & ~~~~~~ \left[0,\frac12\right]\times \left[\frac14,1\right] \\
         & \cup ~~\left\{(\gamma,\beta) | \gamma \in \left[\frac12,\frac23\right], \beta \in \left[\frac{(1-\gamma)^2}{3-4\gamma},1\right]\right\} \\
         & \cup  ~~\left\{(\gamma,\beta) |\gamma\in \left[\frac23,1\right], \beta \in \left[2\gamma - 1,1\right]\right\}.
\end{align*}  
(We note that  \cite[Example 1]{CGSV21} gives exactly the same set as the hard set of $\maxtwoand$, which is a related but not identical result.)

\begin{figure}[H]
    \centering
    \includegraphics[width=5cm]{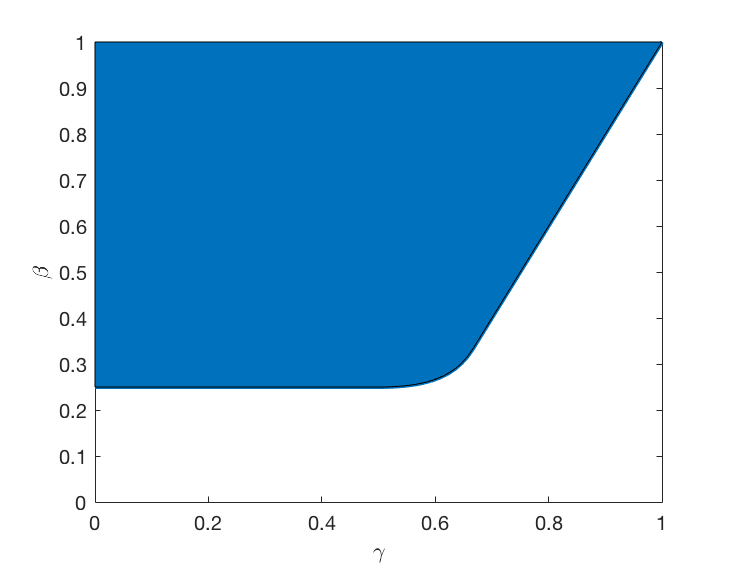}
    \caption{A plot of $H^\cap$.}
    \label{fig:max dcut Hcap}
\end{figure}

Finally, over $\gamma\in[2/3,1]$, $\beta/\gamma$ is minimized at $(\gamma,\beta)=(2/3,1/3)$ and $\beta/\gamma=1/2$; over $\gamma\in[1/2,2/3]$, $\beta/\gamma$ is minimized at $(\gamma,\beta)=(3/5,4/15)$ and $\beta/\gamma=4/9$, yielding $4/9$ as the approximability threshold.

Specifically, \cref{prop:k=2padded} gives us that any pair of distributions $\cD_Y, \cD_N \in \Delta(\cF\times [2]^2), \cD_Y\in S^Y_{3/5}, \cD_Y\in S^N_{4/15}$ witnessing $K^Y_{3/5}\cap K^Y_{4/15}\neq\emptyset$ forms a padded one-wise pair. Finally, \cref{thm:main-negative}, applied to the padded one-wise pair $(\cD^Y, \cD^N)$, implies that \textsf{Max-DICUT} cannot be approximated better with a factor $(4/9+\eps)$ in space $o(\sqrt{n})$ in the streaming setting, which is consistent with the findings in~\cite{CGV20} for the $\textsf{Max-DICUT}$ problem.

\end{examplebox}

\begin{examplebox}{Example 2 ($\textsf{Max-}q\textsf{UG}$).}
Let $k=2$ and $q\geq2$. Let $\cF=\{f:[q]^2\to\{0,1\}\, |\, f^{-1}(1)\text{ is a bijection}\}$. Note that $\textsf{Max-}q\textsf{UG}=\textsf{Max-CSP}(\cF)$. We claim that the quantity $\alpha=\inf_{\beta}\alpha(\beta)=1/q$ where $\alpha(\beta) = \sup_{\gamma | K_\gamma^Y \cap K_\beta^N = \emptyset} \{\beta/\gamma\}$. First, note that $\cD\in S^N_{1/q}$ for every $\cD$ and hence implies $\alpha\geq1/q$.

For simplicity we work with the alphabet $\Z_q =\{0,\ldots,q-1\}$ instead of $[q]$. For $\tau \in \Z_q$ let 
$f_\tau\in\cF$ be the constraint $f_\tau(x,y)=1$ if and only if $x-y = \tau \pmod q$. 
Let $\cD^Y$ be the uniform distribution over $\{(f_{\tau},\sigma,\sigma+\tau)\, |\, \sigma,\tau \in \Z_q\}$. 
Note that obviously we have $\cD^Y \in S^Y_1$. Now let $\cD^N$ be the uniform distribution over $\{f_\tau\, |\, \tau\in\Z_q\}\times\Z_q^2$. Note that for any assignment to two variables $x_{1,\sigma_1},x_{2,\sigma_2}$ the probability
over $\tau$ that it satisfies $f_\tau(x_{1,\sigma_1},x_{2,\sigma_2})$ is exactly $1/q$. If follows that any assignment to
$(x_{i,\sigma})_{i,\sigma}$ satisfies exactly $1/q$ fraction of the constraints in $\cD^N$ and so $\cD^N \in S_{1/q}^N$. 
Observe that the marginals of $\cD^Y$ and $\cD^N$ are the same, i.e., $\vecmu(\cD^Y)=\vecmu(\cD^N)=\vecmu(\unif(\{f_\tau\}\times\Z_q^2))$. This gives us $\vecmu(\unif(\{f_\tau\}\times[q]^2))\in K^Y_1\cap K^N_{1/q}$ so we have $\alpha(\beta)=\beta$ for $\beta \geq 1/q$. Minimizing this over $\beta$,~\cref{thm:main-negative}, applied to the one-wise independent distribution $\cD^Y$ and $\cD^N$, gives that the problem can not be approximated better than $1/q$ in space $o(\sqrt{n})$ in the streaming setting, which is consistent with the findings in~\cite{GT19} for the $\textsf{Max-}q\textsf{UG}$ problem.

\end{examplebox}

\begin{examplebox}{Example 3 ($\textsf{Max-}q\textsf{Col}$)}

Let $k=2$ and $q\geq2$. Let $\cF=\{f_{\ne}\}$ where $f_{\ne}:[q]^2\to\{0,1\}$ is given by $f_{\ne}(x,y) = 1 \Leftrightarrow x \ne y$. 
Note that $\textsf{Max-}q\textsf{Col}=\textsf{Max-CSP}(\cF)$. 
We claim that the quantity $\alpha=\inf_{\beta}\alpha(\beta)=1-1/q$ where $\alpha(\beta) = \sup_{\gamma | K_\gamma^Y \cap K_\beta^N = \emptyset} \{\beta/\gamma\}$. First, note that $\cD\in S^N_{1-1/q}$ for every $\cD$ and hence implies $\alpha\geq1-1/q$. We now show this is also the upper bound by exhibiting $\cD^Y$ and $\cD^N$.

Let $\cD^Y$ be the uniform distribution over $\{(f_{\ne},\sigma,\tau)\, |\, \sigma \ne \tau \in [q]\}$. 
Note that obviously we have $\cD^Y \in S^Y_1$. 
Now let $\cD^N$ be the uniform distribution over $\{f_{\ne}\}\times[q]^2$. This leads to  $\beta = \max_{\cP_\sigma} \{\Exp_{(f,a_1,a_2) \sim \cD^N} [\Exp_{x \sim \cP_{a_1}, y \sim \cP_{a_2}}[ f(x,y)]] \}$. The independence of $a_1$ and $a_2$ in $\cD^N$ allows us to simplify this to $\max_{\cP \in \Delta([q])} \{\Exp_{x,y \sim \cP} [f_{\ne}(x,y)] \}$ and the latter is easily seen to be at most $1- 1/q$. 
Thus we conclude $\cD^N \in S_{1-1/q}^N$. 
Since the marginals of $\cD^Y$ and $\cD^N$ are the same, i.e., $\vecmu(\cD^Y)=\vecmu(\cD^N)=\vecmu(\unif(\{f_{\ne}\}\times[2]\times[q]))$, this gives us $\vecmu(\unif(\{f_{\ne}\}\times[2]\times[q]))\in K^Y_1\cap K^N_{1/q}$ so we have $\alpha(\beta)=\beta$ for $\beta \geq 1-1/q$. Minimizing this over $\beta$,~\cref{thm:main-negative}, applied to the one-wise independent distribution $\cD^Y$ and $\cD^N$, gives that the problem can not be approximated better than $1-1/q$ in space $o(\sqrt{n})$ in the streaming setting.
\end{examplebox}

Another example along the same vein is analyzed in a subsequent work by Singer, Sudan and Velusamy~\cite{SSV21} who show that $(1-1/q,(1/2)(1-1/q))$-$\maxF$ is hard for $\cF = \{f_{<}\}$ where $f_{<}:[q]^2\to \{0,1\}$ is given by  $f_{<}(x,y) = 1$ if and only if $x < y$.
This analysis forms a critical step in their improved analysis of the Maximum Acyclic Subgraph Problem (which is not captured in our framework).

\subsection{Some proofs of theorems asserted in this section}\label{ssec:short-proofs}

In this subsection we prove all results asserted in \cref{ssec:dichotomy} and \cref{ssec:other-results}, with the exception of \cref{thm:main-detailed-dynamic}, \cref{thm:main-negative} and \cref{thm:exact}. 

\subsubsection{Decidability of the Classification}\label{sec:dec_proof}

We prove \cref{thm:pspace-dec} in this section. 
The following lemma states some basic properties of the sets $\sgyf, \sbnf, \kgyf$ and $\kbnf$ and uses them to express the condition ''$\kgyf \cap \kbnf = \emptyset?$'' in the quantified theory of reals. 

\begin{lemma}\label{lem:convex}
For every $k,q \in \N$ $\beta,\gamma \in [0,1]$ and $\cF\subseteq \{f:[q]^k\to \{0,1\}\}$, the sets $\sgyf$, $\sbnf$, $\kgyf$ and $\kbnf$ are bounded, closed, and convex. Furthermore, the condition $\kgyf \cap \kbnf = \emptyset$ can be expressed in the quantified theory of reals with $2$ quantifier alternations, $O(|\cF|q^k + q^2)$ variables, and polynomials of degree at most $k+1$. 
\end{lemma}

\begin{proof} 
We start by observing that $\Delta(\cF \times [q]^k)$ is a bounded convex polytope in $\R^{|\cF| \times [q]^k}$. Furthermore, viewing $\cD$ as a vector in $\R^{|\cF| \times [q]^k}$, for any given $\vecb \in [q]^k$ the quantity $\Exp_{(f,a) \sim \cD}[ C(f,\veca)(\vecb)]$ is linear in $\cD$. Thus $\sgyf$ is given by a single linear constraint on $\Delta(\cF \times [q]^k)$ making it a bounded convex polytope as well. 
$\sbnf$ is a bit more complex - in that there are infinitely many linear inequalities defining it (one for every distribution $(\cP_\sigma)_{\sigma\in[q]}$). Nevertheless this leaves $\sbnf$ bounded, closed (as infinite intersection of closed sets is closed), and convex (though it may no longer be a polytope). 
Finally since $\kgyf$ and $\kbnf$ are linear projections of $\sgyf$ and $\sbnf$ respectively, they retain the features of being bounded, closed and convex.

Finally to get an effective algorithm for intersection detection, we express the intersection condition in the quantified theory of the reals.
To get this, we note that $(\cP_\sigma)_{\sigma\in[q]}$ can be expressed by $q^2$ variables, specifically using variables $\cP_\sigma(\tau)$ for every $\sigma,\tau \in [q]$ where $\cP_\sigma(\tau)$ denotes the probability of $\tau$ in $\cP_\sigma$. In terms of these variables (which will eventually be quantified over) the condition $\Exp_{(f,\veca)\sim \cD} \left[\Exp_{\vecb, b_{i,\sigma} \sim \cP_\sigma} [\cC(f,\veca)(\vecb)]\right] \leq \beta$ is a multivariate polynomial inequality in $(\cP_\sigma)_\sigma$ and $\cD$. (Specifically we get a polynomial of total degree at most $k$ in $(\cP_\sigma)_\sigma$, and of total degree at most one in $\cD$.) This allows us to use the following quantified system to express the condition $\kgyf\cap\kbnf\ne\emptyset$:

\begin{align}
\exists &\cD_Y,\cD_N \in \R^{|\cF| \times q^k}, ~ \forall ((\cP_\sigma)_\sigma) \in \R^{q^2} \mbox{ s.t. } \nonumber\\
&\cD_Y, \cD_N, (\cP_\sigma)_\sigma, \forall \sigma\in[q] \mbox{ are distributions,}\label{eq:poly-dist}\\
& \forall f_0 \in \cF, \forall i \in [k], \tau\in[q] ~ \Pr_{(f,\veca) \sim \cD_Y} [f = f_0 \mbox{ and } a_i = \tau] = \Pr_{(f,\veca) \sim \cD_N} [f = f_0 \mbox{ and } a_i = \tau],\label{eq:poly-marginal} \\
& \Exp_{(f,\veca) \sim \cD_Y} [\cC(f,\veca)(\canon)] \geq \gamma, \label{eq:poly-gamma}\\
& \Exp_{(f,\veca)\sim \cD_N} \left[\Exp_{\vecb, b_{i,\sigma} \sim \cP_\sigma} [\cC(f,\veca)(\vecb)]\right] \leq \beta. \label{eq:poly-beta}
\end{align}

Note that \cref{eq:poly-dist,eq:poly-marginal,eq:poly-gamma} are just linear inequalities in the variables $\cD_Y,\cD_N$. 

As noticed above \cref{eq:poly-beta} is an inequality in the $\cP_\sigma$s and $\cD_N$, of total degree at most $k+1$.

We thus get that the intersection problem can be expressed in the quantified theory of the reals by an expression with two quantifier alternations, $2|\cF|q^k + q^2$ variables and $O(|\cF|q^k + q^2)$ polynomial inequalities, with polynomials of degree at most $k+1$. (Most of the inequalities are of the form $\cD_Y(\vecb) \geq 0$ or $\cD_N(\vecb) \geq 0$. We also have $O(|\cF|kq)$ equalities (saying probabilities must add to one and matching the marginals of $\cD_Y$ and $\cD_N$). Of the two remaining, \cref{eq:poly-gamma} is linear, only \cref{eq:poly-beta} is a higher-degree polynomial. 
\end{proof}

We are now ready to prove \cref{thm:pspace-dec}.

\begin{proof}[Proof of \cref{thm:pspace-dec}]
The quantified polynomial system given by \cref{lem:convex} yields parameters $K = O(|\cF|q^k + q^2)$ for the number of variables and $w = 2$ for the number of alternations. Applying \cref{thm:bpr} with these parameters yields the theorem. 
\end{proof}

\subsubsection{Approximation Resistance}\label{sec:ar_proof}

We start by proving \cref{prop:rho-prop} which asserts that $\rho(\cF) = \rho_{\min}(\cF)$. 

\begin{proof}[Proof of \cref{prop:rho-prop}]
We start by showing $\rho(\cF) \leq \rho_{\min}(\cF)$. 
Fix an instance $\Psi$ of $\maxF$ and let $\cD_{\cF}$ be the distribution on $\cF$ obtained by picking a random constraint of $\Psi$ and looking at the function (while ignoring the variables that the constraint is applied to). By the definition of $\rho(\cF)$, there exists a distribution $\cD \in \Delta([q])$ such that $\Exp_{f \sim \cD_{\cF},  \veca \sim \cD^k} [ f(\veca) ] \geq \rho(\cF)$. Now consider a random assignment to the variables of $\Psi$ where variable $x_j$ is assignment a value independently according to $\cD$. It can be verified that $\Exp_{\vecx}[\val_\Psi(\vecx)] \geq \rho(\cF)$ and so $\val_\Psi \geq \rho(\cF)$. We thus conclude that
$\rho(\cF) \leq \val_\Psi$ for all $\Psi$ and so $\rho(\cF) \leq \rho_{\min}(\cF)$.

We now turn to the other direction. We prove that for every $\epsilon > 0$ we have $\rho_{\min}(\cF) \leq \rho(\cF)+\epsilon$ and the inequality follows by taking limits. Let $\cD_\cF$ be the distribution achieving the minimum in the definition of $\rho(\cF)$. Given $\epsilon > 0$ let $n$ be a sufficiently large integer and let $m = O(n^k/\epsilon)$. Let $\Psi$ be the instance of $\maxF$ on $n$ variables with $m$ constraints chosen as follows: For every $\vecj \in [n]^k$ with distinct coordinates and every $f \in \cF$ we place $\lfloor \cD_{\cF}(f)/\epsilon \rfloor$ copies of the constraint $(f,\vecj)$.

We claim that the $\Psi$ generated above satisfies $\val_\Psi \leq \rho(\cF) + \epsilon/2 + O(1/n)$ and this suffices for the proposition. To see the claim, fix an assignment $\vecnu \in [q]^n$ and let $\cD\in\Delta([q])$ be the distribution induced by sampling $i \in [n]$ uniformly and outputting $\vecnu_i$. On the one hand we have from the definition of $\rho(\cF)$ that $\Exp_{f \sim \cD_\cF, \veca \sim \cD^k} \left[f(\veca) \right]  \leq \rho(\cF)$.
On the other hand we have that the distribution obtained by sampling a random constraint $(f,\vecj)$ of $\Psi$ and outputting $(f,\vecnu|_\vecj)$ is $\epsilon/2+O(1/n)$ close in total variation distance to sampling $f \sim \cD_\cF$ and $\veca \sim \cD^k$. (The $\epsilon/2$ gap comes from the rounding down of each constraint to an integral number, and the $O(1/n)$ gap comes from the fact that $\vecj$ is sampled from $[n]$ without replacement.)
We thus conclude that 
$$\val_\Psi(\vecnu) \leq \Exp_{f \sim \cD_\cF, \veca \sim \cD^k} \left[f(\veca) \right]+ \epsilon/2+O(1/n) \leq \rho(\cF) + \epsilon/2+O(1/n) \leq \rho(\cF) + \epsilon.$$
Since this holds for every $\vecnu$ we conclude that this upper bounds $\val_\Psi$ as well thus establishing the claim, and hence the proposition.
\end{proof}

Now we prove \cref{thm:rho-min-compute} which asserts that $\rho(\cF)$ and thus $\rho_{\min}(\cF)$ is computable.

\begin{proof}[Proof of \cref{thm:rho-min-compute}]

By \cref{prop:rho-prop} we have 
$$\rho_{\min}(\cF) = \rho(\cF) = \min_{\cD_\cF \in \Delta(\cF)} \left\{ \max_{\cD \in \Delta([q])} \left\{ \Exp_{f \sim \cD_\cF, \veca \sim \cD^k} \left[f(\veca) \right]  \right\} \right\}. $$

Viewing $\cD_\cF \in \R^{|\cF|}$ and $\cD \in \R^q$ and noticing that the inner expectation is a degree $k+1$ polynomial in $\cD_\cF$ and $\cD$ we get, again using \cref{thm:bpr}, that there is a space $\poly(|\cF|,q^k,\ell)$ algorithm answering the question ``Is $\rho_{\min}(\cF)\leq \tau$?''.
\end{proof}

Finally we prove \cref{cor:approx-res} which shows that the classification of approximation-resistant $\maxF$ problems is decidable. 

\begin{proof}[Proof of \cref{cor:approx-res}]
By \cref{thm:main-detailed-dynamic} we have that $\maxF$ is approximation-resistant if and only if $K_{1-\eps}^Y(\cF) \cap K_{\rho(\cF)+\eps}^N(\cF) \ne \emptyset$ for every small $\eps > 0$. Taking limits as $\eps\to 0$, this implies that $\maxF$ is approximation resistant if and only if $K_1^Y(\cF) \cap K_{\rho(\cF)}^N(\cF) \ne \emptyset$ . If $K_1^Y(\cF) \cap K_{\rho(\cF)}^N(\cF) = \emptyset$, then by the property that these sets are closed (see~\cref{lem:convex}), we have that there must exist $\eps > 0$ such that  $K_{1 - \eps}^Y(\cF) \cap K_{\rho(f)+\eps}^N(\cF) = \emptyset$. In turn this implies, again by \cref{thm:main-detailed-dynamic}, that the $(1-\eps,\rho(\cF)+\eps)$-approximation version of $\maxF$ can be solved by a streaming algorithm with $O(\log^3 n)$ space.

To get the decidability result, we combine the ingredients from the proof of \cref{thm:rho-min-compute,thm:pspace-dec}. (We can't use them as blackboxes since $\rho_{\min}(\cF)$ may not be rational.) We create a quantified system of polynomial inequalities using a new variable called $\rho$ and expressing the conditions $\rho = \rho(\cF)$ (with further variables for $\cD_\cF$ and $\cD$ as in the proof of \cref{thm:rho-min-compute}) and expressing the conditions $K^Y_1(\cF) \cap K^N_\rho(\cF) \ne \emptyset$ as in the proof of \cref{thm:pspace-dec}. The resulting expression is thus satisfiable if and only if $\cF$ is approximation resistant, and this satisfiability can be decided in polynomial space in the input length $q^k |\cF|$ by \cref{thm:bpr}.
\end{proof}

\subsubsection{Streaming Lower Bounds}\label{sec:streaming_proof}

We now prove \cref{thm:one-wise} (assuming \cref{thm:main-negative}), which asserts that families that support one-wise independence are approximation-resistant. 

\begin{proof}[Proof of \cref{thm:one-wise}]
Let $\cF' \subseteq \cF$ be a family satisfying $\rho(\cF') = \rho(\cF)$ such that every function $f \in \cF'$ supports one-wise independence. Let $\cD_{\cF} \in \Delta(\cF')$
minimize 
$\max_{\cD \in \Delta([q])} \left\{ \Exp_{f \sim \cD_{\cF},  \veca \sim \cD^k} [ f(\veca) ] \right\}$. For $f \in \cF'$ let $\cD_{\cF} \in \Delta([q]^k)$ be the distribution with uniform marginals supported on $f^{-1}(1)$. Now let $\cD_Y$ be the distribution where $(f,\veca) \sim \cD_Y$ is sampled by picking $f \in \cD_{\cF}$ (where $\cD_{\cF}$ is being viewed as an element of $\Delta(\cF)$) and then sampling $\veca \sim \cD_{\cF}$. Now let $\cD_N = \cD_{\cF} \times \textsf{Unif}([q]^k)$. Note that $\cD_Y$ and $\cD_N$ are one-wise independent distributions with $\vecmu(\cD_Y) = \vecmu(\cD_N)$. In particular this implies that $(\cD_Y,\cD_N)$ are a padded one-wise pair. We claim that $\cD_Y \in S_1^Y(\cF)$ and $\cD_N \in S_{\rho(\cF)}^N(\cF)$. The theorem then follows immediately from \cref{thm:main-negative}.

To see the claim, first note that by definition we have that $(f,\veca) \sim \cD_Y$ satisfies $\cC(f,\veca)(\canon) = f(\veca) = f_0(\veca)=1$ with probability $1$. Thus we have $\Exp_{(f,\veca)\sim \cD} [\cC(f,\veca)(\canon)] = 1$ and so $\cD_Y \in S_1^Y(\cF)$. 
Now consider $(f,\veca)\sim \cD_N$. To show $\cD_N \in S_{\rho(\cF)}^N(\cF)$ we need to show that for every family of distributions $(\cP_\sigma \in \Delta([q]))_{\sigma \in [q]}$, the following holds
$\Exp_{(f,\veca)\sim \cD} \left[\Exp_{\vecb, b_{i,\sigma} \sim \cP_\sigma} [\cC(f,\veca)(\vecb)]\right] \leq \rho(\cF)$. Now let $\cP$ be the distribution where $\tau \sim \cP$ is sampled by picking $\sigma \sim \textsf{Unif}([q])$ and then sampling $\tau \sim \cP_\sigma$. We have
\begin{align*}
\Exp_{(f,\veca)\sim \cD} \left[\Exp_{\vecb, b_{i,\sigma} \sim \cP_\sigma} [\cC(f,\veca)(\vecb)]\right]
& = \Exp_{f \sim \cD_{\cF}, \veca \sim \textsf{Unif}([q]^k)} \left[\Exp_{\vecb, b_{i,\sigma} \sim \cP_\sigma} [\cC(f,\veca)(\vecb)]\right]\\
& = \Exp_{f \sim \cD_{\cF}} \left[\Exp_{\veca\sim \cP^k} [f(\veca)]\right]
\\
& \leq \rho(\cF') \\
& = \rho(\cF) \,.
\end{align*}
This proves $\cD_N\in S_{\rho(\cF)}^N(\cF)$ and thus proves the theorem.
\end{proof}

Next we turn to proving \cref{thm:main-detailed-k=q=2}. To do so, we first prove the following simple proposition above distributions or pairs of Boolean variables.

\begin{proposition}\label{prop:k=2padded}
If $\cD_Y, \cD_N \in \Delta(\cF\times [2]^2)$ satisfy $\vecmu(\cD_Y) = \vecmu(\cD_N)$ then $(\cD_Y,\cD_N)$ form a padded one-wise pair.
\end{proposition}

\begin{proof}

For $g \in \cF$, let $P(g)$ denote the probability of sampling a constraint $(f,\vecj)\sim \cD_Y$ with function $f=g$ and let $P$ denote this distribution. Note that since $\vecmu(\cD_Y) = \vecmu(\cD_N)$, $\cD_N$ also samples $g$ with the same probability. Let $\cD_{Y|g}$ denote $\cD_Y$ conditioned on $f = g$. Similarly let $\cD_{N|g}$ denote $\cD_N$ conditioned on $f = g$. 

Now $\cD_{Y|g}$ and $\cD_{N|g}$ are distributions from $\Delta(\{g\} \times [2]^2)$ with matching marginals. We'll show that there exist $\cD_{0|g}$, $\cD'_{Y|g}$ and $\cD'_{N|g}$, and $\tau_g$ such that (1) $\cD_{Y|g} = \tau_g \cD_{0|g} + (1-\tau_g)\cD'_{Y|g}$, (2) $\cD_{N|g} = \tau_g \cD_{0|g} + (1-\tau_g)\cD'_{N|g}$ and (3) $\cD'_{Y|g}$ and $\cD'_{N|g}$ are one-wise independent. Let $\cD_{Y|g} = (p_{1,1},p_{1,2},p_{2,1},p_{2,2})$ where $p_{i,j}$ denotes the probability $\Pr_{(a,b) \sim \cD_{Y|g}}[a=i,b=j]$. If $\cD_{N|g}$ has matching marginals with $\cD_{Y|g}$ then there exists a $\delta_g \in [-1,1]$ such that  $\cD_{N|g} = (p_{1,1}-\delta_g,p_{1,2}+\delta_g,p_{2,1}+\delta_g,p_{2,2}-\delta_g)$. Assume without loss of generality that $\delta_g \geq 0$. Let
$\tau_g = 1-2\delta_g$,
$\cD_{0|g} = \frac1{1-2\delta_g}(p_{1,1}-\delta_g,p_{1,2},p_{2,1},p_{2,2}-\delta_g)$,
$\cD'_{Y|g}= (1/2,0,0,1/2)$ and $\cD'_{N|g}= (0,1/2,1/2,0)$. It can be verified that $\cD'_{Y|g}$ and $\cD'_{N|g}$ are one-wise independent, $\cD_{Y|g} = \tau_g\cD_{0|g} + (1-\tau_g)\cD'_{Y|g}$ and $\cD_{N|g} = \tau_g\cD_{0|g} + (1-\tau_g)\cD'_{N|g}$.

Now let $\tau = \Exp_{f \sim P}[\tau_f]$, and $\cD_0 \in \Delta(\cF\times [q]^k)$ be the distribution where $\veca=(f,\vecb) \in \{\cF\}\times [2]^2$ is sampled with probability $\frac{P(f)\cdot \tau_f \cdot \cD_{0|f}(\veca)}{\tau}$, where $\cD_{0|f}(\veca)$ is the probability of sampling $\veca$ from $\cD_{0|f}$. Note that this is a valid probability distribution as \[\sum_{f\in \cF}\sum_{\vecb\in [2]^2} \frac{P(f)\cdot \tau_f \cdot \cD_{0|f}(f,\vecb)}{\tau} = \sum_{f\in \cF} \frac{P(f)\cdot \tau_f}{\tau }\cdot \sum_{\vecb\in[2]^2} \cD_{0|f}((f,\vecb)) = 1 \, . \] Similarly define $\cD'_Y$ and $\cD'_N$ such that $\veca$ is sampled with probability $\frac{P(f)\cdot (1-\tau_f) \cdot \cD'_{Y|f}(\veca)}{1-\tau}$ and probability $\frac{P(f)\cdot (1-\tau_f) \cdot \cD'_{N|f}(\veca)}{1-\tau}$, respectively.
It can be verified that these choices satisfy (1)~$\cD_{Y} = \tau \cD_{0} + (1-\tau)\cD'_{Y}$, (2) $\cD_{N} = \tau \cD_{0} + (1-\tau)\cD'_{N}$ and (3) $\cD'_{Y}$ and $\cD'_{N}$ are one-wise independent. It follows that $\cD_Y$ and $\cD_N$ form a padded one-wise pair.
\end{proof}

Combining \cref{prop:k=2padded,thm:main-negative} we immediately get the following theorem, which in turn implies \cref{thm:main-intro-k=q=2}.

\begin{proof}[Proof of \cref{thm:main-detailed-k=q=2}]
Part (1) is simply the specialization of Part (1) of \cref{thm:main-detailed-dynamic} to the case $k=2$. For Part~(2), suppose $\vecmu \in K^Y_\gamma \cap K^N_\beta$. Let $\cD_Y \in S^Y_\gamma$ and $\cD_N \in S^N_\beta$ be distributions such that $\vecmu(\cD_Y) = \vecmu(\cD_N) = \vecmu$. Then by \cref{prop:k=2padded} we have that $\cD_Y$ and $\cD_N$ form a padded one-wise pair, and so \cref{thm:main-negative} can be applied to get Part (2). 
\end{proof}

\section{A Streaming Approximation Algorithm for \texorpdfstring{$\maxF$}{Max-CSP(F)}}\label{sec:alg}

In this section we give our main algorithmic result --- an $O(\log^3 n)$-space linear sketching streaming algorithm for $(\gamma,\beta)$-$\maxF$ if $K_\gamma^Y = K_\gamma^Y(\cF)$ and $K_\beta^N = K_\beta^N(\cF)$ are disjoint. (See \cref{def:marginals}.) 

The algorithm in fact works in the (general) dynamic setting where the input instance $\Psi= (C_1,\ldots,C_m; w_1,\ldots,w_m)$ is obtained by inserting and deleting (unweighted) constraints, possibly with repetitions and thus leading to a (integer) weighted instance. Formally, the instance $\Psi= (C_1,\ldots,C_m; w_1,\ldots,w_m)$ is presented as a stream $\sigma_1,\ldots,\sigma_\ell$ where $\sigma_t = (C'_t,w'_t)$ and $w'_t \in \{-1,1\}$ such that $w_i = \sum_{t \in [\ell] : C_i = C'_t} w'_t$. For the algorithmic result to hold, we require that $w_i$'s are non-negative at the end of the stream but the intermediate values can be arbitrary. Furthermore the algorithm requires that the length of the stream be polynomial in $n$ (or else there will be a logarithmic multiplicative factor in the length of the stream in the space usage). 

We now state our main theorem of this section which simply restates Part (1) of \cref{thm:main-detailed-dynamic}.

\begin{theorem}\label{thm:main-positive}
For every $q,k \in \N$, every family of functions $\cF \subseteq \{f:[q]^k \to \{0,1\}\}$ and for every $0\leq\beta<\gamma\leq1$
if $\kgyf \cap \kbnf = \emptyset$, then $(\gamma,\beta)$-$\maxF$ in the dynamic setting admits a probabilistic linear sketching streaming algorithm that uses $O(\log^3 n) $ space.
\end{theorem}

We start with a brief overview of our algorithm. Roughly, given an instance $\Psi$ on $n$ variables with $m$ constraints, our streaming algorithm (implicitly) works with an $n \times q$ bias non-negative matrix $\bias$ whose $(i,\sigma)$th entry tries to capture how much the $i$th variable would like to be assigned the value $\sigma$ (according to our approximation heuristic). Note that any such matrix is too large for our algorithm, so the algorithm does not explicitly maintain this matrix. Our heuristic ensures that $\bias$ is updated linearly by every constraint and so the rich theory of norm-approximations of matrices under linear updates can be brought into play to compute any desired norm of this matrix. Given the intuition that $\bias_{i,\sigma}$ represents the preference of variable $i$ for value $\sigma$, a natural norm of interest to us is $\|\bias\|_{1,\infty} \triangleq \sum_{i=1}^n \{\max_{\sigma\in[q]} \{\bias_{i,\sigma}\}\}$. This norm, fortunately for us, is well-known to be computable using $O(q\log^3 n)$ bits of space~\cite{AKO} (assuming $\bias$ is updated linearly) and we use this algorithm as a black box. 

The question then turns to asking how $\bias$ should be defined. On input a stream $\sigma_1,\ldots,\sigma_\ell$ representing an instance $\Psi=(C_1,\ldots,C_m)$ with $\sigma_i = (C'_i = (\vecj(i),\vecb(i)),w'_i)$, how should $\bias$ be updated? Presumably the $i$-th update will only involve the rows $\vecj(i)_1,\ldots,\vecj(i)_k$ but how should these be updated and how should this update depend on the function $f_i$? Here is where the disjointness of $K^Y$ and $K^N$ comes into play. (We  suppress $\cF$ and $\gamma$ and $\beta$ in the notation of the sets $S^Y_\gamma$, $S^N_\beta$ and $K^Y_\gamma$ and $K^N_\beta$ in this overview.) We show that these sets are convex and closed, and so there is a hyperplane (with margin) separating the two sets. Let $\veclambda = (\lambda_{f,i,\sigma})_{f\in \cF,i \in [k],\sigma\in[q]}$ be the coefficients of this separating hyperplane and let $\tau_N < \tau_Y$ be thresholds such that $\langle\veclambda,\vecmu\rangle\geq \tau_Y$ for $\vecmu \in K^Y$ and $\langle\veclambda,\vecmu\rangle\leq \tau_N$ for $\vecmu \in K^N$. It turns out that the coefficients of $\veclambda$ give us exactly the right information to determine the update to the bias vector: Specifically given an element $\sigma_i$ of the stream with constraint $C'_i = (f_i,\vecj(i))$ and weight $w'_i$ and $\ell \in [k]$ and $\sigma \in [q]$, we add $\lambda_{f_i,\ell,\sigma}\cdot w'_i$ to $\bias_{\vecj(i)_\ell,\sigma}$. We are unable to provide intuition for why these updates work but the proof that the algorithm works is nevertheless quite short!

We now turn to describing our algorithm.
Recall by \cref{lem:convex} that the set $S^Y,S^N,K^Y,K^N$ are all convex and closed. This implies the existence of a separating hyperplane when $K^Y$ and $K^N$ do not intersect. We use a mild additional property to conclude that the coefficients of this hyperplane are non-negative, and we later use this crucially in the computation of the \bias~ of the instance.

\begin{proposition}\label{prop: StrongerHyperplaneSeparation}
Let $\beta,\gamma$ and $\cF$ be such that  $0\le \beta < \gamma \le 1$ and $\kgyf \cap \kbnf = \emptyset$. Then there exists a \emph{non-negative} vector $\veclambda = (\lambda_{f,i,\sigma})_{f\in\cF,i \in [k], \sigma\in [q]}$ and real numbers $\tau_Y > \tau_N$ such that 
\[
 \forall \vecmu \in \kgyf, ~~ \langle \veclambda,\vecmu \rangle \ge \tau_Y \text{~~and~~ } \forall \vecmu \in \kbnf, ~~ \langle \veclambda,\vecmu \rangle \le \tau_N \, .
\]
\end{proposition}

\begin{proof}
The existence of a separating hyperplane follows from standard convexity (see, e.g., \cite[Exercise 2.22]{boyd2004convex}). For us this implies there exists $\veclambda' \in \R^{|\cF| \times kq}$ and $\tau'_N < \tau'_Y$ such that 
\[
\forall \vecmu \in \kgyf, ~~ \langle \veclambda',\vecmu \rangle \ge \tau'_Y \text{~~and~~ } \forall \vecmu \in \kbnf, ~~ \langle \veclambda',\vecmu \rangle \le \tau'_N \, .
\]
But $\veclambda'$ is not necessarily a positive vector. To remedy this we use the fact that $\kgyf\cup\kbnf$ is contained in a hyperplane whose coefficients are themselves positive. In particular we note that for every $\cD \in \Delta(\cF\times [q]^k)$ we have $\langle\mu(\cD),\mathbf{1}\rangle = k$ where $\mathbf{1} \in \R^{|\cF| \times kq}$ is the all ones vector, as verified below:
\[\langle\mu(\cD),\mathbf{1}\rangle = \sum_{f \in \cF, i \in [k], \sigma\in [q]} \mu_{f,i,\sigma} = \sum_{i \in [k]} \left(\sum_{f \in \cF, \sigma\in [q]} \mu_{f,i,\sigma}\right) = \sum_{i\in[k]} 1 = k. \]

Let $\lambda'_{\min} = \min_{f,t,\sigma } \lambda'_{f,t,\sigma}$. Now let $\veclambda$, $\tau_Y$ and $\tau_N$ be given by:
$$\lambda_{f,t,\sigma} = \lambda'_{f,t,\sigma}+|\lambda'_{\min}|\, , ~~ \tau_Y = \tau'_Y+ k\cdot|\lambda'_{\min}| \mbox{~~and~~} \tau_N=\tau'_N+ k\cdot|\lambda'_{\min}|\, .$$
Observe that $\veclambda$ is a non-negative vector and $\tau_Y>\tau_N$. We also have:
\[
\forall \vecmu \in \kgyf, ~~ \langle \veclambda,\vecmu \rangle =  \langle \veclambda',\vecmu \rangle + |\lambda'_{\mathrm{min}}| \ge \langle \mathbf{1},\vecmu \rangle \geq \tau'_Y + k|\lambda'_{\min}| = \tau_Y \]
as desired. Similarly also get 
$\forall \vecmu \in \kbnf, ~~ \langle \veclambda,\vecmu \rangle \le \tau_N$, concluding the proof.
\end{proof}

To use the vector $\veclambda$ given by \cref{prop: StrongerHyperplaneSeparation} we introduce the notion of the bias matrix and the bias of a $\maxF$ instance $\Psi$. 

\begin{definition}[Bias (matrix)]\label{def:bias-matrix}
For a non-negative vector $\veclambda = (\lambda_{f,i,\sigma})_{f\in\cF,i \in [k], \sigma\in [q]} \in \R^{|\cF|kq}$, and instance $\Psi = (C_1,\ldots,C_m;w_1,\ldots,w_m)$ of $\maxF$ where $C_i = (f_i,\vecj(i))$, where $f_i\in \cF$ and $\vecj(i)\in [n]^k$, we let the {\em $\veclambda$-bias matrix} of $\Psi$, denoted $\bias_{\veclambda}(\Psi)$, be the matrix in $\R^{n\times q}$ given by 
\[
\bias_{\veclambda}(\Psi)_{\ell,\sigma} = \frac{1}W \cdot \sum_{i \in [m], t \in [k] : \vecj(i)_t = \ell} \lambda_{f_i,t,\sigma} \cdot w_i \, ,
\]
for $\ell \in [n]$ and $\sigma \in [q]$, where $W=\sum_{i\in [m]} w_i$. The $\veclambda$-bias of $\Psi$, denoted $B_{\veclambda}(\Psi)$, is defined as $B_{\veclambda}(\Psi) = \sum_{\ell=1}^n \max_{\sigma\in [q]} \bias_{\veclambda}(\Psi)_{\ell,\sigma}$. 

\end{definition}

Key to our algorithm for approximating $\maxF$ is the following algorithm to compute the $\ell_{1,\infty}$ norm of a matrix. Recall that for a matrix $M \in \R^{a \times b}$ the $\ell_{1,\infty}$ norm is the quantity $\|M\|_{1,\infty} = \sum_{i \in [a]} \{\max_{j \in {[b]}} \{|M_{ij}|\}\}$. 

\begin{theorem}[Implied by {\protect{\cite[Theorem 4.5]{AKO}}}]\label{thm:ell1infty}
There exists a constant $c>0$ such that the $\ell_{1,\infty}$ norm of an $n \times q$ matrix $M$ can be estimated by a linear sketch to within a multiplicative error of $(1+\epsilon)$ in the turnstile streaming model with $O(\epsilon^{-c} \cdot q \cdot  \log^2 n)$ words (or with $O(\epsilon^{-c} \cdot q \cdot  \log^3 n)$ bits).
\end{theorem}

We note that Theorem 4.5 in \cite{AKO} is much more general. \cref{thm:ell1infty} is the special case corresponding to $X = \ell_{\infty}$ and $E_X$ being simply the identity function. $\alpha(\cdots)$ in this case turns out to be $O(\log n)$ leading to the bounds above \cite{Andoni:personal}.

Note that there is a slight distinction between the definitions of $B_{\veclambda}(\Psi)$ and $\|\bias_{\veclambda}(\Psi)\|_{1,\infty}$, but these quantities are equal since $\bias_{\veclambda}$ is a non-negative matrix (which in turn follows from the fact that $\veclambda$ is non-negative). We thus get the following corollary.

\begin{corollary}\label{cor:bias-alg}
There exists a constant $c$ such that for every $k,q,\cF$ and $\epsilon > 0$, there exists a linear sketching streaming algorithm running in space $O(\epsilon^{-c} \cdot \log^3 n)$ that on input a stream $\sigma_1,\ldots,\sigma_\ell$ representing a $\maxF$ instance $\Psi = (C_1,\ldots,C_m;w_1,\ldots,w_m)$ on $n$ variables, outputs a $(1\pm\epsilon)$ approximation to $B_{\veclambda}(\Psi)$. 
\end{corollary}

We are now ready to describe our algorithm for $(\gamma,\beta)$-$\maxF$. 

\begin{algorithm}[H]
	\caption{A streaming algorithm for $(\gamma,\beta)$-$\maxF$}
	\label{alg:main alg}
\begin{algorithmic}[1]
		\Input A stream $\sigma_1,\ldots,\sigma_\ell$ representing an instance $\Psi$ of $\maxF$. 
		    \State Let $\veclambda \in \R^{|\cF|kq},\tau_N$ and $\tau_Y$ be as given by~\cref{prop: StrongerHyperplaneSeparation} separating $K_\gamma^Y(f)$ and $K_\beta^N(f)$, so $\veclambda$ is non-negative and $\tau_N < \tau_Y$. 
			\State Let $\epsilon = \frac{\tau_Y-\tau_N}{2(\tau_Y+\tau_N)}$.
			\State Using \cref{cor:bias-alg} compute a $(1\pm\epsilon)$ approximation $\tilde{B}$ to $B_{\veclambda}(\Psi)$, i.e., 
			$$(1-\epsilon)B_{\veclambda}(\Psi) \leq \tilde{B} \leq (1+\epsilon)B_{\veclambda}(\Psi) \mbox{ with probability at least $2/3$}.$$ 
			\If{$\tilde{B} \leq \tau_N (1+\epsilon)$}
				\Statex {\bf Output:} NO.
			\Else
				\Statex {\bf Output:} YES.
			\EndIf
	\end{algorithmic}
\end{algorithm}

Given \cref{cor:bias-alg} it follows that the algorithm above uses space $O(\log^3 n)$ on instances on $n$ variables. In what follows we prove that the algorithm correctly solves $(\gamma,\beta)-\maxF$.

\subsection{Analysis of the correctness of \autoref{alg:main alg}}

\begin{lemma}\label{lem:correctness_algorithm}
\cref{alg:main alg} correctly solves $(\gamma,\beta)$-$\maxF$, if $K_\gamma^Y(\cF)$ and $K_\beta^N(\cF)$ are disjoint.  
Specifically, for every $\Psi$, let $\tau_Y,\tau_N,\epsilon,\veclambda,\tilde{B}$ be as given in~\cref{alg:main alg}, we have:
\begin{eqnarray*}
\val_\Psi \geq \gamma & \Rightarrow & B_{\veclambda}(\Psi) \geq \tau_Y \mbox{ and } \tilde{B} > \tau_N(1+\epsilon) \, , \\
\mbox{ and } \val_\Psi \leq \beta & \Rightarrow & B_{\veclambda}(\Psi) \leq \tau_N \mbox{ and } \tilde{B} \leq \tau_N(1+\epsilon) \, ,
\end{eqnarray*}
provided $(1-\epsilon)B_\lambda(\Psi) \leq \tilde{B} \leq (1+\epsilon)B_\lambda(\Psi)$.
\end{lemma}

In the rest of this section, we will prove~\cref{lem:correctness_algorithm}. 
The key to our analysis is a distribution $\cD(\Psi^\vecb) \in \Delta(\cF \times [q]^k)$ that we associate with every instance $\Psi$ and assignment $\vecb\in [q]^{n}$ to the variables of $\Psi$. 
If $\Psi$ is $\gamma$-satisfied by assignment $\vecb$, we prove that $\vecmu(\cD(\Psi^\vecb)) \in K_\gamma^Y(\cF)$. On the other hand, if $\Psi$ is not $\beta$-satisfiable by any assignment, we prove that for every $\vecb$, $\vecmu(\cD(\Psi^\vecb)) \in K_\beta^N(\cF)$. Finally we also show that the bias $B_{\veclambda}(\Psi)$ relates to $\veclambda(\cD(\Psi^\vecb)) \triangleq \langle\vecmu(\cD(\Psi^\vecb), \veclambda \rangle$, where the latter quantity is exactly what needs to be computed (by \cref{prop: StrongerHyperplaneSeparation}) to distinguish the membership of $\vecmu(\cD(\Psi^\vecb))$ in $K_\gamma^Y(\cF)$ versus the membership in $K_\beta^N(\cF)$.

The key step is the definition of these distributions that allows the remaining steps (esp. \cref{lem:optimum_NO}) to be extended, which we present now.

Given an instance $\Psi=(C_1,\ldots,C_m;w_1,\ldots,w_m)$ on $n$ variables with $C_i = (f_i,\vecj(i))$ and an assignment $\vecb \in [q]^{n}$, the distribution $\cD(\Psi^\vecb) \in \Delta(\cF \times [q]^k)$ is sampled as follows: Sample $i \in [m]$ with probability $w_i/W$ where $W=\sum_{i\in [m]} w_i$, and output $(f_i,\vecb\mid_{\vecj(i)})$.

We start by relating the bias $B_{\veclambda(\Psi)}$ to $\cD(\Psi)$.

\begin{lemma}\label{lem:bias_comparsion}
For every vector $\vecb\in[q]^{n}$, we have $\veclambda(\cD(\Psi^\vecb)) = \sum_{\ell=1}^n \bias_{\veclambda}(\Psi)_{\ell,b_\ell} $. Consequently we have $B_{\veclambda}(\Psi) = \sum_{\ell=1}^n \max_{\sigma\in [q]} \bias_{\veclambda}(\Psi)_{\ell,\sigma} = \max_{\vecb\in [q]^{n}} \{ {\veclambda}(\cD(\Psi^\vecb)) \}$.
\end{lemma}
\begin{proof}
We start with the first equality. Fix $b\in [q]^n$. Given $f\in \cF$, $t\in [k]$, and $\sigma\in [q]$, we have
$\mu(\cD(\Psi^\vecb))_{f,t,\sigma} = \frac{1}{W} \sum_{i=1}^m w_i \cdot \mathbbm{1}[f_i=f, b_{j(i)_t}=\sigma]$. Hence,
\begin{align*}
    \veclambda(\cD(\Psi^\vecb)) &= \sum_{f\in \cF,t\in[k],\sigma\in[q]} \mu(\cD(\Psi^\vecb))_{f,t,\sigma} \cdot  \lambda_{f,t,\sigma}\\
    &= \frac{1}{W} \sum_{f\in\cF,t\in[k],\sigma\in[q]} \sum_{i\in [m]}w_i \cdot \mathbbm{1}[f_i=f, b_{j(i)_t}=\sigma] \cdot \lambda_{f,t,\sigma}\\
    &= \frac{1}{W} \sum_{i\in[m],t\in[k],\sigma\in[q]:b_{j(i)_t}=\sigma} w_i \cdot \lambda_{f_i,t,\sigma}\\
    &=  \sum_{\ell=1}^n \frac{1}{W} \sum_{i\in[m],t\in[k]:j(i)_t=l} w_i \cdot \lambda_{f_i,t,b_l}\\
    &= \sum_{\ell=1}^n \bias_{\veclambda}(\Psi)_{\ell,b_\ell} \, . \\
\end{align*}

For the final equality, observe that
\[
B_{\veclambda}(\Psi) = \sum_{\ell=1}^n \max_{\sigma\in [q]} \bias_{\veclambda}(\Psi)_{\ell,\sigma} = \max_{\vecb\in [q]^{n}}\sum_{\ell=1}^n \bias_{\veclambda}(\Psi)_{\ell,b_\ell} = \max_{\vecb\in [q]^{n}} \{ {\veclambda}(\cD(\Psi^\vecb)) \} \, .
\]
\end{proof}
The following lemmas relate $\val_\Psi$ to the properties of $\cD(\Psi^\veca)$.

\begin{lemma}\label{lem:optimum_YES}
For every $\Psi\in \maxF$ and $\vecb\in [q]^{n}$, if $\val_\Psi(\vecb) \geq \gamma$ then $\cD(\Psi^\vecb) \in S_\gamma^Y(\cF)$. 
\end{lemma}

\begin{proof}
Follows from the fact that
\[
\Exp_{(f,\veca)\sim \cD(\Psi^\vecb)}[C(f,a)(\canon)] = \frac{1}{W}\sum_{i=1}^m w_i \cdot f_i(b\mid_{j(i)}) = \val_\Psi(\vecb) \ge \gamma,
\]
implying $\cD(\Psi^\vecb) \in S_\gamma^Y(\cF)$.
\end{proof}

\begin{lemma}\label{lem:optimum_NO}
For every $\Psi\in \maxF$, if $\val_{\Psi} \leq \beta$, then for all $\vecb\in [q]^{n}$, we have $\cD(\Psi^\vecb) \in S_\beta^N(\cF)$.
\end{lemma}

\begin{proof}
We prove the contrapositive. We assume that $\exists \vecb\in [q]^n$ such that $\cD(\Psi^\vecb) \notin S_\beta^N(\cF)$ and show that this implies $\val_{\Psi} > \beta$. Then there exists $(P_\sigma\in \Delta([q]))_{\sigma \in [q]}$ satisfying the following inequality $\Exp_{(f,a)\sim \cD(\Psi^\vecb)} \left[\Exp_{\vecc, c_{i,\sigma} \sim \cP_\sigma} [\cC(f,\veca)(\vecc)]\right] > \beta$.

We thus have
\begin{eqnarray*}
    \beta & < & \Exp_{(f,a)\sim \cD(\Psi^\vecb)} \left[\Exp_{\vecc, c_{i,\sigma} \sim \cP_\sigma} [\cC(f,\veca)(\vecc)]\right] \\
    &= &\Exp_{\vecc, c_{i,\sigma} \sim \cP_\sigma} \left[\Exp_{(f,a)\sim \cD(\Psi^\vecb)} [\cC(f,\veca)(\vecc)]\right]\\
    &= &\Exp_{\vecc, c_{i,\sigma} \sim \cP_\sigma} \left[\frac{1}{W} \sum_{i=1}^m w_i \cdot f_i((c_{t,b_{j(i)_t}})_{t\in[k]})\right]\\
    &= &\frac{1}{W}\sum_{i=1}^m w_i \cdot \Exp_{\vecc, c_{i,\sigma} \sim \cP_\sigma} \left[ f_i((c_{t,b_{j(i)_t}})_{t\in[k]})\right]\\
    &= &\frac{1}{W}\sum_{i=1}^m w_i \cdot \Exp_{\vecx, x_{\ell} \sim \cP_{b_{\ell}}}\left[ f_i((x_{j(i)_t})_{t\in[k]})\right]\\ 
    &= &\Exp_{\vecx, x_{\ell} \sim \cP_{b_{\ell}}} \left[\frac{1}{W}\sum_{i=1}^m w_i \cdot f_i((x_{j(i)_t})_{t\in[k]})\right]\\
    &= & \Exp_{\vecx, x_{\ell} \sim \cP_{b_{\ell}}}\left[\val_{\Psi}(\vecx)\right]\\
    & \leq & \max_{\vecx \in [q]^n} \val_\Psi(\vecx) \\
    & = & \val_\Psi\, 
\end{eqnarray*}
which contradicts the assumption that $\val_\Psi\le \beta$. This concludes the proof of the claim and hence the lemma. 
\end{proof}

The key step above is the one asserting  $\frac{1}{W}\sum_{i=1}^m w_i \cdot \Exp_{\vecc, c_{i,\sigma} \sim \cP_\sigma} \left[ f_i((c_{t,b_{j(i)_t}})_{t\in[k]})\right] = \frac{1}{W}\sum_{i=1}^m w_i \cdot \Exp_{\vecx, x_{\ell} \sim \cP_{b_{\ell}}}\left[ f_i((x_{j(i)_t})_{t\in[k]})\right]$ which relies crucially on column symmetry of the distributions used in the definition of $\sbnf$ in \cref{def:sysn}.
Without this restriction, or even more stringent ones, this step of the rounding would fail.
And the reason we can't use a more stringent restriction will become clear in the proof of \cref{thm:main-negative} (and is specifically used in the proof of \cref{lem:csp value}).
We also note that this key  equality relies on the assumption that the variables in a single constraint are {\em distinct}. In particular the left hand side assumes $c_{i,\sigma}$s are drawn independently whereas the right side allows this only for the distinct variables $x_\ell$ in a constraint. 
\section{Sketching and Streaming Space Lower Bounds for \texorpdfstring{$\maxF$}{Max-CSP(F)}}\label{sec:lower-bound}

In this section, we prove our two lower bound results, modulo a communication complexity lower bound which is proved in \cref{sec:kpart,sec:polar,sec:spl}. We start by restating the results to be proved. Recall (from \cref{def:ow-and-pow}) the notion of a padded one-wise pair of distributions:  $(\cD_1,\cD_2)$ is a padded one-wise pair if there exist $\cD_0,\cD'_1,\cD'_2$ and $\tau \in [0,1]$ such that for every $i \in \{1,2\}, \cD'_i$ is one-wise independent, and $\cD_i = \tau \cD_0 + (1-\tau) \cD'_i$.

The first theorem we prove is the lower bound in the streaming setting for padded one-wise pairs of distributions. We restate the theorem below for convenience.

\restatethmmainnegative*

We also restate the lower bound against sketching algorithms from \cref{thm:main-detailed-dynamic} as a separate theorem below. 

\begin{theorem}[Lower bounds against sketching algorithms]\label{thm:main-negative-dynamic}
For every $q,k \in \N$, every family of functions $\cF \subseteq \{f:[q]^k \to \{0,1\}\}$ and for every $0 \leq \beta<\gamma \leq 1$, if $\kgyf \cap \kbnf \neq \emptyset$, then for every $\epsilon>0$, any sketching algorithm for the $(\gamma-\eps,\beta+\eps)$-$\maxF$ problem requires $\Omega(\sqrt{n})$ space.
Furthermore, if $\gamma = 1$, then any sketching algorithm for $(1,\beta+\epsilon)$-$\maxF$ requires $\Omega(\sqrt{n})$ space.
\end{theorem}

To prove both theorems, we introduce a new communication game we call the  \textit{Signal Detection (SD)} in \cref{sec:sd}. In \cref{thm:communication lb matching moments} we state a lower bound on the communication complexity of this problem. This lower bound is established in \cref{sec:kpart,sec:polar,sec:spl}. We then use this lower bound to prove \cref{thm:main-negative} in \cref{sec:insertionLB} and to prove \cref{thm:main-negative-dynamic} in \cref{sec:lb dynamic}.

{
\begin{figure}[ht]
    \centering
    \includegraphics[width=14cm]{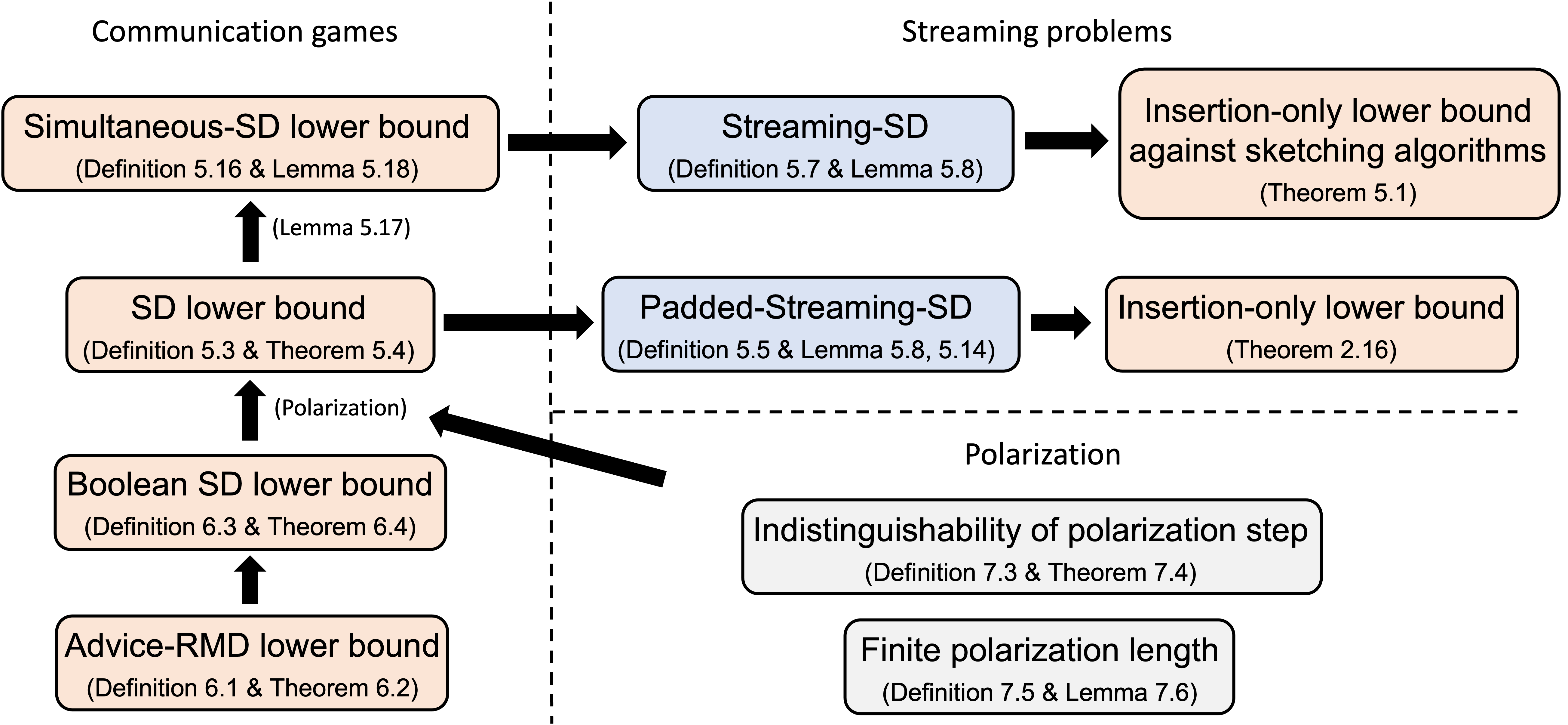}
    \caption{The roadmap of our lower bounds. The top two rows describe the results of this section, while the remaining rows describe notions and results from \cref{sec:kpart,sec:polar,sec:spl}.}
    \label{fig:outline}
\end{figure}

}

\subsection{The Signal Detection Problem and Results}\label{sec:sd}

In this section we introduce our communication game and state the lower bound for this game. 
We start with the definition of a general one-way communication game.

\begin{definition}[One-way communication game]
Given two distributions $\cY$ and $\cN$, an instance of the two-player one-way communication game is a pair $(X,Y)$ either drawn from $\cY$ or from $\cN$. Two computationally unbounded parties, Alice and Bob, receive $X$ and $Y$, respectively. A protocol $\Pi = (\Pi_A,\Pi_B)$ is a pair of functions with $\Pi_A(X) \in \{0,1\}^c$ denoting Alice's message to Bob, and $\Pi_B(\Pi_A(X),Y)\in \{\yes,\no\}$ denoting the protocol's output. We denote this output by $\Pi(X,Y)$. The complexity of this protocol is the parameter $c$ specifying the maximum length of Alice's message $\Pi_A(X)$. The advantage of the protocol $\Pi$ is the quantity 
\[
\left| \Pr_{(X,Y)\sim\cY} [ \Pi(X,Y) = \yes] - \Pr_{(X,Y)\sim \cN} [\Pi(X,Y) = \yes] \right| \,.
\]
\end{definition}

We now define the specific game we are interested in.

\begin{definition}[Signal Detection (SD) Problem]\label{def:sd}
Let $n,k,q\in\mathbb{N}, \alpha\in(0,1)$, where $k$, $q$ and $\alpha$ are constants with respect to $n$, and $\alpha n$ is an integer less than $n/k$. Let $\cF$ be a finite set.
For a pair $\cD_Y$ and $\cD_N$ of distributions over $\cF \times [q]^k$, we consider the following two-player one-way communication problem $(\cF,\cD_Y,\cD_N)$-\textsf{SD}.
\begin{itemize}
\item The generator samples the following objects:
\begin{enumerate}
    \item $\vecx^* \sim \textsf{Unif}([q]^n)$.
    \item $M \in \{0,1\}^{k\alpha n \times n}$ is chosen uniformly among all matrices with exactly one $1$ in each row and at most one $1$ in each column. We let $M = (M_1,\ldots,M_{\alpha n})$ where $M_i \in \{0,1\}^{k \times n}$ is the $i$-th block of rows of $M$, where each block has exactly $k$ rows.
    \item $\vecb = (\vecb(1),\ldots,\vecb(\alpha n))$ is sampled from one of the following distributions: 
\begin{itemize}
    \item (\yes) each $\vecb(i) = (f_i,\widetilde{\vecb}(i)) \in \cF \times [q]^k$ is sampled according to $\cD_Y$.
    \item (\no) each $\vecb(i) = (f_i,\widetilde{\vecb}(i)) \in \cF \times [q]^k$ is sampled according to $\cD_N$.
\end{itemize}
   \item $\vecz = (\vecz(1),\ldots,\vecz(\alpha n))$
   is determined from $M$, $\vecx^*$ and $\vecb= (\vecb(1),\dots,\vecb(\alpha n))$ as follows. Recall that $\vecb(i)=(f_i,\widetilde{\vecb}(i))$.
   We let $\vecz(i) = (f_i,\widetilde{z}_i) \in \cF \times \{0,1\}$ where $\widetilde{z}_i = 1$ iff $M_i \vecx^* = \widetilde{\vecb}(i)$.
\end{enumerate}
\item Alice receives $\vecx^*$ as input.
\item Bob receives $M$ and $\vecz$ as input.
\end{itemize}
\end{definition}

In the special case when the set $\cF$ contains just one element, $|\cF|=1$, we call the corresponding communication problem  $(\cD_Y,\cD_N)$-\textsf{SD}.

We note that our communication game is slightly different from those in previous works: Specifically the problem studied in \cite{GKKRW,KKS} is called the \textit{Boolean Hidden Matching (BHM)} problem from~\cite{GKKRW} and the works \cite{KKSV17,KK19} study a variant called the \textit{Implicit Hidden Partition} problem. While these problems are similar, they are less expressive than our formulation, and specifically do not seem to capture all $\maxf$ problems. 

There are two main differences between the previous settings and our setting. The first difference is the way to encode the matching matrix $M$. In all the previous works, each edge (or hyperedge) is encoded by a single row in $M$ where the corresponding columns are assigned to $1$, so that $m = \alpha n$. However, it turns out that this encoding hides too much information and hence we do not know how to reduce the problem to general \textsf{Max-CSP}. We unfold the encoding by using $k$ rows to encode a single $k$-hyperedge (leading to the setting of $m = k\alpha n$ in our case). The second difference is that we allow the masking vector $\vecb$ to be sampled from a more general distribution. This is also for the purpose of establishing a reduction to general \textsf{Max-CSP}. Due to the above two differences, it is not clear how to derive communication lower bounds for general $\cD_Y$ and $\cD_N$ by reduction from the previous works.

\begin{restatable}[Communication lower bound for $(\cF, \cD_Y,\cD_N)$-\textsf{SD}]{theorem}{reducermdtosd}\label{thm:communication lb matching moments}
For every $k,q$, every finite set $\cF$, every pair of distributions $\cD_Y,\cD_N \in \Delta(\cF\times[q]^k)$ with $\vecmu(\cD_Y) = \vecmu(\cD_N)$ there exists $\alpha_0 > 0$ such that for every $0<\alpha \leq \alpha_0$ and $\delta > 0$ there exists $\tau > 0$ such that the following holds: Every protocol for $(\cF, \cD_Y,\cD_N)$-SD achieving advantage $\delta$ on instances of length $n$ requires $\tau\sqrt{n}$ bits of communication.
\end{restatable}

\cref{sec:kpart,sec:spl,sec:polar} are devoted to proving 
\cref{thm:communication lb matching moments}. The specific proof can be found in \cref{ssec:proof-of-sd}. 
In the rest of this section we use this theorem to prove \cref{thm:main-negative,thm:main-negative-dynamic}.

\subsection{The streaming lower bound}\label{sec:insertionLB}

The hardness of \textsf{SD} suggests a natural path for hardness of $\maxF$ problems in the streaming setting. Such a reduction would take two distributions $\cD_Y \in S_\gamma^Y$ and $\cD_N \in S_\beta^N$ with matching marginals, construct  distributions $\cY$ and $\cN$ of \textsf{RMD}, and then interpret these distributions (in a natural way) as distributions over instances of $\maxf$ that are indistinguishable to small space algorithms. While the exact details of this ``interpretation'' need to be spelled out, every step in this path can be achieved. Unfortunately this does not mean any hardness for $\maxf$ since the CSPs generated by this reduction would consist of instances that have at most one constraint per variable, and such instances are easy to solve! 

To go from the instance suggested by the  \textsf{SD} problem to hard CSP instances, we instead pick $T$ samples (somewhat) independently from the distributions $\cY$ and $\cN$ suggested by the \textsf{SD} problem and concatenate these. With an appropriate implementation of this notion (see \cref{def:ssd}) it turns out it is possible to use the membership of the underlying distributions in $S^Y_\gamma$ and $S^N_\beta$ to argue that the resulting instances $\Psi$ do (almost always) have $\val_\Psi \geq \gamma$ or $\val_\Psi \leq \beta$. (We prove this after appropriate definitions in \cref{lem:csp value}.) But now one needs to connect the streaming problem generated from the $T$-fold sampled version to the SD problem.

To this end we formalize the $T$-fold streaming problem, which we call $(\cD_Y,\cD_N,T)$\textsf{-streaming-SD} problem, in \cref{def:ssd}. Unfortunately, we are not able to reduce the $(\cD_Y, \cD_N)$\textsf{-SD} problem to $(\cD_Y,\cD_N,T)$\textsf{-streaming-SD} problem for all 
$\cD_Y$ and $\cD_N$.\footnote{Roughly, this problem arises from the fact that the $T$ samples $(\vecx^{*}(t),M(t),\vecz(t))$ are not sampled independently from $\cY$ (or $\cN$ for $t \in [T]$). Instead they are sampled independently conditioned on $\vecx^*(1) = \cdots = \vecx^*(T)$. This hidden correlation in {\em both} the \yes\ and the \no\ cases turns out to be a serious problem.} But in the setting where $\cD_Y$ and $\cD_N$ have uniform marginals then we are able to effect the reduction and thus show that the streaming problem requires large space. This is a special case of \cref{lem:reduce to streaming,cor:space lb 1 wise} which we discuss next.

We are able to extend our reduction from \textsf{SD} to \textsf{streaming-SD} slightly beyond the uniform marginal case, to the case where $\cD_Y$ and $\cD_N$ form a padded one-wise pair, but both the streaming problem and the analysis of the resulting CSP value need to be altered to deal with this case, as elaborated next. Let $\tau \in [0,1]$ and $\cD_0, \cD'_Y, \cD'_N$ be such that for $i \in \{Y,N\}$ we have $\cD_i = \tau \cD_0+ (1-\tau)\cD'_i$ and $\cD'_i$ has uniform marginals. Our padded streaming problem, denoted $(\cD'_Y,\cD'_N,T,\cD_0,\tau)$\textsf{-padded-streaming-SD} problem, includes an appropriately large number of constraints generated according to $\cD_0$, followed by $T$ samples chosen according to the $(\cD'_Y,\cD'_N,T)$\textsf{-streaming-SD} problem. See \cref{def:ssd} for a formal definition. In \cref{lem:csp value} we show that the CSP value of the resulting streaming problem inherits the properties of $\cD_Y$ and $\cD_N$ (which is not as immediate for \textsf{padded-streaming-SD} as for \textsf{streaming-SD}). We then show effectively that $(\cD'_Y,\cD'_N)$-\textsf{SD} reduces to $(\cD'_Y,\cD'_N,T,\cD_0,\tau)$\textsf{-padded-streaming-SD}. See \cref{lem:reduce to streaming,cor:space lb 1 wise}. Putting these together leads to a proof of \cref{thm:main-negative}.

\subsubsection{The (Padded) Streaming SD Problem}

\begin{definition}[$(\cF,\cD_Y,\cD_N,T)$-\textsf{streaming-SD}]\label{def:ssd}

For $k,q,T\in\mathbb{N}$, $\alpha\in(0,1/k]$, a finite set $\cF$ and distributions $\cD_Y,\cD_N$ over $\cF\times[q]^k$, the streaming problem $(\cF,\cD_Y,\cD_N,T; \alpha,k,q)$-\textsf{streaming-SD} is the task of distinguishing, for every $n$, $\vecsigma \sim \cY_{\strm,n}$ from $\vecsigma \sim \cN_{\strm,n}$ where for a given length parameter $n$, the distributions $\cY_\strm = \cY_{\strm,n}$ and $\cN_\strm=\cN_{\strm,n}$ are defined as follows:
\begin{itemize}
    \item Let $\cY$ be the distribution over instances of length $n$, i.e., triples $(\vecx^*,M,\vecz)$, from the definition of $(\cF,\cD_Y,\cD_N)$-\textsf{SD}. For $\vecx \in [q]^n$, let $\cY|_\vecx$ denote the distribution $\cY$ conditioned on $\vecx^* = \vecx$. The stream $\vecsigma\sim\cY_\strm$ is sampled as follows: Sample $\vecx^*$ uniformly from $[q]^n$. Let $(M^{(1)},\vecz^{(1)}),\ldots,(M^{(T)},\vecz^{(T)})$ be sampled independently according to $\cY|_{\vecx^*}$. Let $\vecsigma^{(t)}$ be the pair $(M^{(t)},\vecz^{(t)})$ presented as a stream of edges with labels in $\cF\times\{0,1\}$, i.e., $\vecz^{(t)}=(f_i,\tilde{z}_i)$. Specifically for $t \in [T]$ and $i \in [\alpha n]$, let $\vecsigma^{(t)}(i) = (e^{(t)}(i),\vecz^{(t)}(i))$ where $e^{(t)}(i)$ is the $i$-th hyperedge of $M^{(t)}$, i.e., $e^{(t)}(i) = (j^{(t)}(k(i-1)+1),\ldots,j^{(t)}(k(i-1)+k)$ and $j^{(t)}(\ell)$ is the unique index $j$ such that $M^{(t)}_{j,\ell} = 1$. Finally we let $\vecsigma = \vecsigma^{(1)} \circ \cdots \circ \vecsigma^{(T)}$ be the concatenation of the $\vecsigma^{(t)}$s. 
    \item $\vecsigma\sim\cN_\strm$ is sampled similarly except we now sample $(M^{(1)},\vecz^{(1)}),\ldots,(M^{(T)},\vecz^{(T)})$  independently according to $\cN|_{\vecx^*}$ where $\cN|_\vecx$ is the distribution $\cN$ condition on $\vecx^* = \vecx$.
\end{itemize}

\end{definition}

Again when $\alpha,k,q$ are clear from context we suppress them and simply refer to the $(\cF,\cD_Y,\cD_N,T)$-\textsf{streaming-SD} problem. 
\begin{remark}\label{rem:uniform-ssd}

We note that when $\cD_N = \cD_\cF\times\textsf{Unif}([q]^k)$ for some $\cD_\cF\in\Delta(\cF)$, then the distributions $\cN|_{\vecx^*}$ are identical for all $\vecx^*$ (and the variables
$\vecz^{(t)}(i)$ are distributed as $\cD_\cF\times\bern(q^{-k})$ independently for every $t,i$). 
\end{remark}

For technical reasons, we need the following \textit{padded} version of \textsf{streaming-SD} to extend our lower bound techniques in the streaming setting beyond the setting of one-wise independent distributions.
\begin{definition}[$(\cF,\cD_Y,\cD_N,T,\cD_0,\tau)$-\textsf{padded-streaming-SD}]\label{def:pssd}

For $k,q,T\in\mathbb{N}$, $\alpha\in(0,1/k]$, $\tau\in[0,1)$, a finite set $\cF$, and distributions $\cD_Y,\cD_N,\cD_0$ over $\cF\times[q]^k$, the streaming problem\\ $(\cF, \cD_Y,\cD_N,T,\cD_0,\tau; \alpha,k,q)$-\textsf{padded-streaming-SD} is the task of distinguishing, for every $n$, $\vecsigma \sim \cY_{\pstrm,n}$ from $\vecsigma \sim \cN_{\pstrm,n}$ where for a given length parameter $n$, the distributions $\cY_\pstrm = \cY_{\pstrm,n}$ and $\cN_\pstrm=\cN_{\pstrm,n}$ are defined as follows: Sample $\vecx^*$ from $[q]^n$ uniformly. For each $i\in[\frac{\tau}{1-\tau}\alpha nT]$, uniformly sample a tuple $e^{(0)}(i)=(i_1,\dots,i_k)\in\binom{[n]}{k}$ and $(f_i,\vecb^{(0)}(i))\sim\cD_0$, let $\vecsigma^{(0)}(i)=(e^{(0)}(i),(f_i,\mathbf{1}_{\vecb^{(0)}(i)=\vecx^*|_{e^{(0)}(i)}}))$. Next, sample $\vecsigma^{(1)},\dots,\vecsigma^{(T)}$ according to the Yes and No distribution of $(\cF, \cD_Y,\cD_N,T)$-\textsf{streaming-SD} respectively. Finally, let $\vecsigma = \vecsigma^{(0)} \circ \cdots \circ \vecsigma^{(T)}$ be the concatenation of the $\vecsigma^{(t)}$s.

\end{definition}
Again when $\alpha,k,q$ are clear from context we suppress them and simply refer to the \\$(\cF,\cD_Y,\cD_N,T,\cD_0,\tau)$-\textsf{padded-streaming-SD} problem.
Note that when $\tau=0$, $(\cF,\cD_Y,\cD_N,T,\cD_0,\tau)$-\textsf{padded-streaming-SD} is the same as $(\cF,\cD_Y,\cD_N,T)$-\textsf{streaming-SD}.

\subsubsection{CSP value of \psSD}\label{sssec:csp value}

There is a natural way to convert instances of \psSD\ to instances of a $\maxF$\ problem. In this section we make this conversion explicit and show to use properties of the underlying distributions $\cD_0,\cD_Y,\cD_N$ to get bounds on the value of the instances produced.

Note that an instance $\vecsigma$ of \psSD\ is simply a sequence $(\sigma(1),\ldots,\sigma(\ell))$ 
where each $\sigma(i) = (\vecj(i),\vecz(i))$ with $\vecj(i) \in [n]^k$ and $\vecz(i)=(f_i,\tilde{z}_i) \in \cF\times\{0,1\}$. This sequence is already syntactically very close to the description of a $\maxF$ instance. 
Formally, we define an instance $\Psi(\vecsigma)$ of $\maxF$ as follows. For each $\sigma_i=(\vecj(i),\vecz(i))$ with $\vecz(i)=(f_i,\tilde{z}_i)$, if $\tilde{z}_i=1$ we add the constraint $f_i(\vecx|_{\vecj(i)})$ to $\Psi(\vecsigma)$; otherwise, we do not add any constraint to the formula.

In what follows we show that if $\cD_Y \in S_\gamma^Y$ then for all sufficiently large constant $T$, and sufficiently large $n$, if we draw $\vecsigma \sim \cY_{\pstrm,n}$, then with high probability, $\Psi(\vecsigma)$ has value at least $\gamma-o(1)$. Conversely if $\cD_N \in S_\beta^N$, then for all sufficiently large $n$, if we draw $\vecsigma \sim \cN_{\pstrm,n}$, then with high probability $\Psi(\vecsigma)$ has value at most $\beta+o(1)$.

\begin{lemma}[CSP value of \psSD]\label{lem:csp value}
For every $q,k\in\mathbb{N}$, $\cF \subseteq \{f:[q]^k \to \{0,1\}\}$,  $0 \leq \beta < \gamma \leq 1$, $\epsilon > 0$, $\tau=[0,1)$, and distributions $\cD_Y, \cD_N, \cD_0 \in \Delta(\{-1,1\}^k)$ there exists $\alpha_0$ such that for every $\alpha\in(0,\alpha_0]$ the following hold for every sufficiently large $T$: 
\begin{enumerate} 
\item If $\tau\cD_0+(1-\tau)\cD_Y \in S_\gamma^Y$, then for every sufficiently large $n$, the $(\cF, \cD_Y,\cD_N,T,\cD_0,\tau)$-\psSD\ \yes\ instance $\vecsigma \sim \cY_{\pstrm,n}$ satisfies $\Pr[\val_{\Psi(\vecsigma)} < (\gamma-\epsilon)] \leq \exp(-n)$.\footnote{In this lemma and proof we use $\exp(-n)$ to denote functions of the form $c^{-n}$ for some $c>1$ that does not depend on $n$ or $T$, but may depend on all other parameters including $q,k,\cD_Y,\cD_N,\cD_0,\beta,\gamma,\epsilon$.} 
\item If $\tau\cD_0+(1-\tau)\cD_N \in S_\beta^N$, then for every sufficiently large $n$, the $(\cF, \cD_Y,\cD_N,T,\cD_0,\tau)$-\psSD\ \no\ instance $\vecsigma \sim \cN_{\pstrm,n}$ satisfies $\Pr[\val_{\Psi(\vecsigma)} > (\beta+\epsilon)] \leq \exp(-n)$. 
\end{enumerate} 
Furthermore, if $\gamma = 1$ then $\Pr_{\vecsigma\sim\cY_{\pstrm,n}}\left[\val_{\Psi(\vecsigma)}=1\right]=1$. 
\end{lemma}

\begin{proof}
We assume $\eps \leq 1/2$ (and if not we prove the lemma for $\eps'=\frac12$ and this implies the lemma also for $\eps$). 
We prove the lemma for $\alpha_0 = \frac{\epsilon}{20kq^k}$ and $T_0 = 1000/(\epsilon^2\alpha)$. In what follows we set $\eta = \frac{\epsilon}{20kq^k}$. 

In what follows we let $N_0 = \frac{\tau \alpha n T}{1-\tau}$,  $N_t = \alpha n$ for $t \in [T]$ and $N = N_0 + TN_1$.
Recall that an instance of $(\cF, \cD_Y,\cD_N,T,\cD_0,\tau)$-\psSD\ consists of a stream $\vecsigma = \vecsigma^{(0)} \circ \cdots \circ \vecsigma^{(T)}$ where $\vecsigma^{(t)} = (\vecsigma^{(t)}(i) | i \in [N_t])$ and $\vecsigma^{(t)}(i)=(e^{(t)}(i),(f^{(t)}_{i},\tilde{z}^{(t)}(i))$ where $e^{(t)}(i)$ denotes a $k$-uniform hyperedge on $[n]$ and $f^{(t)}(i) \in \cF$ and $\tilde{z}^{(t)}(i) \in \{0,1\}$. 
Finally recall that $\vecsigma^{(t)} \sim \cY|_{\vecx^*}$ in the \yes\ case and $\vecsigma^{(t)} \sim \cN|_{\vecx^*}$ in the \no\ case independently for each $t$, where $\vecx^* \sim \textsf{Unif}([q]^n)$ is common across all $t$. 
We use 
$\cI = (\{0\} \times [T_0]) \cup ([T] \times [T_1])$ to denote the set of legal pairs of indices $(t,i)$. 
We let $m$ denote the total number of constraints in $\Psi(\vecsigma)$ with $m_t$ denoting the number of constraints from $\vecsigma^{(t)}$ for $0 \leq t\leq T$. (Note that $m$ and the $m_t$'s are random variables.)

For $\eta>0$, define $\vecx^*$ to be {\em $\eta$-good} if for every $\sigma \in [q]$, we have $|\{i \in [n] ~|~ \vecx^*_i = \sigma\}| \in [(1-\eta)\cdot\frac{n}q,(1+\eta)\cdot\frac{n}q]$. A straightforward application of Chernoff bounds shows that for every $\eta>0$ the vector $\vecx^*$ is $\eta$-good with  probability $1-\exp(-n)$. 

Below we condition on a good $\vecx^*$ and prove the following: (1) We show the expected value of $m$ is roughly $q^{-k}\cdot N$ and furthermore $m$ is sharply concentrated around its expectation. (2) In the \yes\ case we prove that the expected number of constraints satisfied by $\vecx^*$ is roughly at least $\gamma\cdot q^{-k} \cdot N$ and again this variable is sharply concentrated around its expectation. (3) In the \no\ case we prove that the expected number of constraints satisfied by any assignment is roughly at most $\beta\cdot q^{-k} \cdot N$ and again this variable is sharply concentrated around its expectation. We note that the sharp concentration part is essentially the same in all cases and it is bounding the expectations that is different in each case. That being said the analysis of the \no\ case does require sharper concentration since we need to take a union bound over all possible assignments.

\paragraph{Bounding the number of constraints.}
We start with step (1). Fix an $\eta$-good $\vecx^*$. Note that $m_t = \sum_{i \in [N_t]} \tilde{z}^{(t)}(i)$ for every $0\leq t\leq T$. 
We divide the analysis into two subparts. In step (1a) we bound $\mu \triangleq \Exp[\tilde{z}^{(t)}(i)]$ (in particular this expectation does not depend on $i$ or $t$). Note that $m = \sum_{t=0}^T \sum_{i \in [N_t]} \tilde{z}^{(t)}(i)$ and so bounding $\mu$ bounds $\Exp[m] = \mu\cdot N$.
Then in step (1b) we show that $m$ is concentrated around its expected value. 

For step (1a), let $p_{\sigma}$ denote the fraction of occurrences of the letter $\sigma$ in $\vecx^*$, i.e., $p_{\sigma} = \frac1n|\{i \in [n] | \vecx^*_i = \sigma\}|$. Note that given a sequence $\widetilde{\vecb}^{(t)}(i)=\veca\in[q]^k$, the probability that $\tilde{z}^{(t)}(i)=1$ over a random choice of $e^{(t)}(i)$ depends on $\veca$ as well as the $p_\sigma$'s. (Specifically this probability is $\prod_{j=1}^k p_{\veca_j} \pm O(k^2/n)$, where the additive correction term accounts for the sampling without replacement in the choice of $e^{(t)}(i)$.) However if the vector $\vecx^*$ is good, this dependence has little quantitative effect. In particular, if $\vecx^*$ is $\eta$-good, we have $\mu \in (\frac1q \pm \eta)^k\pm O(k^2/n)$ and thus we get $q^{-k} - 2 k \eta \leq \mu \leq q^{-k}+2k\eta$ provided $\eta \leq 1/(4kq)$ and $n$ is sufficiently large. This simplifies further to $\mu \in (1\pm\frac{\eps}{10})q^{-k}$ using $\eta \leq q^{-k}\eps/(20k)$.  Summing up over $(t,i)\in\cI$ we get $\Exp[m] \in (1\pm\frac{\eps}{10})q^{-k}N$. 

We now turn to step (1b), i.e., proving that $m$ is concentrated around its expectation. (In this part we work a little harder than necessary to prove that the failure probability is $\exp(-nT)$ rather than $\exp(-n)$. This is not necessary, but will be needed for the similar step in step (3).) 
Let $\tilde{Z}$ denote the set of random variables $\{\tilde{z}^{(t)}(i) | (t,i)\in \cI\}$ and for $(t,i)\in \cI$, let $\tilde{Z}_{-(t,i)} = \tilde{Z} \setminus \{\tilde{z}^{(t)}(i)\}$. We first show that for every $(t,i) \in \cI$ we have $\Exp[\tilde{z}^{(t)}(i)~|~\tilde{Z}_{-(t,i)}] \in (1\pm\frac\eps{10})\Exp[\tilde{z}^{(t)}(i)]$. 
Let $B_t$ denote the $t$-th block of variables, i.e., $B_t = \{\tilde{z}^{(t)}(i) | i \in [N_t]\}$. Now note that the only dependence among the $\tilde{z}^{(t)}(i)$'s is among the variables within a block while the blocks themselves are independent. Furthermore the variables in the block $B_0$ are independent of each other. Thus for $i\in[N_0]$ we have $\Exp[\tilde{z}^{(0)}(i)| Z_{-(0,i)}] = \Exp[\tilde{z}^{(0)}(i)]$.
For $t>0$, we have the variables from block $B_t$ may depend on each other due to the constraint that the underlying set of hyperedges are vertex disjoint. Fix $(t,i)\in \cI$ with $t>0$ and let $S$ be the set of variables touched by the hyperedges from block $B_t$, excluding $e^{(t)}(i)$. 
Now consider picking a hyperedge uniformly from $[n]$ and let $\psi$ be the probability that this hyperedge touches $S$. We clearly have $\psi \leq k|S|/n \leq k\alpha$. On the other hand, $\psi$ also upper bounds the difference between $\Exp[\tilde{z}^{(t)}(i)~|~\tilde{Z}_{-(t,i)}]$ and 
$\Exp[\tilde{z}^{(t)}(i)]$, so we have:
$$| \Exp[\tilde{z}^{(t)}(i)~|~\tilde{Z}_{-(t,i)}] - \Exp[\tilde{z}^{(t)}(i)] | \leq \psi \leq k \alpha \leq \frac{\eps q^{-k}}{20}\leq \frac{\eps}{10}\Exp[\tilde{z}^{(t)}(i)].$$
Applying \cref{lem:our-azuma} to the variables of $\tilde{Z}$ (arranged in some arbitrary order) we have 
$\Pr[m \not \in ((q^{-k} \cdot (1\pm\eps/10)^3] \leq \exp(-nT)$. Using $(1\pm\eps/10)^3 \subseteq (1\pm\eps/2)$ for $\eps <1$ we get:
\begin{equation}
    \Pr[m \not \in (1\pm\eps/2)\cdot q^{-k}N ] \leq \exp(-nT) \label{eqn:csp value m}
\end{equation}

\paragraph{Lower bounding the number of satisfied constraints in the \yes\ case.}
Let $Z^{(t)}(i)$ be the indicator variable for the event that the $i$-th element of $\vecsigma^{(t)}$ produces a constraint that is satisfied by $\vecx^*$, i.e.,  $Z^{(t)}(i) = \tilde{z}^{(t)}(i)\cdot f_i(\vecx^*|_{\vecj^{(t)}(i)})$. Note that the number of constraints satisfied by $\vecx^*$ is $\sum_{(t,i)\in\cI} Z^{(t)}(i)$. Note further that $Z^{(0)}(i)$'s are identically distributed across $i\in [N_0]$, and $Z^{(t)}(i)$ are also identically distributed across $t \in [T]$ and $i \in [N_0]$. By construction (see \cref{def:pssd}) we have $\Exp[Z^{(0)}(i)] = \Exp_{(f,\vecb)\sim \cD_0}[f(\vecb)\cdot \Exp_{\vecj}[\One(\vecx^*|_\vecj = \vecb)]]$. By the $\eta$-goodness of $\vecx^*$, we have that for every $\vecb \in [q]^k$, $\Exp_{\vecj}[\One(\vecx^*|_\vecj = \vecb)] \geq (1-\frac\eps{10})q^{-k}$. Thus we get $\Exp[Z^{(0)}(i)] \geq (1-\frac\eps{10})q^{-k}\cdot \Exp_{(f,\vecb)\sim \cD_0}[f(\vecb)]$. Similarly for $t>0$ we have $\Exp[Z^{(t)}(i)] = \Exp_{(f,\vecb)\sim \cD_Y}[f(\vecb)\cdot \Exp_{\vecj}[\One(\vecx^*|_\vecj = \vecb)]] \geq (1-\frac\eps{10})q^{-k}\cdot \Exp_{(f,\vecb)\sim \cD_Y}[f(\vecb)]$. Using linearity of expectations we now get
\begin{align*}
    \Exp\left[ \sum_{(t,i) \in \cI} Z^{(t)}(i) \right] & = N_0 \Exp[Z^{(0)}(1)] + T N_T  \Exp[Z^{(1)}(1)]\\
        & = N (\tau \Exp[Z^{(0)}(1)] + (1-\tau)  \Exp[Z^{(1)}(1)])\\
        & \geq \left(1-\frac{\eps}{10}\right) q^{-k} N \cdot (\tau \Exp_{(f,\vecb)\sim \cD_0}[f(\vecb)] 
             + (1-\tau) \Exp_{(f,\vecb)\sim \cD_Y}[f(\vecb)])\\
        & = \left(1-\frac{\eps}{10}\right) q^{-k} N \cdot \Exp_{(f,\vecb)\sim \tau\cD_0 + (1-\tau)\cD_Y}[f(\vecb)] \\
        & \geq \gamma \cdot \left(1-\frac{\eps}{10}\right) q^{-k} N,
\end{align*}
where the final inequality uses $\tau\cD_0 + (1-\tau)\cD_Y \in \sgyf$. 
The concentration can be analyzed exactly as in step (1b). In particular if we let $Z$ denote all variables $Z^{(t)}(i)$'s, then we have 
$\Exp[Z^{(t)}(i) | Z \setminus \{Z^{(t)}(i)\}] \geq \Exp[Z^{(t)}(i)] - \frac{\eps}{10}q^{-k}$. 
\begin{equation}
\Pr\left[\sum_{(t,i) \in \cI} Z^{(t)}(i) \leq (\gamma - \frac{3\eps}{10})\cdot  q^{-k} N \leq  \gamma\cdot (1-\frac{\eps}{10}) q^{-k} N - \frac{\eps}5 q^{-k} N \right] \leq \exp(-nT).
\label{eqn:csp value yes}
\end{equation}

\paragraph{Upper bounding the number of satisfiable constraints in the \no\ case.}
Fix an assignment $\vecnu \in [q]^k$ and consider the expected number of constraints satisfied by $\vecnu$. (We will later take a union bound over all $\vecnu$.) 
Let $W^{(t)}(i)$ be the indicator variable for the event that the $i$-th element of $\vecsigma^{(t)}$ produces a constraint that is satisfied by $\vecnu$, i.e.,  $W^{(t)}(i) = \tilde{z}^{(t)}(i)\cdot f_i(\vecnu|_{\vecj^{(t)}(i)})$. Note once again that $W^{(0)}(i)$'s are identically distributed across $i$ and $W^{(t)}(i)$ are identical across $t >0$ and $i$. Let $\mu_0 = \Exp[W^{(0)}(1)]$ and $\mu_N = \Exp[W^{(1)}(1)]$. Note that the expected number of satisfied constraints is $\Exp[\sum_{(t,i) \in \cI} W^{(t)}(i)] = N \cdot (\tau \mu_0 + (1-\tau) \mu_N)$, so we bound $\mu_0$ and $\mu_N$. By construction we have  
$$\mu_0 = \Exp_{(f,\vecb)\sim \cD_0, \vecj}[f(\vecnu|_\vecj)\cdot \One(\vecx^*|_\vecj = \vecb)] = \Exp_{(f,\vecb)\sim \cD_0, \vecj}[\One(\vecx^*|_\vecj = \vecb)]\cdot  \Exp_{(f,\vecb)\sim \cD_0, \vecj}[f(\vecnu|_\vecj) ~|~ \One(\vecx^*|_\vecj = \vecb)]$$ 
where $\vecj$ is a uniform random sequence of $k$ distinct elements of $[n]$. 
As argued earlier for every $\vecb$ we have $\Exp_{\vecj}[\One(\vecx^*|\vecj=b)] \leq (1+\frac{\eps}{10})q^{-k}$ for $\eta$-good $\vecx^*$. So we turn to bounding the second term. 

For $\sigma,\rho \in [q]$ let $\cP_{\sigma}(\rho)$ be the fraction of coordinates in $\vecnu$ that take the value $\rho$ among those coordinates where $\vecx^*$ is $\sigma$, i.e., $\cP_{\sigma}(\rho) = \frac{|\{i\in[n] | \vecnu_i = \rho ~\& ~ \vecx^*_i = \sigma\}|}{|\{i\in[n] | \vecx^*_i = \sigma\}|}$. Note that for every $\sigma$, $\cP_{\sigma}$ is a probability distribution in $\Delta(q)$. 
Furthermore, conditioning on $\vecx^*|_\vecj(\ell) = \vecb(\ell)$, the distribution of $\vecnu|_\vecj(\ell)$ is given by $\cP_{\vecb(\ell)}$. Thus the joint distribution of $\vecnu|_\vecj$ is $O(k^2/n)$-close in total variation distance to $\cP_{\vecb(1)} \times \cdots \times \cP_{\vecb(k)}$.
We thus have 
\begin{align*}
    \Exp_{(f,\vecb)\sim \cD_0, \vecj}[f(\vecnu|_\vecj) ~|~ \One(\vecx^*|_\vecj = \vecb)] 
         &\leq \Exp_{(f,\veca)\sim \cD_0}[\Exp_{\vecc, c_\ell \sim \cP_{a_\ell}}[f(\vecc)]]+O(k^2/n) \\
         & = \Exp_{(f,\veca)\sim \cD_0}[\Exp_{\vecd, d_{\ell,\sigma} \sim \cP_{\sigma}}[\cC(f,\veca)(\vecd)]]+O(k^2/n),
\end{align*}
where $\vecc \in [q]^k$ and $\vecd \in [q]^{k \times q}$. Note that the final expression is simply a change of notation applied to the middle expression above to make the expression syntactically closer to the notation in the definition of $\sbnf$. 
Combining with the bound on $\Exp_{\vecj}[\One(\vecx^*|\vecj=b)]$ above we get 
$$\mu_0 \leq  (1+\frac{\eps}{10}) q^{-k} \cdot (\Exp_{(f,\veca)\sim \cD_0}[\Exp_{\vecd, d_{\ell,\sigma} \sim \cP_{\sigma}}[\cC(f,\veca)(\vecd)]]) + O(k/n).$$
Similarly we get 
$$\mu_N \leq  (1+\frac{\eps}{10}) q^{-k}  \cdot (\Exp_{(f,\veca)\sim \cD_N}[\Exp_{\vecd, d_{\ell,\sigma} \sim \cP_{\sigma}}[\cC(f,\veca)(\vecd)]]) + O(k/n).$$
Now combining the two conditions above we get 
\begin{align*}
(\tau \mu_0 + (1-\tau) \mu_N) 
&\leq (1+\frac{\eps}{10}) q^{-k} \cdot \left(\Exp_{(f,\veca)\sim \tau \cD_0+(1-\tau)\cD_N}[\Exp_{\vecd, d_{\ell,\sigma} \sim \cP_{\sigma}}[\cC(f,\veca)(\vecd)]] \right)  +O(k^2/n)\\
& \leq \beta \cdot (1+\frac{\eps}{10}) q^{-k} + O(k^2/n) \\
& \leq \beta \cdot (1+\frac{\eps}{9}) q^{-k},
\end{align*}
where the final inequality uses the fact that $n$ is sufficiently large.
We thus conclude the the expected number of constraints satisfied by $\vecnu$ is at most $\beta \cdot (1+\frac{\eps}{9})q^{-k} N$. 
Concentration around the mean is now similar to before. In particular if we let $W$ denote the set of all $W^{(t)}(i)$'s then we still have If we  $\Exp[W^{(t)}(i) | W \setminus \{W^{(t)}(i)\}] \leq \Exp[W^{(t)}(i)]+k\alpha \leq \Exp[W^{(t)}(i)] + \frac{\eps}{10}q^{-k}N$, and so by \cref{lem:our-azuma} we get 
$$\Pr\left[\sum_{(t,i) \in \cI} W^{(t)}(i) \geq  (\beta+\frac{2\eps}9)\cdot q^{-k} N \geq \beta(1+\eps/9) q^{-k}N + \frac{\eps}{9} q^{-k}N\right] \leq \exp(-nT).$$
In particular by using $T$ sufficiently large, we get that the probability that  more than $(\beta+\frac{2\eps}9)\cdot q^{-k} N$ constraints are satisfied by $\vecnu$ is at most $c^{-n}$ for some $c > q$.
So by a union bound over all possible $\vecnu$'s we get the following:
\begin{equation}
    \label{eqn:csp value no} 
    \Pr\left[\exists \vecnu\in [q]^k \mbox{ s.t. }\vecnu \mbox{ satisfies more than } (\beta+\frac{2\eps}9)\cdot q^{-k} N \mbox{ constraints}\right] \leq \exp(-nT).
\end{equation}

\paragraph{Putting it together.}
Putting the above together we get that in the \yes\ case with probability $1 - \exp(-n)$ we have that $\vecx^*$ is good and the number of constraints is at most $(1+\frac\eps2)q^{-k}N$ (by \cref{eqn:csp value m}) while the number of satisfied constraints is at least $(\gamma - \frac{3\eps}{10})\cdot  q^{-k} N$ (by \cref{eqn:csp value yes}). 
Taking ratios we get 
$$\val_{\Psi(\vecsigma)} \geq
\frac{\gamma - \frac{3\eps}{10}}{1+\frac\eps2} \geq \gamma - \eps.
$$

Similarly in the \no\ case we have with probability at least $1-\exp(-n)$ we have that $\vecx^*$ is good, and the number of constraints is at least $(1-\frac\eps2)q^{-k}N$ (by \cref{eqn:csp value m}) while the number of satisfied constraints is at most $(\beta+\frac{2\eps}9)\cdot q^{-k} N$ (by \cref{eqn:csp value no}).
Taking ratios we get 
$$\val_{\Psi(\vecsigma)} \leq \frac{\beta+\frac{2\eps}9}{1-\frac\eps2}\leq \beta+\eps.$$
This proves the main part of the lemma.

The furthermore part follows from the fact that if $\gamma=1$ then every constraint in the \yes\ case is satisfied by $\vecx^*$. 

\end{proof}

\subsubsection{Reduction from one-way \texorpdfstring{$(\cD_Y,\cD_N)$-\SD}{(DY,DN)-SD)}\ to \texorpdfstring{$\psSD$}{padded-streaming-SD}}

We start by reducing \textsf{SD} to \textsf{padded-streaming-SD} in the special case where $\cD_N$ is ``uniform on the variables'' in the sense defined next.
We say a distribution $\cD \in \Delta(\cD \times [q]^k)$ is {\em uniform on the variables} if there exists a distribution $\cD_f \in \Delta(\cF)$ such that $\cD = \cD_f \times \textsf{Unif}([q]^k)$. 
The following lemma implies that in this special case \textsf{padded-streaming-SD} is hard. Since this holds for all one-wise independent distributions $\cD_Y$, by applying the lemma twice we get that \textsf{padded-streaming-SD} is hard for all one-wsie independent $\cD_Y$ and $\cD_N$.

\begin{lemma}\label{lem:reduce to streaming}
Let $\cF$ be a finite set, $T,q,k\in\mathbb{N}$, $\alpha \in (0,\alpha_0(k)]$, $\tau\in[0,1)$, and $\cD_Y,\cD_N,\cD_0\in\Delta(\cF\times[q]^k)$ with $\cD_Y$ being one-wise independent and $\cD_N=\cD_f\times\unif([q]^k)$ for some $\cD_f\in\Delta(\cF)$ and $\vecmu(\cD_Y)=\vecmu(\cD_N)$. Suppose there is a streaming algorithm $\ALG$ that solves $(\cF,\cD_Y,\cD_N,T,\cD_0,\tau)$-\textsf{padded-streaming-SD} on instances of length $n$ with advantage $\Delta$ and space $s$, then there is a one-way protocol for $(\cF,\cD_Y,\cD_N)$-\textsf{SD} on instances of length $n$ using at most $sT$ bits of communication achieving advantage at least $\Delta/T$.
\end{lemma}

The proof of~\cref{lem:reduce to streaming} is based on a hybrid argument (\textit{e.g.,}~\cite[Lemma 6.3]{KKS}). We provide a proof here based on the proof of \cite[Lemma 4.11]{CGV20}.

\begin{proof}[Proof of~\cref{lem:reduce to streaming}]

Note that since we are interested in distributional advantage, we can fix the randomness in $\ALG$ so that it becomes a deterministic algorithm. By an averaging argument the randomness can be chosen to ensure the advantage does not decrease. Let $\Gamma$ denote the evolution of function of $\ALG$ as it processes a block of edges. That is, if the algorithm is in state $s$ and receives a stream $\vecsigma$ then it ends in state $\Gamma(s,\vecsigma)$. Let $s_0$ denote its initial state. \\

We consider the following collection of (jointly distributed) random variables: Let $\vecx^* \sim \textsf{Unif}(\{-1,1\}^n)$. Denote $\cY=\cY_{\pstrm,n}$ and $\cN=\cN_{\pstrm,n}$. Let $(\vecsigma_Y^{(0)},\vecsigma_Y^{(1)},\ldots,\vecsigma_Y^{(T)}) \sim \cY|_{\vecx^*}$. Similarly, let $(\vecsigma_N^{(0)},\vecsigma_N^{(1)},\ldots,\vecsigma_N^{(T)}) \sim \cN|_{\vecx^*}$. Recall by Remark~\ref{rem:uniform-ssd} that since $\cD_N=\cD_f\times\unif([q]^k)$, we have $\cN|_{\vecx^*}$ is independent of $\vecx^*$, a feature that will be crucial to this proof. 

Let $S_t^Y$ denote the state of 
$\ALG$ after processing $\vecsigma_Y^{(0)},\ldots,\vecsigma_Y^{(t)}$, i.e., $S_0^Y = \Gamma(s_0,\vecsigma_Y^{(0)})$ and $S_t^Y = \Gamma(S_{t-1}^Y,\vecsigma_Y^{(t)})$ where $s_0$ is the fixed initial state (recall that $\ALG$ is deterministic). 
Similarly let $S_t^N$ denote the state of $\ALG$ after processing $\vecsigma_N^{(0)},\ldots,\vecsigma_N^{(t)}$. Note that since $\vecsigma_Y^{(0)}$ has the same distribution (conditioned on the same $\vecx^*$) as $\vecsigma_N^{(0)}$ by definition, we have $\|S_{0}^Y-S_{0}^N\|_{tvd}=0$.

Let $S^Y_{a:b}$ denote the sequence of states $(S_a^Y,\ldots,S_b^Y)$ and similarly for $S^N_{a:b}$. Now let $\Delta_t = \|S_{0:t}^Y-S_{0:t}^N\|_{tvd}$. 
Observe that $\Delta_0=0$ while $\Delta_T\geq\Delta$. (The latter is based on the fact that $\ALG$ distinguishes the two distributions with advantage $\Delta$.)
Thus $\Delta \leq \Delta_T - \Delta_0 = \sum_{t=0}^{T-1} (\Delta_{t+1} - \Delta_t)$ and so there exists $t^*\in\{0,1,\dots,T-1\}$ such that
\[
\Delta_{t^*+1} - \Delta_{t^*} = \|S_{0:t^*+1}^Y-S_{0:t^*+1}^N\|_{tvd}-\|S_{0:t^*}^Y-S_{0:t^*}^N\|_{tvd}\geq\frac{\Delta}{T} \, .
\]

Now consider the random variable $\tilde{S} = \Gamma(S_{t^*}^Y,\vecsigma_N^{(t^*+1)})$ (so the previous state is from the \yes\ distribution and the input is from the \no\ distribution). We claim below that $\| S_{t^*+1}^Y - \tilde{S} \|_{tvd} = \Exp_{A \sim_d S_{0:t^*}^Y}[\| S_{t^*+1}^Y|_{S_{0:t^*}^Y=A} - \tilde{S}|_{S_{0:t^*}^Y=A} \|_{tvd} ] \geq \Delta_{t^*+1} - \Delta_{t^*}$.
Once we have the claim, we show how to get a space $T\cdot s$ protocol for $(\cF,\cD_Y,\cD_n)$-\textsf{SD} with advantage $\Delta_{t^*+1} - \Delta_{t^*}$ concluding the proof of the lemma.

\begin{claim}
$\| S_{t^*+1}^Y - \tilde{S} \|_{tvd} \geq \Delta_{t^*+1} - \Delta_{t^*}$.
\end{claim}

\newcommand{\cpl}{\textrm{couple}} 

\begin{proof}
First, by triangle inequality for the total variation distance, we have
\[
\|S_{t^*+1}^Y - \tilde{S}\|_{tvd} \geq \|S_{t^*+1}^Y - S_{t^*+1}^N\|_{tvd} - \|\tilde{S} - S_{t^*+1}^N\|_{tvd} \, .
\]
Recall that $\tilde{S} = \Gamma(S_{t^*}^Y,\vecsigma_N^{(t^*+1)})$ and $S^N_{t^*+1} = \Gamma(S_{t^*}^N,\vecsigma_N^{(t^*+1)})$. Also, note that $\vecsigma_N^{(t^*+1)}$ follows the product distribution $(\cD_f\times\bern(q^{-k}))^{\alpha n}$ and in particular is independent of $S_{t^*}^Y$ and $S_{t^*}^N$. (This is where we rely crucially on the property $\cD_N = \cD_f\times\unif([q]^k)$.)
Furthermore $\Gamma$ is a deterministic function, and so we can apply the data processing inequality (Item (2) of~\cref{prop:tvd properties} with $X=S^Y_{t^*}$, $Y=S^N_{t^*}$, $W=\vecsigma_N^{(t^*+1)}$, and $f=\Gamma$) to conclude 
\[\|\tilde{S} - S_{t^*+1}^N\|_{tvd} = \|\Gamma(S_{t^*}^Y,\vecsigma_N^{(t^*+1)}) -  \Gamma(S_{t^*}^N,\vecsigma_N^{(t^*+1)})\|_{tvd}\leq\|S_{t^*}^Y-S_{t^*}^N\|_{tvd}.\]
Combining the two inequalities above we get 

\[
\|S_{t^*+1}^Y - \tilde{S}\|_{tvd} \geq \|S_{t^*+1}^Y - S_{t^*+1}^N\|_{tvd} - \|S_{t^*}^Y - S_{t^*}^N\|_{tvd}=\Delta_{t^*+1}-\Delta_{t^*}
\]
as desired.

\end{proof}

We now show how a protocol can be designed for $(\cF,\cD_Y,\cD_N)$-\textsf{SD} that achieves advantage at least $\theta = \Exp_{A \sim_d S_{0:t^*}^Y}[\| S_{t^*+1}^Y|_{S_{0:t^*}=A} - \tilde{S}|_{S_{0:t^*}=A} \|_{tvd} ] \geq \Delta_{t^*+1} - \Delta_{t^*}$ concluding the proof of the lemma. The protocol uses the distinguisher $T_A:\{0,1\}^{s} \to \{0,1\}$ such that $\Exp_{A,S_{t^*+1}^Y,\tilde{S}}[T_A(S_{t^*+1}^Y)] - \Exp[T_A(\tilde{S})] \geq \theta$ which is guaranteed to exist by the definition of total variation distance.

Our protocol works as follows: Let Alice receive input $\vecx^*$ and Bob receive inputs $(M,\vecz)$ sampled from either $\cY_{\SD}|_{\vecx^*}$ or $\cN_{\SD}|_{\vecx^*}$ where $\cY_{\SD}$ and $\cN_{\SD}$ are the Yes and No distribution of $(\cF,\cD_Y,\cD_N)$-$\SD$ respectively.
\begin{enumerate}
\item Alice samples $(\vecsigma^{(0)},\vecsigma^{(1)},\ldots,\vecsigma^{(T)}) \sim \cY|_{\vecx^*}$ and computes $A = S_{0:t^*}^Y \in \{0,1\}^{(t^*+1)s}$ and sends $A$ to Bob.
\item Bob extracts $S_{t^*}^Y$ from $A$, computes $\widehat{S} = \Gamma(S_{t^*}^Y,\vecsigma)$, where $\vecsigma$ is the encoding of $(M,\vecz)$ as a stream, and outputs \yes\ if $T_A(\widehat{S})=1$ and \no\ otherwise.
\end{enumerate}
Note that if $(M,\vecz)\sim \cY_{\SD}|_{\vecx^*}$ then $\widehat{S} \sim_d S_{t^*+1}^Y|_{S^Y_{0:t^*} = A}$ while if $(M,\vecz)\sim \cN_{\SD}|_{\vecx^*}$ then $\widehat{S} \sim \tilde{S}_{S^Y_{0:t^*} = A}$. 
It follows that the advantage of the protocol above exactly equals $\Exp_A[T_A(S_{t^+1}^Y)] - \Exp_A[T_A(\tilde{S})] \geq \theta \geq \Delta_{t^*+1}-\Delta_{t^*} \geq \Delta/T$.
This concludes the proof of the lemma.
\end{proof}

By combining~\cref{lem:reduce to streaming} with~\cref{thm:communication lb matching moments}, we immediately have the following consequence.

\begin{lemma}\label{cor:space lb 1 wise}
For $k \in \N$ let  $\alpha_0(k)$ be as given by \cref{thm:communication lb matching moments}.
Let $T\in\mathbb{N}$, $\alpha\in(0,\alpha_0(k)]$, $\tau\in[0,1)$, and $\cD_0, \cD_Y,\cD_N,  \in \Delta(\cF \times [q]^k)$ where $\cD_Y$ and $\cD_N$ are one-wise independent distributions with $\vecmu(\cD_Y)=\vecmu(\cD_N)$.

Then every streaming algorithm $\ALG$ solving $(\cF,\cD_Y,\cD_N,T,\cD_0,\tau)$-\textsf{padded-streaming-SD} in the streaming setting with advantage $1/8$ for all lengths $n$ uses space $\Omega(\sqrt{n})$.
\end{lemma}

\begin{proof}
Let $\ALG$ be an algorithm using space $s$ solving  $(\cF,\cD_Y,\cD_N,T,\cD_0,\tau)$-\textsf{padded-streaming-SD} with advantage $1/8$.
For $g \in \cF$, let $p_g = \Pr_{(f,\sigma)\sim \cD_Y}[f = g]$ and let $\cD_f$ be the distribution given by $\cD_f(g) = p_g$. Let
$\cD_M = \cD_f \times \textsf{Unif}([q]^k)$. Note that $\cD_M$ is uniform on the variables and satisfies $\vecmu(\cD_M) = \vecmu(\cD_Y) = \vecmu(\cD_N)$. Then by the triangle inequality $\ALG$ solves either the $(\cF,\cD_Y,\cD_M,T,\cD_0,\tau)$-\textsf{padded-streaming-SD} with advantage $1/16$ or it solves the\\ 
$(\cF,\cD_N,\cD_M,T,\cD_0,\tau)$-\textsf{padded-streaming-SD} with advantage $1/16$. Assume without loss of generality it is the former. 
Then by \cref{lem:reduce to streaming}, there exists a  one-way protocol for $(\cF,\cD_Y,\cD_M)$-\textsf{SD} using at most $sT$ bits of communication with advantage at least $1/(16T)$. Applying \cref{thm:communication lb matching moments} with $\delta = 1/(16T) > 0$, we now get that $s = \Omega(\sqrt{n})$. 
\end{proof}

\subsubsection{Proof of the streaming lower bound}\label{sec:lb insert}

We are now ready to prove \cref{thm:main-negative}.

\begin{proof}[Proof of \cref{thm:main-negative}]

We combine \cref{thm:communication lb matching moments}, \cref{cor:space lb 1 wise} and \cref{lem:csp value}. So in particular we set our parameters $\alpha$ and $T$ so that the conditions of these statements are satisfied. Specifically  
$k$ and $\epsilon>0$, let $\alpha^{(1)}_0$ be the constant from \cref{thm:communication lb matching moments}
and let $\alpha^{(2)}_0$ be the constant from \cref{lem:csp value}. Let $\alpha_0 = \min\{\alpha_0^{(1)},\alpha_0^{(2)}\}$, 
Given $\alpha \in (0,\alpha_0)$ let $T_0$ be the constant from \cref{lem:csp value} and let $T = T_0$. (Note that these choices allow for both \cref{thm:communication lb matching moments} and \cref{lem:csp value} to hold.)

Suppose there exists a streaming algorithm $\ALG$ that solves $(\gamma-\epsilon,\beta+\epsilon)$-$\maxF$. 
Let $\tau\in[0,1)$ and $\cD_Y, \cD_N,\cD_0$ be distributions such that (i) $\cD_Y$ and $\cD_N$ are one-wise independent, (ii) $\tau\cD_0+(1-\tau)\cD_Y\in S_\gamma^Y(\cF)$, and (iii) $\tau\cD_0+(1-\tau)\cD_N\in S_\beta^N(\cF)$.

Let $n$ be sufficiently large and let $\cY_{\strm,n}$ and $\cN_{\strm,n}$ denote the distributions of \yes\ and \no\ instances of $(\cF,\cD_Y,\cD_N,T,\cD_0,\tau)$-\psSD\ of length $n$. 
Since $\alpha$ and $T$ satisfy the conditions of \cref{lem:csp value}, we have for every sufficiently large $n$
\[
\Pr_{\vecsigma\sim\cY_{\strm,n}}\left[\val_{\Psi(\vecsigma)}<\left(\gamma-\epsilon\right)\right]=o(1)
\text{~~~and~~~}
\Pr_{\vecsigma\sim\cN_{\strm,n}}\left[\val_{\Psi(\vecsigma)}>\left(\beta+\epsilon\right)\right]=o(1) \, .
\]

We conclude that $\ALG$ can distinguish \yes\ instances of \textsf{Max-CSP}($\cF$)  from \no\ instances  with advantage at least $1/4-o(1)\geq 1/8$. 
However, since $\cD_Y,\cD_N$ and $\alpha$ satisfy the conditions of \cref{cor:space lb 1 wise} (in particular $\cD_Y$ and $\cD_N$ are one-wise independent and $\alpha \in (0,\alpha_0(k))$) such an algorithm requires space at least $\Omega(\sqrt{n})$. Thus, we conclude that any streaming algorithm that solves $(\gamma-\epsilon,\beta+\epsilon)$-\textsf{Max-CSP}($\cF$) requires $\Omega(\sqrt{n})$ space.

Finally, note that if $\gamma = 1$ then in \cref{lem:csp value}, we have $\val_\Psi=1$ with probability one. Repeating the above reasoning with this information, shows that $(1,\beta+\epsilon)-\maxF$ requires $\Omega(\sqrt{n})$-space.

\end{proof}

\subsection{The lower bound against sketching algorithms}\label{sec:lb dynamic}

In the absence of a reduction from \textsf{SD} to \textsf{streaming-SD} for general $\cD_Y$ and $\cD_N$, we turn to other means of using the hardness of \textsf{SD}. In particular, we use lower bounds on the communication complexity of a $T$-player communication game in the {\em simultaneous communication} setting --- one which is significantly easier to obtain lower bounds for than the one-way setting.
Below we describe a family of  $T$-player simultaneous communication games, which we call $(\cF,\cD_Y,\cD_N,T)$-\textsf{simultaneous-SD}.
(See \cref{def:simsd}.) We then show a simple reduction from $(\cF,\cD_Y,\cD_N)$-\textsf{SD} to $(\cF,\cD_Y,\cD_N,T)$-\textsf{simultaneous-SD}. Combining this reduction with our lower bounds on \textsf{SD} and the reduction from \textsf{simultaneous-SD} to streaming complexity leads to the proof of \cref{thm:main-negative-dynamic}.

\subsubsection{\texorpdfstring{$T$}{T}-Player Simultaneous Version of SD}\label{sec:T-SD}

In this section, we consider the complexity of \textit{$T$-player number-in-hand simultaneous message passing communication games} (abbrev. $T$-player simultaneous communication games). Such games are described by two distributions $\cY$ and $\cN$. An instance of the game is a $T$-tuple $(X^{(1)},\dots,X^{(T)})$ either drawn from $\cY$ or from $\cN$ and $X^{(t)}$ is given as input to the $t$-th player. A (simultaneous communication) protocol $\Pi = (\Pi^{(1)},\dots,\Pi^{(T)},\Pi_{\text{ref}})$ is a $(T+1)$-tuple of functions with $\Pi^{(t)}(X^{(t)}) \in \{0,1\}^c$ denoting the $t$-th player's message to the \textit{referee}, and $\Pi_{\text{ref}}(\Pi^{(1)}(X^{(1)}),\dots,\Pi^{(T)}(X^{(T)}))\in \{\yes,\no\}$ denoting the protocol's output. We denote this output by $\Pi(X^{(1)},\dots,X^{(T)})$. The complexity of this protocol is the parameter $c$ specifying the maximum length of $\Pi^{(1)}(X^{(1)}),\dots,\Pi^{(T)}(X^{(T)})$ (maximized over all $X$). The advantage of the protocol $\Pi$ is the quantity 
$$\left| \Pr_{(X^{(1)},\dots,X^{(T)})\sim\cY} [ \Pi(X^{(1)},\dots,X^{(T)}) = \yes] - \Pr_{(X^{(1)},\dots,X^{(T)})\sim \cN} [\Pi(X^{(1)},\dots,X^{(T)}) = \yes] \right|. $$

\begin{definition}[$(\cF,\cD_Y,\cD_N,T)$-\textsf{simultaneous-SD}]\label{def:simsd}

For $k,T\in\N$, $\alpha\in(0,1/k]$, a finite set $\cF$, 

distributions $\cD_Y,\cD_N$ over $\cF\times[q]^k$, the $(\cF,\cD_Y,\cD_N,T)$-\textsf{simultaneous-SD} is a $T$-player communication game given by a family of instances $(\cY_{\simul,n},\cN_{\simul,n})_{n\in\N,n\geq1/\alpha}$ where for a given $n$, $\cY=\cY_{\simul,n}$ and $\cN=\cN_{\simul,n}$ are as follows: Both $\cY$ and $\cN$ are supported on tuples $(\vecx^*,M^{(1)},\dots,M^{(T)},\vecz^{(1)},\dots,\vecz^{(T)})$ where $\vecx^*\in[q]^n$, $M^{(t)}\in\{0,1\}^{k\alpha n\times n}$, and $\vecz^{(t)}\in(\cF\times\{0,1\})^{k\alpha n}$, where the pair $(M^{(t)},\vecz^{(t)})$ are the $t$-th player's inputs for all $t\in[T]$. We now specify the distributions of $\vecx^*$, $M^{(t)}$, and $\vecz^{(t)}$ in $\cY$ and $\cN$:
\begin{itemize}
\item In both $\cY$ and $\cN$, $\vecx^*$ is distributed uniformly over $[q]^n$.
\item In both $\cY$ and $\cN$ the matrix $M^{(t)}\in\{0,1\}^{\alpha kn\times n}$ is chosen uniformly (and independently of $\vecx^*$) among matrices with exactly one $1$ per row and at most one $1$ per column.
\item The vector $\vecz^{(t)}$ is determined from $M^{(t)}$ and $\vecx^*$ as follows. Sample a random vector $\vecb^{(t)}\in (\cF\times[q]^k)^{\alpha kn}$ whose distribution differs in $\cY$ and $\cN$. Specifically, let $\vecb^{(t)}=(\vecb^{(t)}(1),\dots,\vecb^{(t)}(\alpha n))$ be sampled from one of the following distributions (independent of $\vecx^*$ and $M$):
\begin{itemize}
    \item $\cY$:  Each $\vecb^{(t)}(i)=(f_i,\tilde{\vecb}(i))\in\cF\times[q]^k$ is sampled independently according to $\cD_Y$.
    \item $\cN$:  Each $\vecb^{(t)}(i)=(f_i,\tilde{\vecb}(i))\in\cF\times[q]^k$ is sampled independently according to $\cD_N$.
\end{itemize}
We now set $\vecz^{(t)}=(f_i,\tilde{z}_i)$ where $\tilde{z}_i=1$ iff $ = (M^{(t)} \vecx^*)= \tilde{\vecb}^{(t)}(i)$.
\end{itemize}

If $\cF \subseteq\{f:[q]^k\to\{0,1\}\}$, then given an instance $\vecsigma = (\vecx^*,M^{(1)},\dots,M^{(T)},\vecz^{(1)},\dots,\vecz^{(T)})$,  we will let $\Psi(\vecsigma)$ represent the associated instance of $\maxF$ as described in~\cref{sssec:csp value}.

\end{definition}

Note that the instance $\Psi(\vecsigma)$ obtained in the \yes\ and \no\ cases of $(\cF,\cD_Y,\cD_N,T)$-\textsf{simultaneous-SD} are distributed exactly according to instances derived in the \yes\ and \no\ cases of\\ $(\cF,\cD_Y,\cD_N,T,\cD_0,\tau=0)$-\textsf{padded-streaming-SD} and thus \cref{lem:csp value} can still be applied to conclude that \yes\ instances usually satisfy $\val_{\Psi(\vecsigma)} \geq \gamma - o(1)$ and \no\ instances usually satisfy $\val_{\Psi(\vecsigma)} \leq \beta - o(1)$. We will use this property when proving \cref{thm:main-negative-dynamic}.

We start by showing the \textsf{simultaneous-SD} problems above do not have low-communication protocols when the marginals of $\cD_Y$ and $\cD_N$ match.

\begin{lemma}\label{lem:reduce SD to simul-SD}
Let $\cF$ be a finite set, $k,q,T\in\N$, $\cD_Y,\cD_N\in \Delta(\cF \times [q]^k)$, and let $\alpha\in(0,1/k]$.
Suppose there is a protocol $\Pi$ that solves $(\cF,\cD_Y,\cD_N,T)$-\textsf{simultaneous-SD} on instances of length $n$ with advantage $\Delta$ and space $s$, then there is a one-way protocol for $(\cF,\cD_Y,\cD_N)$-\textsf{SD} on instances of length $n$ using at most $s(T-1)$ bits of communication and achieving advantage at least $\Delta/T$.
\end{lemma}

\begin{proof} 

Let us first fix the randomness in $\Pi$ so that it becomes a deterministic protocol. Note that by an averaging argument the advantage of $\Pi$ does not decrease. Recall that $\cY$ and $\cN$ are Yes and No input distribution of $(\cF,\cD_Y,\cD_N,T)$-\textsf{simultaneous-SD} and we have
\[
\Pr_{X\sim\cY}[\Pi(X)=\yes]-\Pr_{X\sim\cN}[\Pi(X)=\yes]\geq\Delta \, .
\]
Now, we define the following distributions $\cD_0,\dots,\cD_{T}$. Let $\cD_0=\cY$ and $\cD_{T}=\cN$. For each $t\in[T-1]$, we define $\cD_t$ to be the distribution of input instances of $(\cF,\cD_Y,\cD_N,T)$-\textsf{simultaneous-SD} by sampling $\vecb^{(t')}(i)$ independently according to $\cD_Y$ (resp. $\cD_N$) for all $t'\leq t$ (resp. $t'>t$) and $i$ (see~\cref{def:simsd} to recall the definition).  Next, for each $t\in[T]$, let
\[
\Delta_t=\Pr_{X\sim\cD_t}[\Pi(X)=\yes]-\Pr_{X\sim\cD_{t-1}}[\Pi(X)=\yes] \, .
\]
Observe that $\sum_{t\in[T]}\Delta_t=\Delta$ and hence there exists $t^*\in[T]$ such that $\Delta_{t^*}\geq\Delta/T$.

Now, we describe a protocol $\Pi'$ for $(\cF,\cD_Y,\cD_N)$-\textsf{SD} as follows. On input $(\vecx^*,M,\vecz)$, Alice receives $\vecx^*$ and Bob receives $(M,\vecz)$. Alice first samples matrices $M^{(1)},\dots,M^{(t^*-1)},M^{(t^*+1)},\dots,M^{(T)}$ as the second item in~\cref{def:simsd}. Next, Alice samples $\vecb^{(t')}(i)=(f_i,\tilde{\vecb}^{(t')}(i))$ according to $\cD_Y$ (resp. $\cD_N$) for all $t'<t^*$ (resp. $t'>t^*$) and $i\in[\alpha nT]$ and sets $\vecz^{(t')}(i)=(f_i,\tilde{z}_i)$ as the third item in~\cref{def:simsd}. Note that Alice can do this because she possesses $\vecx^*$. Finally, Alice sends $\{\Pi^{(t')}(M^{(t')},\vecz^{(t')})\}_{t'\in[T]\backslash\{t^*\}}$ to Bob. After receiving Alice's message $(X^{(1)},\dots,X^{(t^*-1)},X^{(t^*+1)},\dots,X^{(T)})$, Bob computes $\Pi^{(t^*)}(M,\vecz)$ and outputs $\Pi'(M,\vecz)=\Pi_{\text{ref}}(X^{(1)},\dots,X^{(t^*-1)},\Pi^{(t^*)}(M,\vecz),X^{(t^*+1)},\dots,X^{(T)})$.

It is clear from the construction that the protocol $\Pi'$ uses at most $s(T-1)$ bits of communication. To see $\Pi'$ has advantage at least $\Delta/T$, note that if $(\vecx^*,M,\vecz)$ is sampled from the Yes distribution $\cY_{\textsf{SD}}$ of $(\cF,\cD_Y,\cD_N)$-\textsf{SD}, then $((M^{(1)},\vecz^{(1)}),\dots,(M^{(t^*-1)},\vecz^{(t^*-1)}),(M,\vecz),(M^{(t^*+1)},\vecz^{(t^*+1)}),\dots,(M^{(T)},\vecz^{(T)}))$ follows the distribution $\cD_{t^*}$. Similarly, if $(\vecx^*,M,\vecz)$ is sampled from the No distribution $\cN_{\textsf{SD}}$ of $(\cF,\cD_Y,\cD_N)$-\textsf{SD}, then $((M^{(1)},\vecz^{(1)}),\dots,(M^{(t^*-1)},\vecz^{(t^*-1)}),(M,\vecz),(M^{(t^*+1)},\vecz^{(t^*+1)}),\dots,(M^{(T)},\vecz^{(T)}))$ follows the distribution $\cD_{t^*-1}$. Thus, the advantage of $\Pi'$ is at least
\begin{align*}
&\Pr_{(M,\vecz)\sim\cY_{\textsf{SD}},\Pi'}[\Pi'(M,\vecz)=\yes]-\Pr_{(M,\vecz)\sim\cN_{\textsf{SD}},\Pi'}[\Pi'(M,\vecz)=\yes]\\
=\, &\Pr_{X\sim\cD_{t^*}}[\Pi(X)=\yes]-\Pr_{X\sim\cD_{t^*-1}}[\Pi(X)=\yes]=\Delta_{t^*}\geq\Delta/T \, .
\end{align*}
We conclude that there is a one-way protocol for $(\cF,\cD_Y,\cD_N)$-\textsf{SD} using at most $s(T-1)$ bits of communication achieving advantage at least $\Delta/T$.

\end{proof}

As an immediate consequence of~\cref{thm:communication lb matching moments} and~\cref{lem:reduce SD to simul-SD} we get that $(\cF,\cD_Y,\cD_N,T)$-\textsf{simultaneous-SD} requires $\Omega(\sqrt{n})$ bits of communication when the marginals of $\cD_Y$ and $\cD_N$ match.

\begin{lemma}\label{cor:communication lb matching moments sim}
For every $k,q\in\N$, there exists $\alpha_0 > 0$ such that for every $\alpha\in(0,\alpha_0)$ and $\delta>0$ the following holds: For every finite set $\cF$ and $T\in\N$ and every pair of distributions $\cD_Y,\cD_N\in \Delta(\cF \times [q]^k)$ with $\vecmu(\cD_Y) = \vecmu(\cD_N)$,
there exists $\tau > 0$ and $n_0$ such that for every $n \geq n_0$, every protocol for $(\cF,\cD_Y,\cD_N,T)$-\textsf{simultaneous-SD} achieving advantage $\delta$ on instances of length $n$ requires $\tau\sqrt{n}$ bits of communication. 
\end{lemma}

We are now ready to prove \cref{thm:main-negative-dynamic}.

\subsubsection{Proof of Theorem~\ref{thm:main-negative-dynamic}}

\begin{proof}[Proof of \cref{thm:main-negative-dynamic}]
The proof is a straightforward combination of \cref{lem:csp value} and \cref{cor:communication lb matching moments sim}
and so we pick parameters so that all these are applicable. 
Given $\epsilon$ and $k$, let $\alpha^{(1)}_0$ be as given by \cref{lem:csp value} and let $\alpha^{(2)}_0$ be as given by \cref{cor:communication lb matching moments sim}. Let $\alpha = \min\{\alpha^{(1)}_0,\alpha^{(2)}_0\}$. Given this choice of $\alpha$, let $T_0$ be as given by \cref{lem:csp value}. We set $T = T_0$ below. Let $n$ be sufficiently large. 

Throughout this proof we will be considering integer weighted instances of $\maxF$ on $n$ variables with constraints. Note that such an instance $\Psi$ can be viewed as a vector in $\Z^N$ where $N = O(|\cF|\times n^k)$ represents the number of possibly distinct constraints applications on $n$ variables. 
Let $\Gamma = \{\Psi | \val_\Psi \geq \gamma - \epsilon\}$. 
Let $B = \{\Psi | \val_\Psi \leq \beta + \epsilon\}$. 
Suppose there exists a sketching algorithm $\ALG_1$ that solves $(\gamma-\epsilon,\beta+\epsilon)$-\textsf{Max-CSP}($\cF$) using at most $s(n)$ bits of space. Note that $\ALG_1$ must achieve advantage at least $1/3$ on the problem $(\Gamma,B)$. 
By running several independent copies of $\ALG_1$ and thresholding appropriately, we can get an algorithm $\ALG_2$ with space $O(s)$ and advantage $1 - \frac{1}{100}$ solving $(\Gamma,B)$.

Now, let $\COMP$ and $\COMB$ be the compression and combination functions as given by this sketching algorithm (see~\cref{def:sketching alg}). We use these to design a protocol for $(\cF,\cD_Y,\cD_N,T)$-\textsf{simultaneous-SD} as follows.

Let $(M^{(t)},\vecz^{(t)})$ denote the input to the $t$-th player in $(\cF,\cD_Y,\cD_N,T)$-\textsf{simultaneous-SD}. Each player turn his/her inputs into $\Psi^{(t)}=(C^{(t)}_1,\dots,C^{(t)}_{m_t})$ 
where $C^{(t)}_i$ corresponds to the constraint $(\vecj^{(t)}(i),f_i^{(t)})$ with $\vecj^{(t)}_i\in[n]^k$ the indicator vector for the $i$-th hyperedge of $M^{(t)}$. Next, the players use shared randomness to compute the sketch of his/her input $\COMP(\Psi^{(t)})$ and send it to the referee. Finally, the referee computes the sketch for all streams $\COMB( \COMP(\Psi^{(1)}),\dots,\COMP(\Psi^{(T)}))$ and outputs the corresponding answer. 

To analyze the above, note that the communication is $O(s)$. Next, by the advantage of the sketching algorithm, we have that 
\begin{equation}
\min_{\Psi \in \Gamma}[\ALG_2(\Psi) = 1] - \max_{\Psi \in B}[\ALG_2(\Psi) = 1] \geq 1-12/100.
\label{eq:dyn-one}    
\end{equation} 
Now we consider what happens when $\Psi \sim \cY_{\simul,n}$ and $\Psi \sim \cN_{\simul,n}$. By \cref{lem:csp value} we have that $\Pr_{\Psi \sim \cY_{\simul,n}}[\Psi \in \Gamma] \geq  1 - o(1)$ and $\Pr_{\Psi \sim \cN_{\simul,n}}[\Psi \in B] \geq  1 - o(1)$. Combining with \cref{eq:dyn-one} we thus get
\[
\Pr_{\Psi \sim \cY_{\simul,n}}[\ALG_2(\Psi) = 1] - \Pr_{\Psi \sim \cN_{\simul,n}}[\ALG_2(\Psi) = 1] \geq 1 - 12/100 - o(1) \geq 1/2,
\]
We thus get an $O(s)$ simultaneous communication protocol for $(\cF,\cD_Y,\cD_N,T)$-\textsf{simultaneous-SD} with advantage at least $1/2$. 

Now we conclude by applying \cref{cor:communication lb matching moments sim} with $\delta = 1/2$ to get that $s = \Omega(\sqrt{n})/T = \Omega(\sqrt{n})$, thus yielding the theorem.

\end{proof}

\section{Hardness of Advice-Signal-Detection with Uniform Marginals}\label{sec:kpart}
The goal of this section is to prove a variant of \cref{thm:communication lb matching moments}  that will be used in  \cref{sec:polar} and \cref{sec:spl} for a proof of the general case of \cref{thm:communication lb matching moments}. Recall that in the $(\cD_Y,\cD_N)$-\textsf{SD} problem $|\cF|=1$, so we omit $\cF$. The main result of this section, presented in \cref{thm:kpart}, gives an $\Omega(\sqrt{n})$ lower bound on the communication complexity of $(\cD_Y,\cD_N)$-\textsf{SD} for distributions with matching marginals $\vecmu(\cD_Y)=\vecmu(\cD_N)$ for the case when (i) the alphabet is Boolean $\{-1,1\}$~\footnote{Throughout this section we use $\{-1,1\}$ to denote the Boolean domain.}, (ii) the marginals are uniform $\vecmu(\cD_Y)=\vecmu(\cD_N)=0^k$, but (iii) both players also receive a specific advice vector $\veca$. We define the corresponding Advice-SD communication game below.

In order to prove the hardness of Advice-SD, we first define the Randomized Mask Detection with advice (Advice-RMD) communication game, and prove an $\Omega(\sqrt{n})$ lower bound on the communication complexity of this game in \cref{thm:rmd}. The proof of the main result of this section, \cref{thm:kpart}, will then follow from the corresponding lower bounds for Advice-RMD in \cref{thm:rmd}.

\subsection{Hardness of \textsf{Advice-RMD}}
In this section we state a theorem that establishes hardness of \textsf{RMD} in the Boolean setting and with uniform marginals while allowing for advice. The proof of this theorem is postponed to \cref{sec:rmd_proof}. First we define the \textsf{Advice-RMD} one-way communication game.

\begin{definition}[\textsf{Advice-RMD}]
Let $n,k\in\mathbb{N}, \alpha\in(0,1)$, where $k$ and $\alpha$ are constants with respect to $n$, and $\alpha n$ is an integer less than $n/k$.
For a pair $\cD_Y$ and $\cD_N$ of distributions over $\{-1,1\}^k$, we consider the following two-player one-way communication problem $(\cD_Y,\cD_N)$-\textsf{Advice-RMD}.
\begin{itemize}
\item The generator samples the following objects:
\begin{enumerate}
    \item $\vecx^* \sim \textsf{Unif}(\{-1,1\}^n)$.
    \item $\Gamma \in S_n$ is chosen uniformly among all permutations of $n$ elements.
    \item We let $M \in \{0,1\}^{k\alpha n \times n}$ be a partial permutation matrix capturing $\Gamma^{-1}(j)$ for $j \in [k\alpha n]$. Specifically, $M_{ij} = 1$ if and only if $j = \Gamma(i)$. We view $M = (M_1,\ldots,M_{\alpha n})$ where each $M_i \in \{0,1\}^{k \times n}$ is a block of $k$ successive rows of $M$.
    \item $\vecb = (\vecb(1),\ldots,\vecb(\alpha n))$ is sampled from one of the following distributions: 
\begin{itemize}
    \item (\yes) each $\vecb(i) \in \{-1,1\}^k$ is sampled according to $\cD_Y$.
    \item (\no) each $\vecb(i)  \in \{-1,1\}^k$ is sampled according to $\cD_N$.
\end{itemize}
   
    \item $\vecz=M\vecx^* \odot \vecb$, where $\odot$ denotes the coordinate-wise product of the elements.
    \item Define a vector $\veca\in [k]^n$ as $a_j =i$ where $i = \Gamma^{-1}(j) \pmod k$ for every $j \in [n]$.
\end{enumerate}
\item Alice receives $\vecx^*$ and $\veca$ as input.
\item Bob receives $M$, $\vecz$, and $\veca$ as input.
\end{itemize}
\end{definition}

We follow the approach of~\cite{GKKRW} to prove the following theorem showing a $\Omega(\sqrt{n})$ communication lower bound for Boolean \textsf{Advice-RMD}. We postpone the proof to~\cref{sec:rmd_proof}.

\begin{theorem}[Communication lower bound for Boolean \textsf{Advice-RMD}]\label{thm:rmd}
For every $k\in\N$, and every pair of distributions $\cD_Y,\cD_N \in \Delta(\{-1,1\}^k)$ with uniform marginals $\vecmu(\cD_Y)=\vecmu(\cD_N)=0^k$ there exists $\alpha_0>0$ such that for every $\alpha\leq\alpha_0$ and $\delta > 0$ there exists $\tau>0$ such that every protocol for $(\cD_Y,\cD_N)$-\textsf{Advice-RMD} achieving advantage $\delta$ requires $\tau\sqrt{n}$ bits of communication on instances of length $n$. 
\end{theorem}

\subsection{Hardness of \textsf{Advice-SD}}
Let us first extend the definition of the Signal Detection (\textsf{SD}) problem to the following \textsf{Advice-SD} one-way communication game.
\begin{definition}[\textsf{Advice-SD}]\label{def:advice_sd}
Let $n,k,q\in\mathbb{N}, \alpha\in(0,1)$, where $k$, $q$ and $\alpha$ are constants with respect to $n$, and $\alpha n/k$ is an integer less than $n$. For a pair $\cD_Y$ and $\cD_N$ of distributions over~$[q]^k$, we consider the following two-player one-way communication problem $(\cD_Y,\cD_N)$-\textsf{Advice-SD}.
\begin{itemize}
\item The generator samples the following objects:
\begin{enumerate}
    \item $\vecx^* \sim \textsf{Unif}([q]^n)$.
    \item $\Gamma \in S_n$ is chosen uniformly among all permutations of $n$ elements.
    \item We let $M \in \{0,1\}^{k\alpha n \times n}$ be a partial permutation matrix capturing $\Gamma^{-1}(j)$ for $j \in [k\alpha n]$. Specifically, $M_{ij} = 1$ if and only if $j = \Gamma(i)$. We view $M = (M_1,\ldots,M_{\alpha n})$ where each $M_i \in \{0,1\}^{k \times n}$ is a block of $k$ successive rows of $M$.
    \item $\vecb = (\vecb(1),\ldots,\vecb(\alpha n))$ is sampled from one of the following distributions: 
\begin{itemize}
    \item (\yes) each $\vecb(i) \in [q]^k$ is sampled according to $\cD_Y$.
    \item (\no) each $\vecb(i)  \in [q]^k$ is sampled according to $\cD_N$.
\end{itemize}
   \item $\vecz = (z_1,\ldots,z_{\alpha n}) \in \{0,1\}^{\alpha n}$ 
   is determined from $M$, $\vecx^*$ and $\vecb$ as follows. 
   We let $z_i = 1$ if $M_i \vecx^* = \vecb(i)$, and $z_i = 0$ otherwise.
   \item Define a vector $\veca\in [k]^n$ as $a_j =i$ where $i = \Gamma^{-1}(j) \pmod k$ for every $j \in [n]$.
\end{enumerate}
\item Alice receives $\vecx^*$ and $\veca$ as input.
\item Bob receives $M$, $\vecz$, and $\veca$ as input.
\end{itemize}
\end{definition}
Almost immediately we get the following corollary for the \textsf{Advice-SD} problem from~\cref{thm:rmd}.

\begin{theorem}[Communication lower bound for Boolean \textsf{Advice-SD}]\label{thm:kpart}
For every $k\in\N$, and every pair of distributions $\cD_Y,\cD_N \in \Delta(\{-1,1\}^k)$ with uniform marginals $\vecmu(\cD_Y)=\vecmu(\cD_N)=0^k$ there exists $\alpha_0>0$ such that for every $0<\alpha\leq\alpha_0$ and $\delta > 0$ there exists $\tau > 0$, such that every protocol for $(\cD_Y,\cD_N)$-advice-SD achieving advantage $\delta$ requires $\tau\sqrt{n}$ bits of communication on instances of length $n$. 
\end{theorem}
\begin{proof}
We show that a~protocol achieving advantage $\delta$ in the $(\cD_Y,\cD_N)$-\textsf{Advice-SD} game with $s$ bits of communication implies a protocol achieving advantage $\delta$ for the $(\cD_Y,\cD_N)$-\textsf{Advice-RMD} game with $s$ bits of communication. Then the lower bounds of \cref{thm:rmd} for distributions with matching marginals will finish the proof.

Assume that there exists Bob's algorithm $\cB(M, \vecz, \veca, \text{Alice's message})$ that distinguishes $\vecb_i\sim\cD_Y$ and $\vecb_i\sim\cD_N$ with advantage $\delta$ in the \textsf{Advice-SD} game. For the \textsf{Advice-RMD} game, we keep the same algorithm for Alice, and modify Bob's algorithm as follows. Bob receives $M\in\{0,1\}^{k\alpha n \times n}, \vecz\in\{-1,1\}^{k\alpha n}, \veca$, and Alice's message, and partitions $\vecz=(\vecz_1,\ldots,\vecz_{\alpha n})$ where $\vecz_i\in\{-1,1\}^{k}$. For each $i\in[\alpha n]$, Bob computes $\widetilde{z}_i\in\{0,1\}$ as follows: $\widetilde{z}_i=1$ if and only if $\vecz_i=1^k$. Now Bob sets $\vecz'=(\widetilde{z}_1,\ldots,\widetilde{z}_{\alpha n})\in\{0,1\}^{\alpha n}$, and outputs $\cB(M, \vecz', \veca, \text{Alice's message})$. It is easy to see that in both $\yes$ and $\no$ cases, the distribution of the vectors $\vecz'$ computed by Bob is the distribution of vectors $\vecz$ sampled in the $(\cD_Y,\cD_N)$-\textsf{Advice-SD} game. Thus, the protocol achieves advantage $\delta$ for the $(\cD_Y,\cD_N)$-\textsf{Advice-SD} game using $s$ bits of communication as desired.
\end{proof}

\subsection{Proof of Theorem~\ref{thm:rmd}}\label{sec:rmd_proof}

Our proof of \cref{thm:rmd} follows the methodology of \cite{GKKRW} with some modifications as required by the Advice-RMD formulation. Their proof uses Fourier analysis to reduce the task of proving a communication lower bound to that of proving some combinatorial identities about randomly chosen matchings. We follow the same approach and this leads us to different conditions about randomly chosen hypermatchings which requires a fresh analysis in \cref{lem:step 2 1 wise}.

Without loss of generality in the following we assume that $n$ is a multiple of $k$. A~vector $\veca\in [k]^n$ is called an advice vector if for every $i\in[k], \, |\{j\colon a_j=i\}|=n/k$. 
For an advice vector $\veca\in [k]^n$, we say that a~partial permutation matrix $M \in \{0,1\}^{k\alpha n \times n}$ of a permutation $\Gamma$ is $\veca$-respecting if for every $i\in[k\alpha n]$ and $j\in[n]$, $M_{ij}=1$ if and only if $a_j =i \pmod k$. Intuitively, $\veca$ is the advice vector that tells you which congruence class $\Gamma(j)$ lies in.

For each advice vector $\veca\in [k]^n$, each $\veca$-respecting partial permutation matrix $M \in \{0,1\}^{k\alpha n \times n}$, distribution $\cD$ over $\{-1,1\}^k$, and a fixed Alice's message, the posterior distribution function $p_{M,\cD,\veca}:\{-1,1\}^{k \alpha n}\rightarrow[0,1]$ is defined as follows. For each $\vecz\in\{-1,1\}^{k \alpha n}$, let
\[
p_{M,\cD,\veca}(\vecz) := \Pr_{\substack{\vecx^*\in\{-1,1\}^n\\\vecb\sim\cD^{\alpha n}}}[\vecz=(M\vecx^*)\odot\vecb\ |\ M, \ \veca,\ \text{Alice's message}]=\Exp_{\vecx^*\in A}\Exp_{\vecb\sim\cD^{\alpha n}}[\mathbf{1}_{\vecz=(M \vecx^*)\odot \vecb}] \,,
\]
where $A\subset\{-1,1\}^n$ is the set of Alice's inputs that correspond to the message.

\begin{lemma}\label{lem:step 1 1 wise}
Let $\veca\in [k]^n$, $A\subseteq\{-1,1\}^n$, and $f:\{-1,1\}^n\rightarrow\{0,1\}$ be the indicator function of $A$. Let $k\in\mathbb{N}$ and $\alpha\in(0,1/100k)$. Let $\cD$ be a distribution over $\{-1,1\}^k$ such that $\mathbb{E}_{\veca\sim\cD}[a_j]=0$ for all $j\in[k]$.
\[
\Exp_{\substack{M\\M \text{is } \veca\text{-resp.}}}[\|p_{M,\cD,\veca}-U\|_{tvd}^2]\leq\frac{2^{2n}}{|A|^2}\sum_{\ell\geq2}^{k \alpha n}h(\ell)\cdot\sum_{\substack{\vecv\in\{0,1\}^n\\|\vecv|=\ell}}\widehat{f}(\vecv)^2 \, ,
\]
where $U\sim\unif(\{-1,1\}^{k\alpha n})$ and for each $\ell\in[n]$,
\[
h(\ell)=\max_{\substack{\vecv_\ell\in\{0,1\}^n\\|\vecv_\ell|=\ell}}  \Pr_{\substack{M\\M \text{is } \veca\text{-resp.}}}\left[\exists \vecs\in\{0,1\}^{k\alpha n}\backslash\{0^{k \alpha n}\},\ |\vecs(i)|\neq1\, \forall i,\ M^\top \vecs=\vecv_\ell\right] \, .
\]
Here for a vector $\vecs\in\{0,1\}^{k\alpha n}$ and integer $i\in[\alpha n]$, $\vecs(i)\in\{0,1\}^k$ denotes the $i$-th group of $k$ coordinates of~$\vecs$.
\end{lemma}
\begin{proof}
Observe that
\[
\|p_{M,\cD,\veca}-U\|_2^2=\sum_{\vecs\in\{0,1\}^{k\alpha n}}\left(\widehat{p}_{M,\cD,\veca}(\vecs)-\widehat{U}(\vecs)\right)^2=\sum_{\vecs\in\{0,1\}^{k\alpha n}\backslash\{0^{k\alpha n}\}}\widehat{p}_{M,\cD,\veca}(\vecs)^2 \,.
\]
Now by the Cauchy–Schwarz inequality we have that
\begin{align}
\Exp_{\substack{M\\M \text{is } \veca\text{-resp.}}}\left[\|p_{M,\cD,\veca}-U\|_{tvd}^2\right]&\leq2^{2k\alpha n}\Exp_{\substack{M\\M \text{is } \veca\text{-resp.}}}\left[\|p_{M,\cD,\veca}-U\|_2^2\right]\nonumber\\
&=2^{2k\alpha n}\Exp_{\substack{M\\M \text{is } \veca\text{-resp.}}}\left[\sum_{\vecs\in\{0,1\}^{k\alpha n}\backslash\{0^{k\alpha n}\}}\widehat{p_{M,\cD,\veca}}(\vecs)^2\right] \, . \label{eq:upper bound tvd 1 wise}
\end{align}

The following claim shows that the Fourier coefficients of the posterior distribution $p_{M,\cD,\veca}$ can be bounded from above by a certain Fourier coefficient of the indicator function $f$. Let's define $\GOOD:=\{\vecs\in\{0,1\}^{k\alpha n}\, |\, |\vecs(i)|\neq1\ \forall i\}$.

\begin{claim}\label{claim:1 wise fourier bound}
\[
\Exp_{\substack{M\\M \text{is } \veca\text{-resp.}}}[\|p_{M,\cD,\veca}-U\|_{tvd}^2]\leq\frac{2^{2n}}{|A|^2}\sum_{\vecs\in\GOOD\backslash\{0^{k\alpha n}\}}\Exp_{\substack{M\\M \text{is } \veca\text{-resp.}}}\left[\widehat{f}(M^\top \vecs)^2\right] \, . 
\]
\end{claim}
\begin{proof}
Observe that
\begin{align*}
\widehat{p_{M,\cD,\veca}}(\vecs)&= \frac{1}{2^{k\alpha n}}\sum_{\vecz\in\{-1,1\}^{k\alpha n}}p_{M,\cD,\veca}(\vecz)\prod_{\substack{i\in[\alpha n],j\in[k]\\s(i)_{j}=1}}z(i)_{j} \, .
\intertext{Recall that $p_{M,\cD,\veca}(\vecz)=\Exp_{\vecx^*\in A}\Exp_{\vecb\sim\cD^{\alpha n}}[\mathbf{1}_{\vecz=M \vecx^*\odot \vecb}]$, the equation becomes}
&=\frac{1}{2^{k\alpha n}}\cdot\Exp_{\vecx^*\in A}\left[\prod_{\substack{i\in[\alpha n],j\in[k]\\s(i)_{j}=1}}(M \vecx^*)_{i,j}\right]\Exp_{\vecb\sim\cD^{\alpha n}}\left[\prod_{\substack{i\in[\alpha n],j\in[k]\\s(i)_{j}=1}}b(i)_{j}\right] \, .
\intertext{Since $\Exp_{\veca\sim\cD}[a_j]=0$ for all $j\in[k]$, the right most sum is $0$ if there exists $i$ such that $|\vecs(i)|=1$. This equation becomes}
&\leq\frac{1}{2^{k\alpha n}}\cdot\left|\Exp_{\vecx^*\in A}\left[\prod_{\substack{i\in[\alpha n],j\in[k]\\s(i)_{j}=1}}(M \vecx^*)_{i,j}\right]\right|\cdot\mathbf{1}_{\vecs\in\GOOD} \, .
\intertext{Note that as each row and column of $M$ has at most $1$ non-zero entry, we have}
&=\frac{1}{2^{k\alpha n}}\cdot\left|\Exp_{\vecx^*\in A}\left[\prod_{\substack{i\in[n]\\(M^\top \vecs)_{i}=1}} \vecx^*_{i}\right]\right|\cdot\mathbf{1}_{\vecs\in\GOOD}
\end{align*}

Now we relate the above quantity to the Fourier coefficients of $f$. Recall that $f$ is the indicator function of the set $A$ and hence for each $\vecv\in\{0,1\}^n$, we have
\[
\widehat{f}(\vecv)=\frac{1}{2^n}\sum_{\vecx^*}f(\vecx^*)\prod_{i\in[n]:v_i=1}\vecx^*_i=\frac{1}{2^n}\sum_{\vecx^*\in A}\prod_{i\in[n]:v_i=1}\vecx^*_i \, .
\]
Thus, the Fourier coefficient of $p_M$ can be bounded as follows.
\begin{equation}\label{eq:1 wise fourier coefficient}
\widehat{p_{M,\cD,\veca}}(\vecs)\leq\frac{1}{2^{\alpha kn}}\cdot\frac{2^n}{|A|}\left|\widehat{f}(M^\top \vecs)\right|\cdot\mathbf{1}_{\vecs\in\GOOD} \, .
\end{equation}
By plugging~\cref{eq:1 wise fourier coefficient} into~\cref{eq:upper bound tvd 1 wise}, we have the desired bound and complete the proof of~\cref{claim:1 wise fourier bound}.
\end{proof}
Next, by~\cref{claim:1 wise fourier bound}, we have
\begin{align*}
\Exp_{\substack{M\\M \text{is } \veca\text{-resp.}}}[\|p_{M,\cD,\veca}-U\|_{tvd}^2]&\leq\frac{2^{2n}}{|A|^2}\sum_{\vecs\in\GOOD\backslash\{0^{\alpha kn}\}}\Exp_{\substack{M\\M \text{is } \veca\text{-resp.}}}\left[\widehat{f}(M^\top \vecs)^2\right] \, . 
\intertext{Since for a fixed $M$, the map $M^\top$ is injective, the right hand side of the above inequality has the following combinatorial form.}
&=\frac{2^{2n}}{|A|^2}\sum_{\vecv\in\{0,1\}^n\backslash\{0^n\}}\Pr_{\substack{M\\M \text{is } \veca\text{-resp.}}}\left[\exists \vecs\in\GOOD\backslash\{0^{k\alpha n}\},\ M^\top \vecs=\vecv\right]\widehat{f}(\vecv)^2 \, . 
\intertext{By symmetry, the above probability term will be the same for $\vecv$ and $\vecv'$ having the same Hamming weight. Recall that $$h(\ell)=\max_{\substack{\vecv_\ell\in\{0,1\}^n\\|\vecv_\ell|=\ell}}  \Pr_{\substack{M\\M \text{is } \veca\text{-resp.}}}\left[\exists \vecs\in\GOOD\backslash\{0^{k\alpha n}\},\ M^\top \vecs=\vecv_\ell\right] \,,$$ this equation becomes}
&\leq\frac{2^{2n}}{|A|^2}\sum_{\ell\geq1}^nh(\ell)\cdot\sum_{\substack{\vecv\in\{0,1\}^n\\|\vecv|=\ell}}\widehat{f}(\vecv)^2 \, .
\intertext{
Note that for $\ell=1$ and every $\ell>\alpha kn$, $h(\ell)=0$ by definition. Thus, this expression simplifies to the following.}
&=\frac{2^{2n}}{|A|^2}\sum_{\ell\geq2}^{\alpha kn}h(\ell)\cdot\sum_{\substack{\vecv\in\{0,1\}^n\\|\vecv|=\ell}}\widehat{f}(\vecv)^2 \, .
\end{align*}
This completes the proof of~\cref{lem:step 1 1 wise}.
\end{proof}

Now we bound from above the combinatorial quantity $h(\ell)$ from \cref{lem:step 1 1 wise}.
\begin{lemma}\label{lem:step 2 1 wise}
For every $0<\alpha\in(0,1/100k^2)$ and $\ell\in[k \alpha n]$, we have
\[
h(\ell)=\max_{\substack{\vecv_\ell\in\{0,1\}^n\\|\vecv_\ell|=\ell}} 
\Pr_{\substack{M\\M \text{is } \veca\text{-resp.}}}\left[\exists \vecs\neq0,\ |\vecs(i)|\neq1\ \forall i,\ M^\top \vecs=\vecv_\ell\right]
\leq\left(\frac{\ell}{n}\right)^{\ell/2} (e^3\alpha k^5)^{\ell/2} \,.
\]
\end{lemma}
\begin{proof}
By symmetry, without loss of generality we can fix the advice vector $\veca=(1^{n/k}2^{n/k}\ldots k^{n/k})$. For non-negative integers $\ell_1,\ldots,\ell_k$, we say that $\vecv_\ell\in\{0,1\}^n$ is an $(\ell_1,\ldots,\ell_k)$-vector if for every $i\in[k]$, $\vecv$ has exactly $\ell_i$ entries equal~$1$ in the $i$th group of $n/k$ coordinates. For fixed values of $\ell_i$,
let us define 
\[
h(\ell_1,\ldots,\ell_k)=\Pr_{\substack{M\\M \text{is } \veca\text{-resp.}}}\left[\exists \vecs\neq0,\ |\vecs(i)|\neq1\ \forall i,\ M^\top \vecs \text{ is a } (\ell_1,\ldots,\ell_k) \text{-vector}\right] \, .
\]
We note that 
\begin{align}\label{eq:hell}
h(\ell)=\max_{\substack{\ell_1,\ldots,\ell_k\geq0\\\sum_i \ell_i=\ell}} h(\ell_1,\ldots,\ell_k) \,.
\end{align}

An equivalent way to compute the probability $h(\ell_1,\ldots,\ell_k)$ is to fix the matching $M=\{(i, n/k+i,\ldots,(k-1)n/k+i)| i\in[\alpha n] \}$, and to let $\vecv$ be a random $(\ell_1,\ldots,\ell_k)$-vector . Then 
\begin{align}\label{eq:hbound}
h(\ell_1,\ldots,\ell_k)
=
\Pr_{\substack{\vecv\\ \vecv \text{ is } (\ell_1,\ldots,\ell_k)}}\left[\exists \vecs\neq0,\ |\vecs(i)|\neq1\ \forall i,\ M^\top \vecs=\vecv\right] 
=\frac{|U|}{|V|}
\, ,
\end{align}
where $V\subseteq\{0,1\}^n$ is the set of all $(\ell_1,\ldots,\ell_k)$-vectors, and $U =\{\vecu\in V\colon \exists \vecs\neq0,\ |\vecs(i)|\neq1\ \forall i,\ M^\top \vecs=\vecu\}$. From $\ell_1+\ldots+\ell_k=\ell$, the number of $(\ell_1,\ldots,\ell_k)$-vectors is
\begin{align}\label{eq:vbound}
|V|=\prod_{i=1}^k \binom{n/k}{\ell_i} \geq \binom{n/k}{\sum_{i=1}^k{\ell_i}}=\binom{n/k}{\ell} \geq \left(\frac{n}{k\ell}\right)^\ell \,,
\end{align}
where the first inequality uses that $n/k\geq k\alpha n\geq \ell$ for $\alpha\leq1/k^2$.

For a vector $\vecs\in\{0,1\}^{k\alpha n}$, let $T_\vecs=\{i\colon |\vecs(i)|>0\}$ be the set of indices of non-zero blocks of $\vecs$. In order to give an upper bound on the size of $U$, first we pick a set $T_\vecs$, and then we choose a vector $\vecu$ such that $M^\top \vecs=\vecu$ for some $\vecs$ corresponding to the set $T_\vecs$. Note that since for each $i\in T$, $\vecs(i)>0$ and $\vecs(i)\neq1$ by the definition of $h(\ell)$, the size of $t=|T|\leq k/2$. For every $t$, the number of ways to choose $T_\vecs$ is $\binom{\alpha n}{t}$. For a fixed $T_\vecs$, it remains to choose the $\ell$ coordinates of $\vecu$ among at most $kt$ non-zero coordinates of $\vecs$. 
For a vector $\vecs\in\{0,1\}^{k\alpha n}$, let $T_\vecs=\{i\in[\alpha n]\colon |\vecs(i)|>0\}$ be the set of indices of non-zero blocks of $\vecs$. In order to give an upper bound on the size of $U$, first we pick a set $T$, and then we choose a vector $\vecu$ such that $M^\top \vecs=\vecu$ for some $\vecs$ with (i) $|\vecs(i)\neq1|$ for all $i$ and (ii) $T_\vecs=T$. Note that since for each $i\in T$, $\vecs(i)>0$ and $|\vecs(i)|\neq1$, the size of $t=|T|\leq \ell/2$. For every $t$, the number of ways to choose $T$ is $\binom{\alpha n}{t}$. For a fixed $T$, it remains to choose the $\ell$ coordinates of $\vecu$ among at most $kt$ non-zero coordinates of $\vecs$.
This gives us the following upper bound on the size of $|U|$.
\begin{align}\label{eq:boundingu}
|U|
\leq \max_{t\leq \ell/2} \binom{\alpha n}{t}\binom{kt}{\ell} \,.
\end{align}
The second term of the upper bound in \cref{eq:boundingu} can be bounded from above by \[\binom{kt}{\ell} \leq \left(\frac{ekt}{\ell}\right)^{\ell} \leq \left(\frac{ek\ell/2}{\ell}\right)^{\ell}=\left(\frac{ek}{2}\right)^{\ell}\,. \]

Now we'll show that the first term of the upper bound in \cref{eq:boundingu} can be bounded from above by 
$\left(\frac{2ek\alpha n}{\ell}\right)^{\ell/2}
$.
If $\ell\geq 2\alpha n$, then
\[
\binom{\alpha n}{t} \leq 2^{\alpha n}\leq 2^{\ell/2} \leq \left(\frac{2ek\alpha  n}{\ell}\right)^{\ell/2} \,,
\]
where in the last inequality we use $\ell\leq k \alpha n$.
If $\ell <2\alpha n$, then $t\leq \ell/2 < \alpha n$, and
\[
\binom{\alpha n}{t} \leq \left(\frac{e\alpha n}{t}\right)^t \leq \left(\frac{2e\alpha n}{\ell}\right)^{\ell/2}<
\left(\frac{2ek \alpha n}{\ell}\right)^{\ell/2} \,.
\]

The above implies that
\begin{align}\label{eq:ubound}
|U|
\leq \max_{t\leq \min\{\alpha n, \ell/2\}} \binom{\alpha n}{t}\binom{kt}{\ell}
\leq \left(\frac{ek}{2}\right)^{\ell} \left(\frac{2ek \alpha n}{\ell}\right)^{\ell/2} 
\leq \left(\frac{n}{\ell}\right)^{\ell/2} (e^3\alpha k^3)^{\ell/2} \,.
\end{align}
Finally, from \cref{eq:hell,eq:hbound,eq:vbound,eq:ubound},
\[
h(\ell)=\max_{\substack{\ell_1,\ldots,\ell_k\geq0\\\sum_i \ell_i=\ell}} h(\ell_1,\ldots,\ell_k) 
=\frac{|U|}{|V|} 
\leq
\left(\frac{k\ell}{n}\right)^\ell \cdot \left(\frac{n}{\ell}\right)^{\ell/2} (e^3\alpha k^3)^{\ell/2}
\leq
\left(\frac{\ell}{n}\right)^{\ell/2} (e^3\alpha k^5)^{\ell/2}
\,.
\]
\end{proof}

In \cref{lem:distance_p_to_u} below we give the final ingredient needed for the proof of \cref{thm:rmd}. If $U$ is the uniform distribution over $\{-1,1\}^k$, then we show that for every large set $A\subseteq\{0,1\}^n$ of inputs $x$ corresponding to a fixed Alice's message (and a fixed advice $\veca$), $\Exp_{M,M \text{is } \veca\text{-resp.}}[\|p_{M,\cD,\veca}-U\|_{tvd}^2]$ is small.

\begin{lemma}\label{lem:distance_p_to_u}
For every $k\in\N$ there exists $\alpha_0>0$ such that for every $0<\alpha\leq\alpha_0, \delta\in(0,1)$, and $c\leq \frac{\delta \sqrt{n}}{100 \sqrt{\alpha k^5}}$ the following holds for all large enough~$n$. If $\cD$ is a distribution over $\{-1,1\}^k$ such that for all $j\in[k],\, \Exp_{\veca\sim\cD}[a_j]=0$,
and $A\subseteq\{-1,1\}^n$ is of size $|A|\geq 2^{n-c}$, then 
\[
\Exp_{\substack{M\\M \text{is } \veca\text{-resp.}}}[\|p_{M,\cD,\veca}-U\|_{tvd}^2]\leq \frac{\delta^2}{16} \,.
\]
where $U\sim\unif(\{-1,1\}^{k\alpha n})$.
\end{lemma}
\begin{proof}
\cref{lem:step 1 1 wise} and \cref{lem:step 2 1 wise} imply that for every $A$ of size $|A|\geq 2^{n-c}$,
\[
\Exp_{\substack{M\\M \text{is } \veca\text{-resp.}}}[\|p_{M,\cD,\veca}-U\|_{tvd}^2] \leq\frac{2^{2n}}{|A|^2}\cdot\sum_{\ell\geq2}^{k \alpha n}\left(\frac{\ell}{n}\right)^{\ell/2} (e^3\alpha k^5)^{\ell/2} \sum_{\substack{\vecv\in\{0,1\}^n\\|\vecv|=\ell}}\widehat{f}(\vecv)^2 \, .
\]

For every $\ell\in[4c]$, \cref{lem:kkl} implies that 
\begin{align*}
\frac{2^{2n}}{|A|^2}\sum_{\substack{\vecv\in\{0,1\}^n\\|\vecv|=\ell}}\widehat{f}(\vecv)^2&\leq\left(\frac{4\sqrt{2}c}{\ell}\right)^\ell \, .
\end{align*}
By the Parseval identity,
$\sum_{\vecv}\widehat{f}(\vecv)^2\leq1$. This gives us that
\begin{align*}
\Exp_{\substack{M\\M \text{is } \veca\text{-resp.}}}[\|p_{M,\cD,\veca}-U\|_{tvd}^2]&\leq\sum_{\ell\geq2}^{4c}\left(\frac{\ell}{n}\right)^{\ell/2} (e^3\alpha k^5)^{\ell/2} \cdot\left(\frac{4\sqrt{2}c}{\ell}\right)^\ell+\frac{2^{2n}}{|A|^2}\cdot\max_{4c<\ell\leq k\alpha n}\left\{\left(\frac{\ell}{n}\right)^{\ell/2} (e^3\alpha k^5)^{\ell/2}\right\}\,.
\intertext{Recall that $c\leq \frac{\delta \sqrt{n}}{100 \sqrt{\alpha k^5}}$. Let $\alpha_0=\frac{1}{2e^3k^5}$. Then for every $\alpha\leq\alpha_0$, the max term on the right hand side is maximized by $\ell=4c+1$ for all large enough $n$,}
&\leq\sum_{\ell\geq2}^{4c}\left(\frac{32e^3\alpha k^5 c^2}{n\ell}\right)^{\ell/2} + \left(\frac{8e^3c\alpha k^5}{n}\right)^{2c} \\
&\leq \sum_{\ell\geq2}^{4c} \left(\frac{\delta^2}{30}\right)^{\ell/2}+\left(\frac{8e^3\delta\sqrt{\alpha}}{100\sqrt{k^3}\sqrt{n}}\right)^{2c}
\\
&<\frac{\delta^2}{16} \, .
\end{align*}
\end{proof}

We are ready to finish the proof of \cref{thm:rmd}.
\begin{proof}[Proof of \cref{thm:rmd}]
Let us set $\tau=\frac{\delta}{200\sqrt{\alpha k^5}}$, and let $\alpha_0$ be as set in \cref{lem:distance_p_to_u}. 
Suppose that there exists a one-way communication protocol for $(\cD_Y,\cD_N)$-\textsf{Advice-RMD} that uses $s=\tau\sqrt{n}$ bits of communication and has advantage at least $\delta$. By the triangle inequality there must exist a protocol with advantage $\delta/2$ and $s$ bits of communication for either the $(\cD_Y,\cD_{unif})$-\textsf{Advice-RMD} or the $(\cD_N,\cD_{unif})$-\textsf{Advice-RMD} problem. Without loss of generality, we assume that $(\cD_Y,\cD_{unif})$-\textsf{Advice-RMD} can be solved with advantage $\delta/2$. Then,
\[
\|p_{M,\cD_Y,\veca}-p_{M,\cD_{unif},\veca}\|_{tvd}\geq\frac{\delta}{2} \, .
\]
Without loss of generality, we can assume that Alice's protocol is deterministic. In other words, for every $\veca$, Alice's $s$-bit communication protocol partitions the set of $\{-1,1\}^n$ of inputs $x$ into $2^s$ sets $A_1,\ldots,A_{2^s}\subseteq\{-1,1\}^n$ according to the message sent by Alice. Therefore, at least $(1-\delta/4)$-fraction of inputs $x\in\{-1,1\}^n$ belongs to sets $A_i$ of size $|A_i|\geq \frac{\delta}{4}\cdot 2^{n-s}\geq 2^{n-c}$ for $c=s+1-\log\delta$. By \cref{lem:distance_p_to_u}, for every $A_i$ of size $|A_i|\geq2^{n-c}$,
\[
\|p_{M,\cD_Y,\veca}-p_{M,\cD_{unif},\veca}\|_{tvd}|_{\vecx^*\in A_i}
=\Exp_{\substack{M\\M \text{is } \veca\text{-resp.}}}[\|p_{M,\cD,\veca}-U\|_{tvd}|_{\vecx^*\in A_i}]
\leq \delta/4 \,.
\]
Finally, 
\begin{align*}
\|p_{M,\cD_Y,\veca}-p_{M,\cD_{unif},\veca}\|_{tvd}
&\leq \Pr[x\in A_i\colon |A_i|< 2^{n-c}]\\
&+\Pr[x\in A_i\colon |A_i|\geq 2^{n-c}]\cdot\|p_{M,\cD_Y,\veca}-p_{M,\cD_{unif},\veca}\|_{tvd}|_{\vecx^*\in A_i}\\
&\leq \delta/4 + (1-\delta/4)\cdot \delta/4\\
&<\delta/2 \,.
\end{align*}
\end{proof}

\section{Hardness of Signal Detection}\label{sec:polar}

In this section we extend the hardness result of the SD problems for the special distributions described in~\cref{sec:kpart} to the fully general setting, thus proving the following theorem.

\reducermdtosd*

The bulk of this section is devoted to proving that for every pair of distributions $\cD_Y$ and $\cD_N$, we can find a path (a sequence) of intermediate distributions $\cD_Y = \cD_0,\cD_1,\ldots,\cD_L = \cD_N$ such that adjacent pairs in this sequence are indistinguishable by a ``basic'' argument, where a basic argument is a combination of an indistinguishability result from~\cref{thm:polarization step hardness} and a shifting argument.

Our proof comes in the following steps:

\begin{enumerate}
    \item For every marginal vector $\vecmu$, we identify a {\em canonical} distribution $\cD_{\vecmu}$ that we use as the endpoint of the path. So it suffices to prove that for all $\cD$, $\cD$ is indistinguishable from $\cD_{\vecmu(\cD)}$, i.e., there is a path of finite length from $\cD$ to $\cD_{\vecmu(\cD)}$. 
  
    \item We give a combinatorial proof that there is a path of finite length (some function of $k$) that takes us from an arbitrary distribution to the canonical one.
\end{enumerate}
Putting these ingredients together, along with a proof that a ``basic step'' is indistinguishable  gives us the final theorem.

Let $\gd = [q_1]\times \cdots [q_k]$ where $\forall i, q_i\in \mathbb{N}$. We start with the definition of the chain and the canonical distribution. For a distribution $\cD \in \Delta(\gd)$, its support is the set  $\supp(\cD) = \{\veca\in \gd\, |\, \cD(\veca) > 0\}$.  For $\cD \in \gd$, we define the marginal vector $\vecmu(\cD) = (\mu_{i,\sigma})_{i \in [k], \sigma\in [q_i]}$ as $\mu_{i,\sigma} = \Pr_{\veca \sim \cD} [a_i = \sigma]$. Next, we consider the following partial order on $\gd$. For vectors $\veca,\vecb \in \gd$ we use the notation $\veca \leq \vecb$ if $a_i \leq b_i$ for every $i \in [k]$. Further we use $\veca < \vecb$ if $\veca \leq \vecb$ and $\veca \ne \vecb$.

\begin{definition}[Chain]\label{def:chain}
We refer to a sequence $\veca(0)< \veca(1) < \cdots < \veca(\ell)$, $\veca(i) \in \gd$ for every $i \in \{0,\ldots,\ell\}$,  as a {\em chain} of length $\ell$. Note that chains in $\gd$ have length at most $\sum_{i=1}^k (q_i-1)$.
\end{definition}

\begin{lemma}[Canonical distribution]\label{prop:canonical}
Given a vector of marginals $\vecmu=(\mu_{i,\sigma})_{i\in[k],\sigma\in[q_i]}$, there exists a unique distribution $\cD$ with matching marginals ($\vecmu(\cD)=\vecmu$) such that the support of $\cD$ is a chain. We call this the \emph{canonical distribution} $\cD_\mu$ associated with $\vecmu$.
\end{lemma}

\begin{proof}
We will prove the proposition by applying induction on $\sum_{i=1}^k q_i$. In the base case when $\sum_{i=1}^k q_i = k$, there is only one point in the support of the distribution and the claim holds trivially. For  $\sum_{i=1}^k q_i > k$, define $h = \arg \min_{i\in[k]} \mu_{i,q_i}$ and $\tau = \mu_{h,q_h}$. Let $\tilde{q}_h = q_h-1$ and $\tilde{q}_i = q_i$, for $i\ne h$. Define a vector of marginals $\tilde{\vecmu}=(\tilde{\mu}_{i,\sigma})_{i\in[k],\sigma\in[\tilde{q}_i]}$ as follows: $\tilde{\mu}_{i,\sigma} = (\mu_{i,\sigma}-\tau)/(1-\tau)$ if $i\neq h$ and $\sigma=q_i$, and $\tilde{\mu}_{i,\sigma} = \mu_{i,\sigma}/(1-\tau)$ otherwise. By the induction hypothesis, there exists a unique distribution $\tilde{\cD}$ supported on a chain such that $\vecmu(\tilde{\cD})=\tilde{\vecmu}$. Observe that the distribution $\cD=(1-\tau)\tilde{\cD} + \tau \{(q_1,\dots,q_k)\}$ has marginal $\mu$ and is supported on a chain. We will now show that $\cD$ is the unique distribution with these properties. For a distribution $\cD\in\Delta([q_1]\times\cdots\times[q_k])$ and  $\vecv\in[q_1]\times\cdots\times[q_k]$, we define $\cD(\vecv) = \Pr_{\vecc\sim \cD}[\vecc=\vecv] $. Note that it suffices to prove that if $\cD'\in\Delta([q_1]\times\cdots\times[q_k])$ is supported on a chain and $\vecmu(\cD')=\vecmu$, then $\cD'(q_1,\dots,q_k)=\tau$. Clearly $\cD'(q_1,\dots,q_k)\le\tau$. Let $\vecu$ be lexicographically the largest vector smaller than $(q_1,\dots,q_k)$ in the support of $\cD'$. Let $r$ be an index where $\vecu_{r}<q_r$. Since $\cD'$ is supported on a chain, $\cD'(\vecv)=0$ for $\vecv\in[q_1]\times\cdots\times[q_k]$ such that $\vecv_r = q_r$ and $\vecv\neq (q_1,\dots,q_k)$. Hence $\mu_{r,q_r} = \cD'(q_1,\dots,q_k)$. Since $\tau = \min_{i\in[k]} \mu_{i,q_i}$, we have $\tau\le \mu_{r,q_r} = \cD'(q_1,\dots,q_k)$.

\end{proof}

For $\vecu,\vecv \in \gd$, let 
$\vecu' = \min\{\vecu,\vecv\}\triangleq(\min\{u_1,v_1\},\ldots,\min\{u_k,v_k\})$ and let 
$\vecv' = \max\{\vecu,\vecv\}\triangleq(\max\{u_1,v_1\},\ldots,\max\{u_k,v_k\})$. We say $\vecu$ and $\vecv$ are incomparable if $\vecu \not\leq \vecv$ and $\vecv \not\leq \vecu$. Note that if $\vecu$ and $\vecv$ are incomparable then $\{\vecu,\vecv\}$ and $\{\vecu',\vecv' \}$ are disjoint\footnote{To see this, suppose $\vecu=\vecu'$, then we have $u_j=\min\{u_j,v_j\}$ for all $j\in[k]$ and hence $\vecu\leq\vecv$, which is a contradiction. The same analysis works for the other cases.}.

\begin{definition}[Polarization (update) operator]\label{def:polarization operator}
Given a distribution $\cD \in \Delta(\gd)$ and incomparable elements $\vecu,\vecv\in \gd$, we define the $(\vecu,\vecv)$-{\em polarization} of $\cD$, denoted $\cD_{\vecu,\vecv}$, to be the distribution as given below. Let $\epsilon = \min\{\cD(\vecu),\cD(\vecv)\}$. 
\[
\cD_{\vecu,\vecv}(\vecb) =\left\{\begin{array}{ll}
\cD(\vecb)-\epsilon     & ,\ \vecb\in\{\vecu,\vecv\} \\
\cD(\vecb)+\epsilon       & ,\ \vecb \in \{\vecu', \vecv' \} \\
\cD(\vecb) & ,\ \text{otherwise.}
\end{array}\right.
\]
We refer to $\epsilon(\cD,\vecu,\vecv) = \min\{\cD(\vecu),\cD(\vecv)\}$ as the polarization amount.
\end{definition}

It can be verified that the polarization operator preserves the marginals, i.e., $\vecmu(\cD) = \vecmu(\cD_{\vecu,\vecv})$. Note also that this operator is non-trivial, i.e., $\cD_{\vecu,\vecv}=\cD$, if  $\{\vecu,\vecv\} \not\subseteq \supp(\cD)$.

\begin{restatable}[Indistinguishability of the polarization step]{theorem}{polarizationstephardness}\label{thm:polarization step hardness}
Let $n,k,q\in\N$, $\alpha\in(0,1)$ where $k,q,\alpha$ are constants with respect to $n$ and $\alpha n$ is an integer less than $n/k$. For a distribution $\cD\in\Delta([q]^k)$, incomparable vectors $\vecu,\vecv\in[q]^k$, and $\delta>0$, there exists $\tau>0$ such that every protocol for $(\cD,\cD_{\vecu,\vecv})$-SD  achieving advantage $\delta$ requires $\tau\sqrt{n}$ bits of communication. 
\end{restatable}

We defer the proof of this theorem to \cref{sec:indis SD shift} and focus instead on the number of steps

\subsection{Finite upper bound on the number of polarization steps}\label{ssec:finite}

In this section we prove that there is a finite upper bound on the number of polarization steps needed to move from a distribution $\cD \in \Delta(\gd)$ to the canonical distribution with marginal $\vecmu(\cD)$, i.e., $\cD_{\vecmu(\cD)}$. Together with the indistinguishability result from \cref{thm:polarization step hardness} this allows us to complete the proof of \cref{thm:communication lb matching moments} by going from $\cD_Y$ to $\cD_{\vecmu(\cD_Y)}=\cD_{\vecmu(\cD_N)}$ and then to $\cD_N$ by using the triangle inequality for indistinguishability.

In this section we extend our considerations to functions $A:\gd \to \R^{\geq 0}$. Let $\cfk = \{A:\gd \to \R^{\geq 0}\}$. For $A \in \cfk$ and $i\in[k]$, let $\mu_{0}(A) = \sum_{\veca\in\gd} A(\veca)$. Note $\Delta(\gd) \subseteq \cfk$ and $A \in \Delta(\gd)$ if and only if $A\in\cfk$ and $\mu_0(A) = \sum_{\veca \in \gd} A(\veca) = 1$. We extend the definition of marginals, support, canonical distribution, and polarization operators to $\cfk$. In particular we let $\vecmu(A) = (\mu_0,(\mu_{i,\sigma})_{i\in[k],\sigma\in[q_i]})$ where $\mu_{i,\sigma} = \sum_{\veca \in \gd: \veca_i = \sigma}  A(\veca)$. We also define canonical function and polarization operators so as to preserve $\vecmu(A)$. So given arbitrary $A$, let $\cD = \frac{1}{\mu_0(A)}\cdot A$. Note $\cD\in \Delta(\gd)$. For $\vecmu  = (\mu_0,(\mu_{i,\sigma})_{i\in[k],\sigma\in[q_i]})$ where $\forall i, \sum_{\sigma\in [q_i]} \mu_{i,\sigma} = \mu_0$, we define $A_{\vecmu} = \mu_0\cdot \cD_{\vecmu'}$ where $\vecmu' = (\mu_{i,\sigma}/\mu_0)_{i\in[k],\sigma\in[q_i]}$ to be the canonical function associated with $\vecmu$. 

\begin{definition}[Polarization length]\label{def:polarization-length}
For distribution $A\in\cfk$, where $\gd = [q_1]\times\cdots\times[q_k]$, let $N(A)$ be the smallest~$t$ such that there exists a sequence $\vecA = A_0,A_1,\ldots,A_t$ such that $A_0 = A$, $A_t = A_{\vecmu(A)}$ is canonical and for every $i \in [t]$ it holds that there exists incomparable  $\vecu_i,\vecv_i \in \supp(A_{i-1})$ such that $A_i = (A_{i-1})_{\vecu_i,\vecv_i}$. If no such finite sequence exists then let $N(A)$ be infinite. Let $N(k,q_1,\dots,q_k) = \sup_{A\in\cfk}  \{N(A)\}$, and $\tilde{N}(Q) = \max_{k,q_1,\ldots,q_k | \sum_i q_i = Q} N(k,q_1,\ldots,q_k)$. Again, if $N(A) = \infty$ for some $A$ or if no finite upper bound exists, $\tilde{N}(Q)$ is defined to be $\infty$. 
\end{definition}

Note that if $\cD \in \Delta(\gd)$, so is every element in the sequence, so the polarization length bound below applies also to distributions. Our main lemma in this subsection is the following:

\begin{lemma}[A finite upper bound on $\tilde{N}(Q)$]\label{lem:polarization finite}
$\tilde{N}(Q)$ is finite for every finite $Q$. Specifically $\tilde{N}(Q) \leq (Q^2+3)\tilde{N}(Q-1)\, .$
Consequently for every $k,q_1,\ldots,q_k$, $N(q_1,\ldots,q_k)$ is finite as well.
\end{lemma}

We prove~\cref{lem:polarization finite} constructively in the following four steps.

\paragraph{Step 1: The algorithm {\sc Polarize}.}
Let us start with some notations. For $A\in\cF([q_1]\times\cdots\times[q_k])$ we let $A|_{x_\ell = q_\ell}$ denote the function $A$ restricted to the domain $[q_1]\times\cdots\times[q_{\ell-1}] \times \{q_\ell\} \times [q_{\ell+1}]\times\cdots\times[q_k]$. Note that $A|_{x_\ell = q_\ell}$ is effectively a $(k-1)$-dimensional function. We also define $A|_{x_\ell < q_\ell}$ as the restriction of $A$ to the domain $[q_1]\times\cdots\times[q_{\ell-1}] \times [q_\ell-1] \times [q_{\ell+1}]\times\cdots\times[q_k]$..

\begin{algorithm}[H]
	\caption{$\textsc{Polarize}(\cdot)$}
	\label{alg:polarization}
    \begin{algorithmic}[1]
	    \Input $A\in\cF([q_1]\times\cdots\times[q_k])$.
		\If{k=1 OR $\not\exists i: q_i\ge 2$}
		    \State {\bf Output:} $A$.
		\EndIf
		\State WLOG, let $q_k\ge2$.
		\State $t \gets 0$; $Q^- \gets \sum_{i=1}^k (q_i-1) - 1$;  $Q^+ \gets \sum_{i=1}^{k-1}(q_i-1)$
		\State $(A_0)|_{x_k < q_k} \gets $ {\sc Polarize}$(A|_{x_k <q_k})$ ; $(A_0)|_{x_k = q_k} \gets $ {\sc Polarize}$(A|_{x_k = q_k})$ 
		\State Let $(1)^k = \veca_t(0) < \cdots < \veca_t(Q^-) = (q_1,\dots,q_{k-1},q_{k}-1)$ be a chain supporting $(A_t)|_{x_k <q_k}$.
        \State Let $((1)^{k-1},q_k) = \vecb_t(0) < \cdots < \vecb_t(Q^+) = (q_1,\dots,q_{k})$ be a chain supporting $(A_t)|_{x_k = q_k}$.
		\While{$\exists (i,j)$ with $j < Q^+$ s.t. $\max\{\veca_t(i),\vecb_t(j)\} = (q_1,\dots,q_k)$ and $A_t(\veca_t(i)), A_t(\vecb_t(j))>0$}
		    \State Let $(i_t,j_t)$ be the lexicographically smallest such pair $(i,j)$.
		    \State $B_t \gets (A_t)_{\veca_t(i_t),\vecb_t(j_t)}$.
		    \State $(A_{t+1})|_{x_k <q_k} \gets $ {\sc Polarize}($B_t|_{x_k <q_k}$); $(A_{t+1})|_{x_k = q_k} \gets (B_t)|_{x_k = q_k}$.
		    \State $t\gets t+1$.
		    \State Let $(1)^k = \veca_t(0) < \cdots < \veca_t(Q^-) = (q_1,\dots,q_{k}-1)$ be a chain supporting $(A_t)|_{x_k <q_k}$.
            \State Let $((1)^{k-1},q_k) = \vecb_t(0) < \cdots < \vecb_t(Q^+) = (q_1,\dots,q_{k})$ be a chain supporting $(A_t)|_{x_k = q_k}$.
		\EndWhile
		\State Let $\ell \in [k]$ be such that for every $\veca \in [q_1]\times\cdots\times[q_k] \setminus \{(q_1,\dots,q_k)\}$ we have $A_t(\veca)>0 \Rightarrow a_\ell < q_\ell$.
		\State $(A_{t+1})|_{x_\ell < q_\ell} \gets ${\sc Polarize}$(A_t)|_{x_\ell < q_\ell}$;  $(A_{t+1})|_{x_\ell = q_\ell} \gets (A_{t})|_{x_\ell = q_\ell}$.
		\State {\bf Output:} $A_{t+1}$.
	\end{algorithmic}
\end{algorithm}

The goal of the rest of the proof is to show that~\cref{alg:polarization} terminates after a finite number of steps and outputs $A_{\vecmu(A)}$.

\paragraph{Step 2: Correctness assuming {\sc Polarize} terminates.}

\begin{claim}[Correctness condition of {\sc Polarize}]\label{claim:polarization correctness}
For every $A\in\cF([q_1]\times\cdots\times[q_k])$, if $\textsc{Polarize}$ terminates, then $\textsc{Polarize}(A)=A_{\vecmu(A)}$. In particular, $\textsc{Polarize}(A)$ has the same marginals as $A$ and is supported on a chain.
\end{claim}
\begin{proof}
First, by the definition of the polarization operator (\cref{def:polarization operator}), the marginals of $A_t$ are the same for every $t$. So in the rest of the proof, we focus on inductively showing that if {\sc Polarize} terminates, then $\textsc{Polarize}(A)$ is supported on a chain.

The base case where $k=1$ is trivially supported on a chain as desired. 

When $k>1$, note that when the algorithm enters the Clean-up stage, if we let $m$ and $n$ denote the largest indices such that $A_t(\veca_t(m)), A_t(\vecb_t(n)) > 0$ and $ A_t(\vecb_t(n))\neq (q_1,\dots,q_k)$, then the condition that $\max\{\veca_t(m), \vecb_t(n)\} \ne (q_1,\dots,q_k)$ implies that there is a coordinate $\ell$ such that $\veca_t(m)_\ell < q_\ell$ and $\vecb_t(n)_\ell < q_\ell$. Since every $\vecc$ such that $A_t(\vecc) > 0$ and $c_k < q_k$ satisfies $\vecc \leq \veca_t(m)$, we have $A_t(\vecc) > 0$ implies $c_\ell < q_\ell$. Similarly for every $\vecc\ne (q_1,\dots,q_k)$ such that $c_k = q_k$, we have $A_t(\vecc)>0$ implies $c_\ell < q_\ell$.
We conclude that $A_t$ is supported on $\{(q_1,\dots,q_k)\} \cup \{\vecc\, |\, c_\ell < q_\ell\}$. 
Thus, by the induction hypothesis, after polarizing $(A_t)|_{x_\ell < q_\ell}$ and leaving $(A_t)|_{x_\ell = q_\ell}$ unchanged,

we get that the resulting function $A_{t+1}$ is supported on a chain as desired and complete the induction. We conclude that if {\sc Polarize} terminates, we have $\textsc{Polarize}(A)=A_{\vecmu(A)}$.
\end{proof}

\paragraph{Step 3: Invariant in {\sc Polarize}.} 
Now, in the rest of the proof of~\cref{lem:polarization finite}, the goal is to show that for every input $A$, the number of iterations of the while loop in~\cref{alg:polarization} is finite.
The key claim (\cref{clm:key-terminate}) here asserts that the sequence of pairs $(i_t,j_t)$ is monotonically increasing in lexicographic order. Once we establish this claim, it follows that there are at most $Q^-\cdot Q^+$ iterations of the while loop and so $\tilde{N}(Q) \leq (Q^2+3)\tilde{N}(Q-1)\, ,$ proving \cref{lem:polarization finite}.
Before proving \cref{clm:key-terminate}, we establish the following properties that remain invariant after every iteration of the while loop.

\begin{claim}\label{clm:chain-supp}
For every $t\ge0$, we have $(A_t)|_{x_k=q_k}$ and $(A_t)|_{x_k<q_k}$ are both supported on chains.
\end{claim}

\begin{proof}
For $(A_t)|_{x_k<q_k}$, the claim follows from the correctness of the recursive call to {\sc Polarize}. For $(A_t)|_{x_k=q_k}$, we claim by induction on $t$ that the supporting chain $\vecb_t(0) < \cdots < \vecb_t(Q^+)$ never changes (with $t$). To see this, note that $\vecb_t(k-1)=(q_1,\dots,a_k)$ is the only point in the support of $(A_t)|_{x_k=q_k}$ that increases in value, and this is already in the supporting chain. Thus $\vecb_t(0) < \cdots < \vecb_t(Q^+)$ continues to be a supporting chain for $(A_{t+1})|_{x_k=q_k}$. 
\end{proof}

For $\vecc \in [q_1]\times\cdots\times[q_k]$, we say that a function $A:[q_1]\times\cdots\times[q_k] \to \R^{\geq0}$ is {\em $\vecc$-respecting} if for every $\vecc'$ such that $A(\vecc') > 0$, we have $\vecc' \geq \vecc$ or $\vecc' \leq \vecc$. We say that $A$ is {\em $\vecc$-downward-respecting} if $A$ is $\vecc$-respecting and the points in the support of $A$ above $\vecc$ form a partial chain, specifically, if $\vecu,\vecv > \vecc$ have $A(\vecu),A(\vecv) > 0$, then either $\vecu \geq \vecv$ or $\vecv \geq \vecu$. 

Note that if $A$ is supported on a chain then $A$ is $\vecc$-respecting for every point $\vecc$ in the chain. Conversely, if $A$ is supported on a chain and $A$ is $\vecc$-respecting, then $A$ is supported on a chain that includes $\vecc$. 

\begin{claim}\label{clm:subcubes}
Let $A$ be a $\vecc$-respecting function and let $\tilde{A}$ be obtained from $A$ by a finite sequence of polarization updates, as in \cref{def:polarization operator}.
Then $\tilde{A}$ is also $\vecc$-respecting. Furthermore if $A$ is $\vecc$-downward-respecting and $\vecw > \vecc$ then 
$\tilde{A}$ is also $\vecc$-downward-respecting and $A(\vecw) = \tilde{A}(\vecw)$. 
\end{claim}

\begin{proof}
Note that it suffices to prove the claim for a single update by a polarization operator since the rest follows by induction. So let $\tilde{A} = A_{\vecu,\vecv}$ for incomparable $\vecu,\vecv\in \supp(A)$. 
Since $A$ is $\vecc$-respecting, and $\vecu,\vecv$ are incomparable, either $\vecu \leq \vecc,\vecv \leq \vecc$ or $\vecu \geq \vecc, \vecv \geq \vecc$. Suppose the former is true, then $\max\{\vecu, \vecv \}\leq \vecc$ and 
$\min\{\vecu,\vecv\} \leq \vecc$, and hence, $\tilde{A}$ is $\vecc$-respecting. Similarly, in the case when $\vecu \geq \vecc, \vecv \geq \vecc$, we can show that $\tilde{A}$ is $\vecc$-respecting. 
The furthermore part follows by noticing that for $\vecu$ and $\vecv$ to be incomparable if $A$ is $\vecc$-downward-respecting and $A(\vecu),A(\vecv)>0$, then $\vecu,\vecv \leq \vecc$, and so the update changes $A$ only at points below $\vecc$. 
\end{proof}

The following claim asserts that in every iteration of the while loop, by the lexicographically minimal choice of $(i_t,j_t)$, there exists a coordinate $h\in[k-1]$ such that every vector $c<a_t(i_t)$ in the support of $A_t$, $B_t$, or $A_{t+1}$ has $c_h<q_h$, and every vector $c\neq (q_1,\dots,q_k)$ in the support of $(A_t)|_{x_k = q_k}$ has $c_h<q_h$.

\begin{claim}\label{claim:polarization At Bt support}
For every $t\ge 0$, $\exists h \in [k-1]$ such that $\forall \vecc \in [q_1]\times\dots\times[q_k]$, if $\vecc\in \supp(A_t)\cup\supp(B_t)\cup\supp(A_{t+1})$, then the following hold:
\begin{itemize}
    \item If $\vecc<\veca_t(i_t)$, then $c_h<q_h$.
    \item If $c_k=q_k$ and $\vecc\neq(q_1,\dots,q_k)$, then $c_h<q_h$.
\end{itemize}
\end{claim}
\begin{proof}
Since $(i_t,j_t)$ is lexicographically the smallest incomparable pair in the support of $A_t$, for $i < i_t$, $j < Q^+$, and 
$A_t(\veca(i)), A_t(\vecb(j)) > 0$, we have $\max\{\veca(i),\vecb(j)\} \ne (q_1,\dots,q_k)$. Let $m$ be the largest index smaller than $i_t$ such that
$A_t(\veca_t(m)) > 0$. Similarly, let $n < Q^+$ be the largest index such that $A_t(\vecb_t(n)) > 0$. Then the fact that $\max\{\veca_t(m), \vecb_t(n)\} \ne (q_1,\dots,q_k)$ implies that there exists $h\in[k-1]$ such that $\veca_t(m)_h < q_h$ and $\vecb_t(n)_h < q_h$. Now, using the fact (from \cref{clm:chain-supp}) that $(A_t)|_{x_k<q_k}$ is supported on a chain, we conclude that for every $\vecc < \veca_t(i_t)$, $A_t(\vecc)> 0$ implies that $\vecc \leq \veca_t(m)$ and hence, $c_h<q_h$. Similarly, for every vector $\vecc \ne (q_1,\dots,q_k)$ in the support of $(A_t)|_{x_k = q_k}$, by the maximality of $n$, we have $c_h < q_h$.

We now assert that the same holds for $B_t$. First, recall that since $B_t=(A_t)_{\veca_t(i_t),\vecb_t(j_t)}$, we have that $\supp(B_t)\subset\supp(A_t)\cup\{(q_1,\dots,q_k),\min\{\veca_t(i_t),\vecb_t(j_t)\}\}$.
Next, note that the only point (other than $(q_1,\dots,q_k)$) where $B_t$ is larger than $A_t$ is
$\min\{\veca_t(i_t),\vecb_t(j_t)\}$. It suffices to show that $\min\{\veca_t(i_t),\vecb_t(j_t)\}_h < q_h $. We have $\min\{\veca_t(i_t),\vecb_t(j_t)\} \leq \vecb_t(j_t) \leq \vecb_t(n)$ and hence $\min\{\veca_t(i_t),\vecb_t(j_t)\}_h < q_h$.

Finally, we assert that same holds also for $A_{t+1}$. Since $A_{t+1}|_{x_k=q_k} = B_t|_{x_k=q_k}$, the second item in the claim follows trivially. To prove the first item, let us consider $\veca' \in [q_1]\times\cdots\times[q_k]$ defined as follows: $\veca'_h = q_h-1$ and $\veca'_{r} = \veca_t(i_t)_{r}$ for $r \ne h$. Note that $B_t|_{x_k <q_k}$ is $\veca_t(i_t)$-respecting since potentially the only new point in its support (compared to $A_t|_{x_k < q_k}$) is $\min\{\veca_t(i_t),\vecb_t(j_t)\} \leq \veca_t(i_t)$. From the previous paragraph we also have that if $B_t(\vecc) > 0$ and $\vecc < \veca_t(i_t)$, then $c_h < q_h$ and hence, $\vecc \leq \veca'$. On the other hand, if $B_t(\vecc) > 0$ and $\vecc \ge \veca_t(i_t)$, then $\vecc\ge a'$. Therefore, $B_t|_{x_k<q_k}$ is $\veca'$-respecting. By applying \cref{clm:subcubes}, we conclude that $(A_{t+1})|_{x_k<q_k}$ is also $\veca'$-respecting. It follows that if $\vecc < \veca(i_t)$ and $A_{t+1}(\vecc) > 0$, then $\vecc \leq \veca'$ and so $c_h < q_h$.
\end{proof}

\paragraph{Step 4: Proof of~\cref{lem:polarization finite}.}

The following claim establishes that the while loop in the $\textsc{Polarize}$ algorithm terminates after a finite number of iterations.

\begin{claim}\label{clm:key-terminate}
For every $t\ge 0$, $(i_t,j_t)<(i_{t+1},j_{t+1})$ in lexicographic ordering.
\end{claim}
\begin{proof}

Consider the chain $\veca_{t+1}(0) < \cdots < \veca_{t+1}(Q^-)$ supporting $A_{t+1}|_{x_k < q_k}$. Note that for  $i \geq i_t$, $A_{t+1}|_{x_k < q_k}$ is $\veca_t(i)$-respecting (since $A_t|_{x_k<q_k}$ and $B_{t}|_{x_k<q_k}$ were also so).  In particular, $A_t|_{x_k<q_k}$ is $\veca_t(i)$-respecting because it is supported on a chain containing $a_t(i)$. Next $B_t|_{x_k<q_k}$ is $\veca_t(i)$-respecting since potentially the only new point in its support is $\min\{\veca_t(i_t),\vecb_t(j_t)\} \leq \veca_t(i)$. Finally, $A_{t+1}|_{x_k <q_k}$ is also $\veca_t(i)$-respecting using \cref{clm:subcubes}. Thus we can build a chain containing  $\veca_t(i)$ that supports $A_{t+1}|_{x_k<q_k}$. It follows that we can use $\veca_{t+1}(i) = \veca_t(i)$ for $i \geq i_t$. Now consider $i < i_t$. We must have $\veca_{t+1}(i) < \veca_{t+1}(i_t) = \veca_t(i_t)$.
By~\cref{claim:polarization At Bt support}, there exists $h\in [k-1]$ such that for $i < i_t$, $\veca_{t+1}(i)_h < q_h$.

We now turn to analyzing $(i_{t+1},j_{t+1})$. By definition, $A_{t+1}(\veca_{t+1}(i_{t+1}))>0$ and 
$A_{t+1}(\vecb_{t+1}(b_{t+1}))>0$. First, let us show that $i_t \le i_{t+1}$. On the contrary, let us assume that $i_{t+1} < i_t$. It follows from the above paragraph that $\veca_{t+1}(i_{t+1})_h < q_h$. Also, for every $\vecb_{t+1}(j)$ with $j <Q^+$ and $A_{t+1}(\vecb_{t+1}(j))>0$, we have $\vecb_{t+1}(j)_h < q_h$. Therefore, $\max\{\veca(i_{t+1}),\vecb(j_{t+1})\} \ne (q_1,\dots,q_k)$ (in particular $\max\{\veca(i_{t+1}),\vecb(j_{t+1})\}_h<q_h$), which is a contradiction.

Next, we show that if $i_{t+1} = i_t$, then $j_{t+1}\ge j_t$. By the minimality of $(i_t,j_t)$ in the $t$-th round, for $j<j_t$ such that $A_t(b_t(j))>0$, we have $\max\{a_{t}(i_{t}),b_t(j)\} \ne (q_1,\dots,q_k)$. Since  $i_{t+1} = i_t$, $a_{t+1}(i_{t+1})=a_{t+1}(i_{t})=a_{t}(i_{t})$. We already noted in the proof of \cref{clm:chain-supp} that $\vecb_t(0) < \cdots < \vecb_t(Q^+)$ is also a supporting chain for $(A_{t+1})|_{x_k=q_k}$. The only point where the function $A_{t+1}|_{x_k=q_k}$ has greater value than $A_t|_{x_k=q_k}$ is $(q_1,\dots,q_k)$. Therefore, for $j<j_t$ such that $A_{t+1}(b_{t+1}(j))>0$, we have $\max\{a_{t+1}(i_{t+1}), b_{t+1}(j)\} \ne (q_1,\dots,q_k)$ and hence, $j_{t+1}\ge j_t$.

So far, we have established that $(i_{t+1},j_{t+1})\ge (i_t,j_t)$ in lexicographic ordering. Finally, we will show that $(i_{t+1},j_{t+1}) \neq (i_t,j_t)$ by proving that at least one of $A_{t+1}(\veca_{t+1}(i_{t}))$ and $A_{t+1}(\vecb_{t+1}(j_t))$ is zero. The polarization update ensures that at least one of $B_{t}(\veca_{t}(i_{t}))$ and $B_{t}(\vecb_t(j_t))$ is zero. If $B_t(\vecb_t(j_t)) = 0$, then by definition, we have $A_{t+1}(\vecb_{t+1}(j_t)) = A_{t+1}(\vecb_t(j_t)) = 0$. Finally to handle the case $B_{t}(\veca_{t}(i_{t}))=0$, let us again define $\veca'$ as: $\veca'_h = q_h-1$ and $\veca'_{r} = \veca_t(i_t)_{r}$ for $r \ne h$, where $h$ is as given by \cref{claim:polarization At Bt support}. We assert that $B_t|_{x_k<q_k}$ is $\veca'$-downward-respecting. As shown in the proof of \cref{claim:polarization At Bt support}, we have $B_t|_{x_k<q_k}$ is $\veca'$-respecting. The support of $B_t|_{x_k<q_k}$ is contained in $\{\veca_t(0),\cdots,\veca_t(Q^-)\} \cup \{\min\{\veca_t(i_t),\vecb_t(j_t)\}\}$ and $\min\{\veca_t(i_t),\vecb_t(j_t)\} < \veca_t(i_t)$, and by \cref{claim:polarization At Bt support}, $\min\{\veca_t(i_t),\vecb_t(j_t)\}\leq \veca'$. It follows that $B_t|_{x_k<q_k}$ is $\veca'$-downward-respecting. Finally, by the furthermore part of \cref{clm:subcubes} applied to $B_{t}|_{x_k<q_k}$ and $\vecw = \veca_t(i_t)$, we get that $A_{t+1}(\veca_{t+1}(i_{t}))=A_{t+1}(\veca_t(i_t)) = B_{t}(\veca_t(i_t)) = 0$. It follows that $(i_{t+1},j_{t+1}) \ne (i_t,j_t)$. 
\end{proof}

\begin{proof}[Proof of~\cref{lem:polarization finite}]
By~\cref{claim:polarization correctness}, we know that if~\cref{alg:polarization} terminates, then we have $\textsc{Polarize}(A)=A_{\vecmu(A)}$. Hence, the maximum number of polarization updates used in {\sc Polarize} (on input from $\cF([q_1]\times\cdots\times[q_k])$) serves as an upper bound for $\tilde{N}(Q)$, for $Q=\sum_{i=1}^k q_k$.
By~\cref{clm:key-terminate}, we know that there are at most $Q^2$ iterations of the while loop and so $\tilde{N}(Q) \leq (Q^2+3)\tilde{N}(Q-1)$ as desired.
\end{proof}

\subsection{Reduction from single function to a family of functions}\label{sec:mixed}
In this subsection, we prove the following lemma that reduces an SD problem for a single function to an SD problem for a family of functions.

\begin{lemma}\label{lem:reduce-to-single-function}
	Suppose there exists $\cF,\cD_Y,\cD_N$, $\delta > 0$ with $\vecmu(\cD_Y) = \vecmu(\cD_N)$ and a $c=c(n)$-communication protocol achieving advantage $\delta$ solving $(\cF,\cD_Y,\cD_N)$-SD on instances of length $n$ for every $n \geq n_0$. Then
	there exist $\cD_1,\cD_2\in\Delta([q]^k)$ with $\vecmu(\cD_1) = \vecmu(\cD_2)$, $\delta' > 0$, $n'_0$,  and a $c$-communication protocol achieving advantage $\delta'$ solving $(\cD_1,\cD_2)$-SD on instances of length $n\geq n'_0$ using $O(s)$ bits of communication. 
\end{lemma}

	We prove the lemma by a hybrid argument, where we slowly change the distribution $\cD_Y$ to $\cD_N$ by considering one function from $\cF$ at a time. The crux of the lemma is in showing that two adjacent steps in this sequence are at least as hard as some single-function SD problem, which follows from the following lemma.
	
    \begin{restatable}{lemma}{restindis}\label{lem:indis of shift}
    Let $n,k,q\in\N$, $\alpha\in(0,1)$ where $k,q,\alpha$ are constants with respect to $n$ and $\alpha n$ is an integer less than $n/k$. Let $\cF \subseteq\{f:[q]^k\to\{0,1\}\}$ For every $\epsilon,\delta\in(0,1]$, there exist $n'=\Omega(n)$ and constants $\alpha',\delta'\in(0,1)$ such that the following holds. For every distributions $\cD_Y,\cD_N,\cD_0,\cD_1,\cD_2\in\Delta(\cF\times[q]^k)$ such that $\cD_Y=(1-\epsilon)\cD_0+\epsilon\cD_1$ and $\cD_N=(1-\epsilon)\cD_0+\epsilon\cD_2$ and for every $c\in\N$, suppose there exists a protocol for $(\cF,\cD_Y,\cD_N)$-\textsf{SD} with parameters $n$ and $\alpha$ using $c$ bits of communication with advantage $\delta$, then there exists a protocol for $(\cF,\cD_1,\cD_2)$-\textsf{SD} with parameters $n'$ and $\alpha'$ using $c$ bits of communication with advantage $\delta'$.
    \end{restatable}
    
    The proof idea of~\cref{lem:indis of shift} is very similar to that of~\cref{thm:polarization step hardness}. We defer the proof to~\cref{sec:indis SD shift} and turn to showing how \cref{lem:reduce-to-single-function} follows.
    
\begin{proof}[Proof of \cref{lem:reduce-to-single-function}]
	Let $\ALG(\vecx^*; M,\vecz)$ be the $c$-bit protocol for $(\cF,\cD_Y,\cD_N)$-SD achieving advantage $\delta$ guaranteed to exist by the theorem statement. Let $\cF=\{f_1,\ldots,f_\ell\}$. Since $\vecmu(\cD_Y) = \vecmu(\cD_N)$ for each $i\in[m]$, we have that $\Pr[f=f_i\colon (f,\vecb)\sim\cD_Y]=\Pr[f=f_i\colon (f,\vecb)\sim\cD_N]$. Let us denote this probability by $w^{(i)}$, $w^{(i)}=\Pr[f=f_i\colon (f,\vecb)\sim\cD_Y]=\Pr[f=f_i\colon (f,\vecb)\sim\cD_N]$ for each $i\in[\ell]$.
	For each $i\in[\ell]$, let $\cD_Y^{(i)}$ be the distribution of a random variable $\vecb\in[q]^k$ that is sampled from $(f,\vecb)\sim\cD_Y$ conditioned on $f=f_i$. Similarly, for each $i\in[\ell]$, let $\cD_N^{(i)}$ be the distribution of $\vecb\in[q]^k$ from $(f,\vecb)\sim\cD_N$ conditioned on $f=f_i$. This way we have that $\cD_Y$ and $\cD_N$ are the mixture distributions:
	$\cD_Y=\sum_{i\in[\ell]} w^{(i)}\cdot \cD_Y^{(i)}$ and $\cD_N=\sum_{i\in[\ell]} w^{(i)}\cdot \cD_N^{(i)}$.
	
	For every $i\in\{0,\ldots,\ell\}$, we define a distribution $\cD^{(i)}$ as the following mixture distribution:
	\[
	\cD^{(i)}=\sum_{j\in\{1,\ldots,i\}} w^{(j)}\cdot \cD_N^{(j)}+
	\sum_{j\in\{i+1,\ldots,\ell\}} w^{(j)}\cdot \cD_Y^{(j)} \,.
	\]
	Let $p_i=\Pr[\ALG(\vecx^*; M,\vecz)=\yes\colon (f,\vecb)\sim\cD^{(i)}]$ for every $i\in\{0,\ldots,\ell\}$. Observe that $p_0=\Pr[\ALG(\vecx^*; M,\vecz)=\yes\colon (f,\vecb)\sim\cD_Y]$ and $p_\ell=\Pr[\ALG(\vecx^*; M,\vecz)=\yes\colon (f,\vecb)\sim\cD_N]$. Since the advantage of $\ALG$ in distinguishing $\cD_Y$ and $\cD_N$ is at least $\delta$, we have that
	\[
	\delta= \left|p_0-p_\ell\right|=\left|\sum_{i\in\{0,\ldots,\ell-1\}}(p_i-p_{i+1})\right| \leq \sum_{i\in\{0,\ldots,\ell-1\}}\left|p_i-p_{i+1}\right|\,.
	\]
	Let $\delta'=\delta/\ell$. We have that at least one term of this sum is $\left|p_i-p_{i+1}\right|\geq \delta'$. From this we conclude that for some $i\in\{0,\ldots,\ell-1\}$, $\ALG$ achieves advantage at least $\delta'$ for $(\cF,\cD^{(i)},\cD^{(i+1)})$-SD. 
	
	It remains to show that if one can distinguish $\cD^{(i)}$ and $\cD^{(i+1)}$ that differ only for $(f,\vecb)$ with $f=f_{i+1}$, then one can also distinguish $\cD_1=\cD_Y^{(i+1)}$ and $\cD_2=\cD_N^{(i+1)}$. Since $\vecmu(\cD_1) = \vecmu(\cD_2)$, this will finish the proof. We show that $\cD_1$ and $\cD_2$ are distinguishable using \cref{lem:indis of shift}.
	
	Let us define $\eps=w_{i+1}$, $\cD=\frac{1}{1-\eps}\left(\sum_{j\in\{1,\ldots,i\}} w^{(j)}\cdot \cD_N^{(j)}+
	\sum_{j\in\{i+2,\ldots,\ell\}} w^{(j)}\cdot \cD_Y^{(j)}\right)$. Now observe that
	$\cD^{(i)}=(1-\epsilon)\cD+\epsilon\cD_1$ and
	$\cD^{(i+1)}=(1-\epsilon)\cD+\epsilon\cD_2$. Now by \cref{lem:indis of shift}, a protocol that distinguishes $\cD^{(i)}$ and $\cD^{(i+1)}$ implies a protocol for $(\cD_1, \cD_2)$-SD with advantage $\delta''>0$ and communication complexity $O(s)$.
\end{proof}

\subsection{Putting it together}\label{ssec:proof-of-sd}

We now have the ingredients in place to prove \cref{thm:communication lb matching moments} which we recall below for convenience.

\reducermdtosd*

\begin{proof}[Proof of \cref{thm:communication lb matching moments}]
Fix $\cF \subseteq \{f:[q]^k \to \{0,1\}\}$ and distributions $\cD_Y,\cD_N \in \Delta(\cF\times[q]^k)$ with $\vecmu = \vecmu(\cD_Y) = \vecmu(\cD_N)$. \cref{lem:reduce-to-single-function}, applied to $(\cF, \cD_Y,\cD_N)$, gives us $n_0,\delta'$, and distributions
$\cD_Y',\cD_N' \in \Delta([q]^k)$ with $\vecmu'=\vecmu(\cD_Y') = \vecmu(\cD_N')$ such that any $c$-communication protocol for $(\cF,\cD_Y,\cD_N)$-\textsf{SD} with advantage $\delta$ implies a $c$-communication protocol for $(\cD_Y',\cD_N')$-\textsf{SD} with advantage $\delta'$ for all $n\geq n_0$. Now we'll focus on proving a lower bounds for the problem $(\cD_Y',\cD_N')$-\textsf{SD}.

\cref{lem:polarization finite}, applied to $\cD_Y'$, gives us $\cD_0 = \cD_Y',\cD_1,\ldots,\cD_t = \cD_{\vecmu'}$ such that $\cD_{i+1} = (\cD_i)_{\vecu(i),\vecv(i)}$, i.e., $\cD_i$ is an update of $\cD_i$, with $t \leq \tilde{N}(Q) < \infty$, for $Q=\sum_{i=1}^k q_k$. 
Similarly \cref{lem:polarization finite}, applied to $\cD_N'$, gives us $\cD'_0 = \cD_N',\cD'_1,\ldots,\cD'_{t'} = \cD_{\vecmu'}$ such that $\cD'_{i+1} = (\cD'_i)_{\vecu'(i),\vecv'(i)}$ with $t' \leq \tilde{N}(Q) < \infty$. 

Applying \cref{thm:polarization step hardness} with $\delta'' = \delta'/(2\tilde{N}(Q))$ to the pairs $\cD_i$ and $\cD_{i+1}$, we get that there exists $\tau_i$ such that every protocol for 
$(\cD_i,\cD_{i+1})$-\textsf{SD} requires $\tau_i \sqrt{n}$ bits of communication to achieve advantage $\delta''$. 
Similarly applying \cref{thm:polarization step hardness} again with $\delta'' = \delta'/(2\tilde{N}(Q))$ to the pairs $\cD'_i$ and $\cD'_{i+1}$, we get that there exists $\tau'_i$ such that every protocol for 
$(\cD'_i,\cD'_{i+1})$-\textsf{SD} requires $\tau'_i \sqrt{n}$ bits of communication to achieve advantage $\delta''$. 

Letting $\tau' = \min\left\{\min_{i \in [t]}\{\tau_i\}, \min_{i \in [t']}\{\tau'_i\}\right\}$, we get, using the triangle inequality for indistinguishability, that every protocol $\Pi'$ for $(\cD_Y',\cD_N')$-\textsf{SD} achieving advantage $(t+t')\delta''\leq\delta'$ requires $\tau'\sqrt{n}$ bits of communication. Finally, by~\cref{lem:reduce-to-single-function}, every protocol $\Pi$ for $(\cF,\cD_Y,\cD_N)$-\textsf{SD} achieving advantage $\delta$ requires $\tau'\sqrt{n}$ bits of communication.
\end{proof}
\section{Indistinguishability of the Polarization Step}\label{sec:spl}

Recall that in~\cref{def:polarization operator} we define a polarization operator that polarizes a distribution $\cD\in\Delta([q]^k)$ to $\cD_{\vecu,\vecv}\in\Delta([q]^k)$ for every incomparable pair $(\vecu,\vecv)$. In this section, we show that $(\cD,\cD_{\vecu,\vecv})$-SD requires $\Omega(\sqrt{n})$ communication.

\polarizationstephardness*

Let $\vecu\vee\vecv,\vecu\wedge\vecv \in [q]^k$ be given by $u_i\vee v_i = \max\{u_i,v_i\}$ and $u_i\wedge v_i = \min\{u_i,v_i\}$.
Let $\cA_Y = \unif(\{\vecu,\vecv\})$ and $\cA_N = \unif(\{\vecu\vee\vecv,\vecu\wedge\vecv\})$. We prove~\cref{thm:polarization step hardness} in two steps. First, we use the Boolean hardness in~\cref{thm:kpart} to show in~\cref{lem:SD boolean to non-boolean} that the hardness holds for the special case $(\cA_Y,\cA_N)$-SD.
Next, we reduce $(\cA_Y,\cA_N)$-advice-SD to $(\cD,\cD_{\vecu,\vecv})$-SD for arbitrary distribution $\cD\in\Delta([q]^k)$.

\subsection{Reduce a Boolean SD problem to a non-Boolean SD problem}
In this subsection, we consider a special case of $\vecu,\vecv\in[q]^k$ where $u_i\neq v_i$ for every $i\in[k]$. The following key lemma of this subsection establishes the hardness of $(\cA_Y,\cA_N)$-SD via a reduction from a Boolean SD problem to a non-Boolean version. 

\begin{lemma}\label{lem:SD boolean to non-boolean}
Let $n,k,q\in\N$, $\alpha\in(0,1)$ where $k,q,\alpha$ are constants with respect to $n$ and $\alpha n$ is an integer less than $n/k$.
For $\vecu,\vecv \in [q]^k$ satisfying $u_i \ne v_i$ for all $i \in [k]$ and $\delta>0$, there exists $\tau>0$ such that every protocol for $(\cA_Y,\cA_N)$-SD  achieving advantage $\delta$ requires $\tau\sqrt{n}$ bits of communication.
\end{lemma}

We prove~\cref{lem:SD boolean to non-boolean} by a reduction. For such $\vecu,\vecv$, let $\bar{\vecu},\bar{\vecv} \in \{0,1\}^k$ be the Boolean version given by $(\bar{u}_i,\bar{v}_i) = (0,1)$ if $u_i < v_i$ and $(\bar{u}_i,\bar{v}_i) = (1,0)$ if $u_i > v_i$. Let $\bar{\cA}_Y = \unif(\{\bar{\vecu},\bar{\vecv}\})$ and $\bar{\cA}_N = \unif(\{\bar{\vecu}\vee\bar{\vecv},\bar{\vecu}\wedge\bar{\vecv}\})$. Note that both $\bar{\cA}_Y$ and $\bar{\cA}_N$ are distributions on Boolean domain with uniform marginals. Thus,~\cref{thm:kpart} shows that any protocol for $(\bar{\cA}_Y,\bar{\cA}_N)$-advice-SD requires $\Omega(\sqrt{n})$ bits of communication. In the rest of this subsection, we reduce $(\bar{\cA}_Y,\bar{\cA}_N)$-advice-SD to $(\cA_Y,\cA_N)$-SD.

For every $\bar{n},k,\bar{\alpha},q,\delta$, let $n=2q\bar{n}$ and $\alpha=q^{k-1}2^{-(k+2)}\bar{\alpha}$.
Let $\bar{I} = (\bar{\vecx},\bar{\Gamma},\bar{\vecb},\bar{M},\bar{\vecz},\bar{\veca})$ denote an instance of $(\bar{\cA}_Y,\bar{\cA}_N)$-advice-SD of length $\bar{n}$ with parameter $\bar{\alpha}$.
We show below how Alice and Bob can use their inputs and shared randomness to generate an instance 
$I = (\vecx,\Gamma,\vecb,M,\vecz,\veca)$  of $(\cA_Y,\cA_N)$-advice-SD of length $n$ with parameter $\alpha$ ``locally'' and ``nearly'' according to the correct distributions. Namely, we show that with high probability if $\bar{I}$ is a Yes (resp. No) instance of $(\bar{\cA}_Y,\bar{\cA}_N)$-advice-SD, then $I$ will be a Yes (resp. No) instance of $(\cA_Y,\cA_N)$-SD.

\paragraph{Step 1: Specify the shared randomness.}
The common randomness between Alice and Bob is an instances $I_R = (\vecx_R,\Gamma_R,\vecb_R,M_R,\vecz_R,\veca_R)$ drawn according to the Yes\footnote{The reduction also works if we used No distribution. However, the mapping between Yes and No instances would get flipped. Namely, if $\bar{I}$ is a Yes (resp. No) instance of $(\bar{\cA}_Y,\bar{\cA}_N)$-advice-SD, then $I$ will be a No (resp. Yes) instance of $(\cA_Y,\cA_N)$-SD.} distribution of $(\cA_Y,\cA_N)$-advice-SD of length $n$ with parameter $\alpha$. For $j \in [\alpha n]$, let $V_j$ denote the set of variables in the $j$-th constraint, i.e., $V_j = \{\ell \in [n]\, |\,  \Gamma_R(\ell) \in \{ k(j-1)+1,\ldots,k(j-1)+k\}\}$. For $i \in [k]$, let $T_i$ be the set of variables that are in the $i$-th partition and take on values in $\{u_i,v_i\}$, i.e., 
$T_i = \{j \in [n]\, |\, a_j = i\, \&\, (\vecx_R)_i \in \{u_i,v_i\}\}$. Let $U \subseteq [\alpha n]$ be the set of constraints that work on variables in $T_i$, i.e., $U = \{j \in [\alpha n]\, |\, V_j \subseteq \cup_i T_i \}$. See~\cref{fig:polarization step shared randomness} for an example.

\begin{figure}[ht]
    \centering
    \includegraphics[width=12cm]{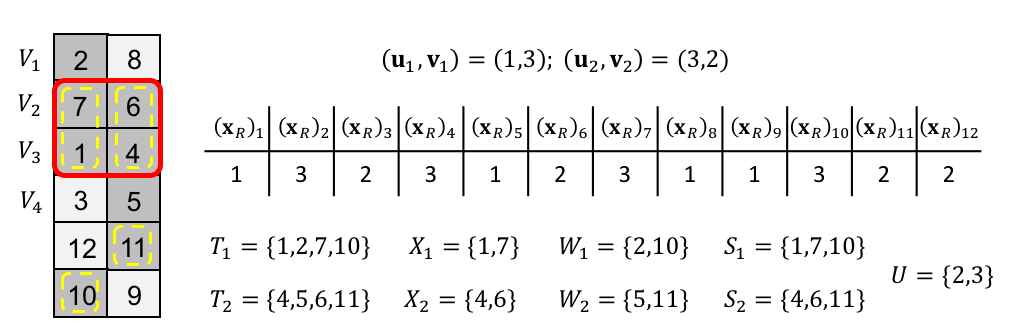}
    \caption{An example of shared randomness used in~\cref{lem:SD boolean to non-boolean}. Here $n=12$, $k=2$, $q=3$, and $\alpha=1/3$. The value of $\vecx_R\in[q]^n$ is listed in a table. Consider $(u_1,v_1)=(1,3)$ and $(u_2,v_2)=(3,2)$. The variables in sets $T_1,T_2$ are marked grey. The variables correspond to the set $U$ are circled with red lines and the variables correspond to sets $S_1,S_2$ are circled with yellow dashed lines.}
    \label{fig:polarization step shared randomness}
\end{figure}

If $|U| \geq \bar{\alpha}\bar{n}$ we say an error of type (1) has occurred. For $i \in [k]$, let $X_i\subseteq T_i$ be the set of variables that operate on constraints in $U$, i.e., $X_i = T_i \cap (\cup_{j \in U} V_j)$. Let $W_i \subseteq T_i$ be a set of variables that do not participate in any constraint, i.e., $W_i = T_i \setminus (\cup_{j \in [\alpha n]} V_j)$. Finally let $S_i$ be any set satisfying $|S_i| = \bar{n}/k$ with $X_i \subseteq S_i \subseteq X_i \cup W_i$ if such a set exists. If no such set exists we say an error of type (2) has occurred.

\paragraph{Step 2: Specify the reduction.}
If there is an error, we simply set $I=I_R$. If no errors have occurred, our reduction will embed $\bar{I}$ into $I_R$ by replacing the constraints in $U$ and the variables in $\cup_i S_i$ as described next. Note that we have to specify variables $(\vecx,\Gamma,\vecb,M,\vecz)$. In particular, we want the private inputs can be computed locally. We verify the local property of the reduction in~\cref{claim:SD boolean to non-boolean local} and prove the correctness of the reduction in~\cref{claim:SD boolean to non-boolean correctness}.

\begin{itemize}
\item $\vecx$: Let $\rho:[\bar{n}] \to \cup_i S_i$ be a bijection satisfying $\bar{a}_j = i \Rightarrow \rho(j) \in S_i$. We now define $\vecx \in [q]^n$ as follows:
\[
x_j = \left \{ \begin{array}{ll} 
        (\vecx_R)_j & \mbox{ if $j \notin \cup_{i \in [k]}S_i$ } \\
        u_i & j \in S_i \mbox{ for some } i \in [k] \mbox{ and } u_i < v_i \mbox{ and } \bar{x}_j = 0 \\
        u_i & j \in S_i \mbox{ for some } i \in [k] \mbox{ and } u_i > v_i \mbox{ and } \bar{x}_j = 1 \\
        v_i & j \in S_i \mbox{ otherwise }
     \end{array} \right .
\]

\item $\Gamma$ and $M$: Let $V = \{V(1),\ldots,V(\bar{n})\}$ with $V(j) < V(j+1)$ be such that $V = \{j \in [n] | \Gamma_R(j) \in \cup_{i \in [k]} S_i\}$. For $j \in [n]$ we let
\[
\Gamma(j) =  \left \{ \begin{array}{ll}
                        \Gamma_R(j) & \mbox{ if $j \notin V$}\\
                        \rho(\bar{\Gamma}(\bar{j})) & \mbox{ if $j = V(\bar{j})$} \\
                      \end{array} \right .
\]
It may be verified that $\Gamma$ is a permutation and furthermore the constraints in $\Gamma$ corresponding to $j \in U$ are derived from constraints of $\bar{I}$. $M$ is then defined as the partial permutation matrix capturing $\Gamma^{-1}(j)$ for $j\in[k\alpha n]$.

\item $\vecb$: Since $\vecb$ is a hidden variable and won't be given to Alice and Bob, we postpone the specification of $\vecb$ to the proof of~\cref{claim:SD boolean to non-boolean correctness}.

\item $\vecz$: Let $\vecz(j)=\bar{\vecz}(V(j))$ if $j\in U$ and $\vecz(j)=\vecz_R(j)$ otherwise.
\end{itemize}

\paragraph{Step 3: Correctness of the reduction assuming no error occurs.}
\begin{claim}[The reduction can be computed locally]\label{claim:SD boolean to non-boolean local}
Let $\bar{I} = (\bar{\vecx},\bar{\Gamma},\bar{\vecb},\bar{M},\bar{\vecz},\bar{\veca})$ be an instance of $(\bar{\cA}_Y,\bar{\cA}_N)$-advice-SD and $I_R = (\vecx_R,\Gamma_R,\vecb_R,M_R,\vecz_R,\veca_R)$ be the shared randomness of Alice and Bob. The above reduction satisfies the following local properties:
\begin{itemize}
\item Alice can compute $\vecx$ using $I_R$ and $(\bar{\vecx},\bar{\veca})$. 
\item Bob can compute $(M,\vecz)$ using $I_R$ and $(\bar{M},\bar{\vecz},\bar{\veca})$.
\end{itemize}
\end{claim}
\begin{proof}
\mbox{}
\begin{itemize}
\item Note that from the construction, it suffices to have $\{S_i\}$, $\bar{\veca}$, $\bar{\vecx}$ to compute $\vecx$. Since $\{S_i\}$ can be obtained from $I_R$, we conclude that Alice can compute $\vecx$ using $I_R$ and $(\bar{\vecx},\bar{\veca})$. 
\item Note that from the construction, it suffices to have $\{S_i\}$, $\Gamma_R$, $\bar{\Gamma}(j)$ where $j\in[k\alpha n]$ to compute $M$. Since $\bar{\Gamma}(j)$ is encoded in $\bar{M}$ for every $j\leq k\alpha n$, and the other information can be obtained from $I_R$, we know that $M$ can be computed from $I_R$ and $\bar{M}$.
Finally, since $\vecz=\vecz'$, $\vecz$ can also be computed from $I_R$. We conclude that Bob can compute $(M,\vecz)$ using $I_R$ and $(\bar{M},\bar{\vecz},\bar{\veca})$.
\end{itemize}
\end{proof}

\begin{claim}[The distribution of $I$]\label{claim:SD boolean to non-boolean correctness}
Let $\bar{I} = (\bar{\vecx},\bar{\Gamma},\bar{\vecb},\bar{M},\bar{\vecz},\bar{\veca})$ be an instance drawn from either the Yes or No distribution of $(\bar{\cA}_Y,\bar{\cA}_N)$-advice-SD and $I_R = (\vecx_R,\Gamma_R,\vecb_R,M_R,\vecz_R,\veca_R)$ be a instance drawn from the Yes distribution of $(\cA_Y,\cA_N)$-advice-SD. Let $I = (\vecx,\Gamma,\vecb,M,\vecz,\veca)$ be the result of applying the above reduction on $\bar{I}$ and $I_R$. Then the following hold.
\begin{itemize}
\item $\vecx\sim\unif([q]^n)$.
\item $M$ is a uniformly random partial permutation matrix as required in the item 3 of~\cref{def:advice_sd}.
\item Suppose there is no error happening in the reduction.
\begin{itemize}
    \item If $\bar{I}$ is a Yes instance, then $\Pr[\vecz(j)=1]=\Pr_{\vecb(j)\sim\cA}[(M\vecx)(j)=\vecb(j)]$ for every $j\in[\alpha n]$.
    \item If $\bar{I}$ is a No instance, then $\Pr[\vecz(j)=1]=\Pr_{\vecb(j)\sim\cA'}[(M\vecx)(j)=\vecb(j)]$ for every $j\in[\alpha n]$.
\end{itemize}
\end{itemize}
Namely, if $\bar{I}$ is a Yes (resp. No) instance of $(\bar{\cA}_Y,\bar{\cA}_N)$-advice-SD, then $I$ is a Yes (resp. No) instance of $(\cA_Y,\cA_N)$-SD.
\end{claim}
\begin{proof}
\mbox{}
\begin{itemize}
\item To prove $\vecx\sim\unif([q]^n)$, observe that $\vecx$ is obtained from $\vecx_R$ by flipping some of the $u_i$ to $v_i$ (and vice versa). In particular, (i) $\vecx_R\sim\unif([q]^n)$ and (ii) the flipping is decided by $\bar{\vecx}$, which is uniformly sampled from $\{0,1\}^{\bar{n}}$ and is independent to $\vecx_R$. Note that for a fixed $\vecx_R$, $S_i$, and $j\in S_i$, the probability of $x_j$ being set to $u_i$ is the same as being set to $v_i$. As a result, by symmetry of $u_i$ and $v_i$, we conclude that $\vecx\sim\unif([q]^n)$.

\item By the symmetry of the $n$ variables, $M$ is a uniformly random partial permutation matrix as required in the item 3 of~\cref{def:advice_sd}. 

\item Suppose there is no error happening in the reduction. We consider the following two cases: (i) $j\in[\alpha n]\backslash U$ and (ii) $j\in U$.
\begin{enumerate}[label=(\roman*)]
    \item For each $j\in[\alpha n]\backslash U$, by the construction we have $\vecz(j)=\vecz_R(j)$ and hence when fixing $\vecx_R,M_R$, we have $\Pr[\vecz(j)=1]=\Pr[\vecz_R(j)=1]=\Pr_{\vecb_R(j)\sim\cA}[(M_R\vecx_R)(j)=\vecb_R(j)]$. We set $\vecb(j) = \vecb_R(j)$ and note that $\vecb(j)\sim\cA_Y$ (resp. $\vecb(j)\sim\cA_N$) if $\bar{\vecb}(j)\sim\bar{\cA}_Y$ (resp. $\bar{\vecb}(j)\sim\bar{\cA}_N$) for every $j\in U$. Finally, since $j\notin U$, there exists $i\in[k]$ such that $(M_R\vecx_R(j))_i=(M\vecx(j))_i\notin\{u_i,v_i\}$ and hence $\Pr_{\vecb_R(j)\sim\cA_Y}[(M_R\vecx_R)(j)=\vecb_R(j)]=\Pr[(M\vecx)(j)=\vecb(j)]=0$. So we have $\Pr[\vecz(j)=1]=\Pr_{\vecb(j)\sim\cA_Y}[(M\vecx)(j)=\vecb(j)]$ (resp. $\Pr[\vecz(j)=1]=\Pr_{\vecb(j)\sim\cA_N}[(M\vecx)(j)=\vecb(j)]$) if $\bar{I}$ is a Yes (resp. No) instance as desired.
    
    \item For each $j\in U$, by construction we have $\vecz(j)=\bar{\vecz}(V(j))$. We set
    \[
    \vecb(j)_i =  \left \{ \begin{array}{ll}
        u_i &  \mbox{ if $u_i < v_i$ and } \bar{\vecb}(V(j))_i = 0 \\
        u_i &  \mbox{ if $u_i > v_i$ and } \bar{\vecb}(V(j))_i = 1 \\
        v_i &  \mbox{ otherwise.}
        \end{array} \right .
    \]
    First, observe that $\vecz(j)=1$ iff $(M\vecx)(j)=\vecb(j)$. To see this, note that
    \begin{align*}
        \vecz(j)=1 &\Leftrightarrow \bar{\vecz}(V(j))=1\\
        &\Leftrightarrow (\bar{M}\bar{\vecx})(V(j))=\bar{\vecb}(V(j)) \, .
        \intertext{For each $i\in[k]$, if $u_i<v_i$ and $\bar{\vecb}(V(j))_i=(\bar{M}\bar{\vecx})(V(j))_i=0$, we have $\vecb(j)_i=(M\vecx)(j)_i=u_i$. Similarly, for all the other situations we have $\vecb(j)_i=(M\vecx)(j)$ and hence the equation becomes}
        &\Leftrightarrow (M\vecx)(j)=\vecb(j)
    \end{align*}
    as desired.
    
    Next, observe that if $\bar{I}$ is a Yes (resp. No) instance, then $\vecb(j)\sim\cA_Y$ (resp. $\vecb(j)\sim\cA_N$). We analyze the two cases as follows. 
    \begin{itemize}
        \item If $\bar{I}$ is a Yes instance, we have $\bar{\vecb}(V(j))\sim\bar{\cA}_Y=\unif(\{\bar{\vecu},\bar{\vecv}\})$. Recall that $(\bar{u}_i,\bar{v}_i)=(0,1)$ if $u_i<v_i$ and $(\bar{u}_i,\bar{v}_i)=(1,0)$ otherwise. Now observe that, by the above choice of $\vecb(j)$, we have $\bar{\vecb}(V(j))=\bar{\vecu}$ iff $\vecb(j)=\vecu$ (resp. $\bar{\vecb}(V(j))=\bar{\vecv}$ iff $\vecb(j)=\vecv$). Thus, we have $\vecb(j)\sim\cA_Y$ as desired.
        \item If $\bar{I}$ is a No instance, we have $\bar{\vecb}(V(j))\sim\bar{\cA}_N=\unif(\{\bar{\vecu}\vee\bar{\vecv},\bar{\vecu}\wedge\bar{\vecv}\})$. Recall that for each $i\in[k]$, $u_i\vee v_i=\max\{u_i,v_i\}$ and $u_i\wedge v_i=\min\{u_i,v_i\}$. Now observe that, by the above choice of $\vecb(j)$, we have $\bar{\vecb}(V(j))=\bar{\vecu}\vee\bar{\vecv}$ iff $\vecb(j)=\vecu\vee\vecv$ (resp. $\bar{\vecb}(V(j))=\bar{\vecu}\wedge\bar{\vecv}$ iff $\vecb(j)=\vecu\wedge\vecv$). Thus, we have $\vecb(j)\sim\cA_N$ as desired.
    \end{itemize}
    
    To sum up, for each $j\in U$, we have $\Pr[\vecz(j)=1]=\Pr_{\vecb(j)\sim\cA_Y}[(M\vecx)(j)=\vecb(j)]$ (resp. $\Pr[\vecz(j)=1]=\Pr_{\vecb(j)\sim\cA_N}[(M\vecx)(j)=\vecb(j)]$) if $\bar{I}$ is a Yes (resp. No) instance as desired.
\end{enumerate}
\end{itemize}
\end{proof}
\paragraph{Step 4: An error occurs with low probability.}

\begin{claim}\label{claim:SD boolean to non-boolean error}
When $n$ is sufficiently large, the probability of an error happening in the reduction is at most $2^{-\Omega((2/q)^{2k}\alpha n)}$.
\end{claim}
\begin{proof}
Recall that for given $\bar{n},k,\bar{\alpha},q,\delta$, we let $n=2q\bar{n}$ and $\alpha=q^{k-1}2^{-(k+2)}\bar{\alpha}$.

Note that $U$ is a sum of $\alpha n$ i.i.d. $\textsf{Bern}((2/q)^k)$. So by concentration inequality, we have $\Pr[|U|>2(2/q)^k\alpha n]<2^{-\Omega((2/q)^{2k}\alpha n)}$. By the choice of parameters, we have $2(2/q)^k\alpha n\leq\bar{\alpha}\bar{n}$. Thus, type (1) error happens with probability at most $2^{-\Omega((2/q)^{2k}\alpha n)}$.

Note that by the choice of parameters, we have $|X_i|=|U|\leq\bar{n}/k$ and hence type (2) error happens only when $|U|+|W_i|<\bar{n}/k$ for some $i\in[k]$. For each $i\in[k]$, note that $|W_i|$ is a sum of $n/k-\alpha n$ i.i.d. $\textsf{Bern}(2/q)$. So by concentration inequality, we have $\Pr[|W_i|<(n/k-\alpha n)/q]<2^{-\Omega((1/q)^2(n/k-\alpha n))}$. By the choice of parameters, we have $(n/k-\alpha n)/q\geq\bar{n}/k$. Thus, type (2) error happens with probability at most $2^{-\Omega((1/q)^2(n/k-\alpha n))}\leq2^{-\Omega((2/q)^{2k}\alpha n)}$.
\end{proof}

\paragraph{Step 5: Proof of~\cref{lem:SD boolean to non-boolean}.}
\begin{proof}[Proof of~\cref{lem:SD boolean to non-boolean}]

For every $\bar{n},k,\bar{\alpha},q,\delta$, we let $n=2q\bar{n}$ and $\alpha=q^{k-1}2^{-(k+2)}\bar{\alpha}$.
Suppose there is a protocol for $(\cA_Y,\cA_N)$-SD using $C(n)$ bits of communication and achieving advantage $\delta$. We show that how to get a protocol $\bar{\Pi}$ for $(\bar{\cA}_Y,\bar{\cA}_N)$-advice-SD with parameters $(\bar{n},\bar{\alpha})$ using $C(n)$ bits of communication and achieving advantage $\delta/2$.

Let $\bar{I} = (\bar{\vecx},\bar{\Gamma},\bar{\vecb},\bar{M},\bar{\vecz},\bar{\veca})$ be an instance drawn from either the Yes or No distribution of $(\bar{\cA}_Y,\bar{\cA}_N)$-advice-SD where $(\bar{\vecx},\bar{\veca})$ is Alice's private input and $(\bar{M},\bar{\vecz},\bar{\veca})$ is Bob's private input. The protocol $\bar{\Pi}$ works as follows. Alice and Bob first use their private input and the shared randomness to compute $\vecx$ and $(M,\vecz)$ respectively. This can be done due to~\cref{claim:SD boolean to non-boolean local}. Next, Alice and Bob simply invoke the protocol $\Pi$ on the new instance $\vecx$ and $(M,\vecz)$ and output the result accordingly.

It is immediate to see that $\bar{\Pi}$ only uses $C(n)$ bits of communication. To show that $\bar{\Pi}$ has advantage at least $\delta/2$, we first show that the joint distribution of $(\vecx,M,\vecz)$ is the same as that from an instance of $(\cA_Y,\cA_N)$-SD if there is no error in the reduction. 
By~\cref{claim:SD boolean to non-boolean correctness}, $\vecx\sim\unif([q]^n)$ and $M$ follows the distribution as required in the item 3 of~\cref{def:advice_sd}. 

When there is no error in the reduction and $\bar{I}$ is sampled from the Yes (resp. No) distribution of $(\bar{\cA}_Y,\bar{\cA}_N)$-advice-SD,~\cref{claim:SD boolean to non-boolean correctness} implies that $\vecz$ follows the conditional distribution (conditioned on $\vecx$ and $M$) of a Yes (resp. No) instance of $(\cA_Y,\cA_N)$-SD as required in the item 5 of~\cref{def:advice_sd}. Next,~\cref{claim:SD boolean to non-boolean error} shows that the probability of an error happening in the reduction is at most $\delta/2$. Finally, by triangle inequality, we conclude that $\bar{\Pi}$ has advantage at least $\delta/2$ in solving $(\bar{\cA}_Y,\bar{\cA}_N)$-advice-SD.

To conclude, by~\cref{thm:kpart}, any protocol for $(\bar{\cA}_Y,\bar{\cA}_N)$-advice-SD with advantage $\delta/2$ requires $\bar{\tau}\sqrt{\bar{n}}$ bits of communication. Thus, we have $C(n)\geq\bar{\tau}\sqrt{\bar{n}}\geq\tau\sqrt{n}$ for some constant $\tau>0$.
\end{proof}

\subsection{Indistinguishability of shifting distributions}\label{sec:indis SD shift}
In this subsection, we prove the following lemma which was used in \cref{sec:mixed} for reducing a single-function SD to a multi-function SD, and will be used in \cref{sec:final} for reductions between various \textsf{SD} problems.

\restindis*

\begin{proof}
Given the parameters $n,\alpha$ and $\epsilon\in(0,1)$, define $n'=\epsilon n$ and $\alpha'=2\alpha$.

Let $(\vecx',M',\vecb',\vecz')$ be an instance of the $(\cF,\cD_1,\cD_2)$-\textsf{SD} problem where $\vecx'\in[q]^{n'}$, $M'\in\{0,1\}^{k\alpha'n'\times n'}$, $\vecb'\in[q]^{k\alpha'n'}$, $\vecz'\in\{0,1\}^{\alpha'n'}$. Let $R'$ be the shared randomness defined later. We specify the map
$(\vecx',M',\vecb',\vecz',R')\mapsto(\vecx,M,\vecb,\vecz)$ where $\vecx\in[q]^{n}$, $M\in\{0,1\}^{k\alpha n\times n}$, $\vecb\in[q]^{k\alpha n}$, $\vecz\in\{0,1\}^{\alpha n}$.

\begin{reduction}{A reduction from $(\cF,\cD_1,\cD_2)$-\textsf{SD} to $(\cF,\cD_Y,\cD_N)$-\text{SD}}
Let $\vecy \sim \textsf{Unif}([q]^{n-n'})$, $\vecw \sim \bern(2\eps)^{\alpha n}$. Let $\Gamma \in \{0,1\}^{n\times n}$ be a uniform permutation matrix. Let $\vecc = (\vecc(1),\ldots,\vecc((n-n')/k))$ where $\vecc(i) \sim \cD$ are chosen independently.
\begin{itemize}
\item Let $R' = (\vecy,\vecw,\Gamma,\vecc)$ be the shared randomness. 
\end{itemize}
Let $\#_w(i) = |\{j \in [i]\, |\, w_j = 1\}|$ denote the number $1$'s among the first $i$ coordinates of $\vecw$. 
If $\#_w(\alpha n) \geq \alpha' n'$ or if $\alpha n - \#_w(\alpha n) \geq (n-n')/k$ we declare an error, Note $\Exp[\#_w(n)] = \alpha'n'/2$ so the probability of error is negligible (specifically it is $\exp(-n)$).

Given $(\vecx',M',\vecb',\vecz',R')$, we now define $(\vecx,M,\vecb,\vecz)$ as follows.
\begin{itemize}
\item Let $\vecx = \Gamma (\vecx',\vecy)$ so $\vecx$ is a random permutation of the concatenation of $\vecx'$ and $\vecy$.
\item Let $M' = (M'_1,\ldots,M'_{\alpha'n'})$ where $M'_i \in \{0,1\}^{k \times n'}$. We extend $M'_i$ to $N_i \in \{0,1\}^{k \times n}$ by adding all-zero columns to the right. For $i \in \{1,\ldots,(n-n')/k\}$, let $P_i \in \{0,1\}^{k \times n}$ be given by $(P_i)_{j\ell} =1$ if and only if $\ell = n' +(i-1)k + j$. Next we define a matrix $\tilde{M} \in \{0,1\}^{k\alpha n \times n} = (\tilde{M}_1,\ldots,\tilde{M}_{\alpha n})$ where $\tilde{M}_i \in \{0,1\}^{k \times n}$ is defined as follows: If $w_i = 1$ then we let $\tilde{M}_i = N_{\#_w(i)}$ else we let $\tilde{M}_i = P_{i - \#_w(i)}$. Finally we let $M = \tilde{M} \cdot \Gamma^{-1}$.
\item Let $\vecb = (\vecb(1),\ldots,\vecb(\alpha n))$ where $\vecb(i)= \vecb'(\#_w(i))$ if $w_i = 1$, otherwise $\vecb(i)= \vecc(i-\#_w(i))$.
\item Let $z_i=1$ if and only if $M_i\vecx=\vecb(i)$ for every $i\in[\alpha n]$.
\end{itemize}
\end{reduction}

Now, we verify that the reduction satisfies the following success conditions.

\begin{reduction}{Success conditions for the reduction}
\begin{enumerate}[label=(\arabic*)]
\item \textbf{The reduction is locally well-defined.} Namely, there exist random strings $R'$ so that (i) Alice can get $\vecx$ through a map $(\vecx',R') \mapsto \vecx$ while Bob can get $(M,\vecz)$ through a map $(M',\vecz',R') \mapsto (M,\vecz)$.

\item \textbf{The reduction is sound and complete.} Namely, (i) $z_i=1$ if and only if $M_i\vecx = \vecb(i)$ for all $i\in[\alpha n]$. (ii) If $\vecb' \sim \cD_1^{\alpha'n'}$, then $\vecb \sim \cD_Y^{\alpha n}$. Similarly if $\vecb' \sim \cD_2^{\alpha'n'}$, then $\vecb \sim \cD_N^{\alpha n}$. (iii) $\vecx\sim \unif([q]^{n})$ and $M$ is a uniformly random matrix conditioned on having exactly one ``$1$'' per row and at most one ``$1$'' per column.
\end{enumerate}
\end{reduction}

\begin{claim}\label{claim:polarization indis shift}
If $\#_w(\alpha n) \leq \alpha' n'$ and  $\alpha n - \#_w(\alpha n) \leq (n-n')/k$, then the second map in the reduction is locally well-defined, sound, and complete. In particular, the error event happens with probability at most $\exp(-\Omega(n))$ over the randomness of $R'$.
\end{claim}
\begin{proof}
To see the reduction is locally well-defined, first note that Alice can compute $\vecx=\Gamma(\vecx',\vecy)$ from $\vecx'$ and the shared randomness $R'$ locally. As for Bob, note that the maximum index needed for $N$ and $\vecb'$ (resp. $P$ and $\vecc$) is at most $\#_w(\alpha n)$ (resp. $\alpha n-\#_w(i)$). Namely, if $\#_w(\alpha n) \leq \alpha' n'$ and $\alpha n - \#_w(\alpha n) \leq (n-n')/k$, then $M$ and $\vecb$ are well-defined. Note that this happens with probability at least $1-2^{-\Omega(n)}$. Also, one can verify from the construction that $M$ and $\vecb$ can be locally computed by $M'$, $\vecb'$, and the shared randomness $R'$.

To see the reduction is sound and complete, (i) $z_i=1$ if and only if $M_i\vecx=\vecb(i)$ for every $i\in[\alpha n]$ directly follows from the construction. As for (ii), if $\vecb'\sim\cD_1^{\alpha'n'}$. Now, for each $i\in[\alpha n]$, $\vecb(i)=\vecb'(\#_w(i))$ with probability $\epsilon$ and $\vecb(i)=\vecc(i-\#_w(i))$ with probability $1-\epsilon$. As $\vecb'(i')\sim\cD_1$ for every $i'\in[\alpha'n']$ and $\vecc(i')\sim\cD_0$ for every $i'\in[(n-n')/k]$, we have $\vecb(i)\sim(1-\epsilon)\cD_0+\epsilon\cD_1=\cD_Y$ as desired. Similarly, one can show that if $\vecb'\sim\cD_2^{\alpha'n'}$, then for every $i'\in[\alpha'n']$ we have $\vecb(i')\sim\cD_N$. Finally, we have $\vecx\sim\unif([q]^n)$ and $M$ is a uniformly random matrix with exactly one ``$1$'' per row and at most one ``$1$'' per column (due to the application of a random permutation $\Gamma$) by construction.

This completes the proof of the success conditions (1)-(2) for the reduction.
\end{proof}
To wrap up the proof of~\cref{lem:indis of shift}, suppose there is a protocol $\Pi$ for $(\cF,\cD_Y,\cD_N)$-\textsf{SD} with parameters $n$ and $\alpha$ using $c$ bits of communication with advantage $\delta$. We describe a protocol $\Pi'$ for $(\cF,\cD_1,\cD_2)$-\textsf{SD} with parameters $n'$ and $\alpha'$ using $c$ bits of communication with advantage at least $\delta-2^{-\Omega(n)}$.

Let $(\vecx',M',\vecb',\vecz')$ be an instance of the $(\cF,\cD_1,\cD_2)$-\textsf{SD} problem where $\vecx'\in[q]^{n'}$, $M'\in\{0,1\}^{k\alpha'n'\times n'}$, $\vecb'\in[q]^{k\alpha'n'}$, $\vecz'\in\{0,1\}^{\alpha'n'}$. Let $R'$ be the shared randomness defined above. In the new protocol $\Pi'$, Alice and Bob computes their private inputs $\vecx$ and $(M,\vecz)$ respectively. By~\cref{claim:polarization indis shift}, the computation can be done locally with their original private inputs and the shared randomness. Also, with probability at least $1-2^{-\Omega(n)}$, the Yes (resp. No) instance of $(\cF,\cD_1,\cD_2)$-\textsf{SD} is mapped to the Yes (resp. No) instance of $(\cF,\cD_Y,\cD_N)$-\textsf{SD}. Namely, by directly applying $\Pi$ on the new inputs, Alice and Bob can achieve $\delta-2^{-\Omega(n)}$ advantage on $(\cF,\cD_1,\cD_2)$ using the same amount of communication as desired.
\end{proof}

\subsection{Proof of Theorem~\ref{thm:polarization step hardness}}\label{sec:final}

Let $\vecu,\vecv$ be incomparable, let $S=\{i\in[k]\, |\, u_i\neq v_i \}$, and let $k''=|S|$.

\paragraph{Step 1: Specify the auxiliary distributions.}
\begin{itemize}
\item Let $\cA_Y=\unif(\{\vecu|_S,\vecv|_S\})$ and $\cA_N=\unif(\{(\vecu|_S)\vee(\vecv|_S),(\vecu|_S)\wedge(\vecv|_S)\})$. By~\cref{lem:SD boolean to non-boolean}, $(\cA_Y,\cA_N)$-\textsf{SD} requires $\tau\sqrt{n}$ space.

\item Let $\cD_1=\unif(\{\vecu,\vecv\})$ and $\cD_2=\unif(\{\vecu\vee\vecv,\vecu\wedge\vecv\})$.

\item Finally, there exists $\cD_0$ such that we have $\cD=(1-2\epsilon)\cD_0+2\epsilon\cD_1$ and $\cD_{\vecu,\vecv}=(1-2\epsilon)\cD_0+2\epsilon\cD_2$.
\end{itemize}

In the following, we are going to describe reduction from $(\cA_Y,\cA_N)$-\textsf{SD} with parameters $(n'',\alpha'',k'')$ to $(\cD_1,\cD_2)$-\textsf{SD} with parameters $(n',\alpha',k)$. And by~\cref{lem:indis of shift}, there exists a reduction from $(\cD_1,\cD_2)$-\textsf{SD} with parameters $(n',\alpha',k)$ to $(\cD,\cD_{\vecu,\vecv})$-\textsf{SD} with parameters $(n,\alpha,k)$.

\paragraph{Step 2: Overview of the reduction from $(\cA_Y,\cA_N)$-\textsf{SD} to $(\cD,\cD_{\vecu,\vecv})$-\textsf{SD}.}
Let $\Pi$ be a protocol for $(\cD,\cD_{\vecu,\vecv})$-\textsf{SD} with parameter $\alpha\leq1/(200k)$ using $C(n)$ communication bits to achieve advantage $\delta$ on instances of length $n$. We let $n''=(k''\epsilon/k)n$, $\alpha''=(2k/k'')\alpha$ and design a protocol $\Pi''$ for $(\cA_Y,\cA_N)$-\textsf{SD} with parameter $\alpha''$ achieving advantage at least $\delta/2$ on instances of length $n''$ using $C''(n'')=C(n)$ communication. 
Thus, by~\cref{lem:SD boolean to non-boolean}, there exists a constant $\tau''>0$ such that $C(n)=C''(n'')\geq\tau''\sqrt{n''}=\tau''\sqrt{(k''\epsilon/k)}\sqrt{n}$ as desired.

To construct such reduction, we first reduce the above instance of $(\cA_Y,\cA_N)$-\textsf{SD} to an instance of $(\cD_1,\cD_2)$-\textsf{SD} with parameters $n'=kn''/k''$ and $\alpha'=\alpha''n''/n'$. Next, we invoke~\cref{lem:indis of shift} to get a protocol $\Pi'$ (from $\Pi$) which achieves $\delta-2^{-\Omega(n)}$ advantage on $(\cD_1,\cD_2)$-\textsf{SD} using $C(n)$ communication.

Without loss of generality, we assume $\Pi'$ is deterministic and our new protocol $\Pi''$ for $(\cA_Y,\cA_N)$-\textsf{SD} uses shared randomness between Alice and Bob. The protocol $\Pi''$ is a map: $(\vecx'',M'',\vecb'',\vecz'',R'')\mapsto(\vecx',M',\vecb',\vecz')$.

Before describing the map, let us first state the desired conditions.

\begin{reduction}{Success conditions for the reduction}
\begin{enumerate}[label=(\arabic*)]
\item \textbf{The reduction is locally well-defined.} Namely, there exist a random string $R''$ so that (i) Alice can get $\vecx'$ through the maps $(\vecx'',R'') \mapsto \vecx'$ while Bob can get $(M',\vecz')$ through the map $(M'',\vecz'',R'') \mapsto (M',\vecz')$.

\item \textbf{The reduction is sound and complete.} Namely, (i) $z'_{i}=1$ if and only if $M'_{i}\vecx' = \vecb'(i)$ for all $i\in[\alpha'n']$. (ii) If $\vecb'' \sim \cA_Y^{\alpha''n''}$ then $\vecb' \sim \cD_1^{\alpha'n'}$. Similarly if $\vecb'' \sim \cA_N^{\alpha''n''}$ then $\vecb' \sim \cD_2^{\alpha'n'}$. (iii) $\vecx'\sim \unif([q]^{n'})$ and $M'$ is a uniformly random matrix conditioned on having exactly one ``$1$'' per row and at most one ``$1$'' per column.

\end{enumerate}
\end{reduction}

\paragraph{Step 3: Specify and analyze the reduction from $(\cA_Y,\cA_N)$-\textsf{SD} to $(\cD_1,\cD_2)$-\textsf{SD}.}
We now specify the first map mentioned above and prove that it satisfies conditions (1)-(2). 

\begin{reduction}{A reduction from $(\cA_Y,\cA_N)$-\textsf{SD} to $(\cD_1,\cD_2)$-\textsf{SD}}
\begin{itemize}
\item Let $R'' \sim \unif([q]^{n'-n''})$ be the shared randomness.
\end{itemize}
Given $(\vecx'',M'',\vecb'',\vecz'',R'')$ we define $(\vecx',M',\vecb',\vecz')$ as follows. To get $M'$, $\vecz'$ and $\vecb'$ we need some more notations. First, note that $\alpha''n''=\alpha'n'$ due to the choice of parameters.
\begin{itemize}
\item Let $\vecx' = (\vecx'',R'')$.
\item $M'$ can be viewed as the stacking of matrices $M''_1,\ldots,M''_{\alpha'' n''} \in \{0,1\}^{k'' \times n''}$. We first extend $M''_i$ by adding all-zero columns at the end to get $N'_i \in \{0,1\}^{k'' \times n'}$. We then stack $N'_i$ on top of $P'_i \in \{0,1\}^{(k-k'')\times n'}$ to get $M'_i$, where $(P'_i)_{j\ell} = 1$ if and only if $\ell = n'' + (i-1)k + j$. We let $M'$ be the stacking of $M'_1,\ldots,M'_{\alpha' n'}$.
\item Let $\vecb'' = (\vecb''(1),\cdots,\vecb''(\alpha' n'))$. Let $\tilde{\vecu} = (u_{k''+1},\ldots,u_{k})$ denote the common parts of $\vecu$ and $\vecv$. We let
$\vecb'(i) = (\vecb''(i),\tilde{\vecu})$ and $\vecb' = (\vecb'(1),\cdots,\vecb'(\alpha' n'))$.
\item Let $z'_i=1$ if and only if $M'_i\vecx' = \vecb'(i)$ for all $i\in[\alpha'n']$ as required.
\end{itemize}
\end{reduction}

\begin{claim}\label{claim:polarization indis map 1}
The above reduction is locally well-defined, sound, and complete.
\end{claim}
\begin{proof}
To see the map is locally well-defined, note that Alice can compute $\vecx'= (\vecx'',R'')$ locally. Similarly, Bob can compute $M'$ locally by construction. As for $\vecz'$, note that for every $i\in[\alpha'n']$, $z'_i=1$ if and only if $z''_i=1$ and $P_i'\vecx'=\tilde{\vecu}$. Since Bob has $\vecz'$ and can locally compute $P_i'\vecx'$ for every $i$, her can also compute $\vecz'$ locally.

To see the map is sound and complete, (i) $z'_i=1$ if and only if $M'_i\vecx'=\vecb'(i)$ follows from the construction. As for (ii), for each $i\in[\alpha'n']=[\alpha''n'']$, if $\vecb''_i\sim\cA_Y=\unif(\{\vecu|_S,\vecv|_S\})$, then $\vecb'_i\sim\unif(\{(\vecu|_S,\tilde{\vecu}),(\vecv|_S,\tilde{\vecu})\})=\unif(\{\vecu,\vecv\})=\cD_1$ as desired. Similarly, one can show that if $\vecb''_i\sim\cA_N$, then $\vecb'_i\sim\cD_1$. Finally, we have $\vecx'\sim\unif([q]^{n'})$ by construction and hence (iii) holds.

This completes the proof of conditions (1)-(2) for the reduction.
\end{proof}

\paragraph{Step 4: Proof of~\cref{thm:polarization step hardness}.}
\begin{proof}[Proof of~\cref{thm:polarization step hardness}]

Let us start with setting up the parameters. Given $k\in(0,1/(200k)),\alpha,n,\cD$, and incomparable pair $(\vecu,\vecv)\in\supp(\cD)$ and polarization amount $\epsilon=\epsilon(\cD,\vecu,\vecv)$, let $k''=|\{i\in[k]\, |\, u_i\neq v_i\}|$, $n''=(k''\eps/k)n$, $\alpha''=(2k/k'')\alpha$, $n'=kn''/k''$, $\alpha'=\alpha''n'/n'$, and $\delta''=\delta/2$.

Now, for the sake of contradiction, we assume that there exists a protocol $\Pi=(\Pi_A,\Pi_B)$ for $(\cD,\cD_{\vecu,\vecv})$-\textsf{SD} with advantage $\delta$ and at most $\tau\sqrt{n}$ bits of communication.

First, by \cref{claim:polarization indis map 1}, if $(\vecx'',M'',\vecz'')$ is a Yes (resp. No) instance of $(\cA_Y,\cA_N)$-\textsf{SD}, then the output of the reduction, i.e., $(\vecx',M',\vecz')$, is a Yes (resp. No) instance of $(\cD_1,\cD_2)$-\textsf{SD}. Next, Alice and Bob run the protocol $\Pi'$ from~\cref{lem:indis of shift} on $(\vecx',M',\vecz')$. By the correctness of the reduction as well as the protocol $\Pi'$, we know that Alice and Bob have advantage at least $\delta-\exp(-\Omega(n))\geq\delta/2=\delta''$ in solving $(\cA_Y,\cA_N)$-\textsf{SD} with at most $\tau\sqrt{n}=\tau\sqrt{(k/(k''\epsilon))n''}$ bits of communication.

Finally, by~\cref{lem:SD boolean to non-boolean}, we know that there exists a constant $\tau_0>0$ such that any protocol for $(\cA_Y,\cA_N)$-\textsf{SD} with advantage $\delta''$ requires at least $\tau_0\sqrt{n''}$ bits of communication. This implies that $\tau\geq\tau_0\sqrt{k''\epsilon/k}$. We conclude that any protocol for $(\cD,\cD_{\vecu,\vecv})$-\textsf{SD} with advantage $\delta$ requires at least $\tau\sqrt{n}$ bits of communication.
\end{proof}

\section{Dichotomy for exact Computation}\label{sec:exact}

In this section we prove \cref{thm:exact}. For this, we will use tight bounds on the randomized communication complexity of the Disjointness (\Disj) and Gap Hamming Distance (\GHD) problems.
\begin{definition}[Disjointness (\Disj)]
In the $\Disj_{n}$ problem, Alice and Bob receive binary strings $x,y\in\{0,1\}^n$ of Hamming weight $\Delta(x)=\Delta(y)=n/4$, respectively. If the Hamming distance $\Delta(x,y)=n/2$ the players must output $1$, if $\Delta(x,y)< n/2$ they must output $0$.
\end{definition}
\begin{definition}[Gap Hamming Distance (\GHD)]
In the $\GHD_{n,t,g}$ problem, Alice and Bob receive binary strings $x,y\in\{0,1\}^n$, respectively. If the Hamming distance $\Delta(x,y)\geq t+g$ the players must output $1$, if $\Delta(x,y)\leq t-g$ they must output $0$, otherwise they may output either $0$ and $1$.
\end{definition}
The following results give tight bounds on the randomized communication complexity of $\Disj$ and $\GHD$.
\begin{theorem}[\cite{ks92,r90}]\label{thm:disj}
For all large enough $n$, any randomized protocol solving $\Disj_n$ with probability $2/3$ must use $\Omega(n)$ bits of communication.
\end{theorem}
\begin{theorem}[\cite{cr12,v12,s12}]\label{thm:ghd}
For every $a\in(0,1/2]$ and every $g\geq 1$, and all large enough $n$ the following holds. If $t\in[an, (1-a)n]$, then any randomized protocol solving $\GHD_{n,t,g}$ with probability $2/3$ must use $\Omega(\min\{n, n^2/g^2\})$ bits of communication.
\end{theorem}

Equipped with these results, we are ready to prove \cref{thm:exact}.

\exactthm*

While this theorem doesn't give tight bounds on the space complexity of $\maxF$ in terms of $n$, the dependence on $\eps$ is tight. For every family of functions $\cF$, if we sample $O(n/\eps^2)$ random constraints, then by the Chernoff bound we preserve the values of all assignments within a factor of $1\pm\eps$.
\begin{proof}
For the first item of the theorem, we note that the maximum number of simultaneously satisfiable constraints in a $\sigma$-satisfiable formula is the number of non-zero constraints $f \in \cF\setminus\{\veczero\}$ in it. This can be computed in space $O(\log{n})$.

Now we turn to the proof of the second item of the theorem in the streaming setting. 
To this end, first we prove that there exists an unsatisfiable instance $I$ of $\textsf{Max-CSP}(\cF\setminus\{\veczero\})$. Let $I$ be an instance on $kq$ variables that has every constraint from $\cF\setminus\{\veczero\}$ applied to every (unordered) $k$-tuple of distinct variables. Any assignment $\vecnu\in[q]^{kq}$ has at least $k$ equal coordinates. That is, there exists $\sigma\in[q]$ such that $\Sigma=\{i\colon \vecnu_i=\sigma\}$ has size $|\Sigma|\geq k$. Since $\cF$ is not $\sigma$-satisfiable, there exists a function $f\in\cF\setminus\{\veczero\}$ that $f(\sigma^k)\neq 1$. Thus, the corresponding constraint of $I$ is not satisfied by $\vecnu$.

Now we pick a minimal unsatisfied formula $J$ on $kq$ variables with constraints from $\cF\setminus\{\veczero\}$, that is a formula such that all proper subsets of the constraints of $J$ can be simultaneously satisfied. Since $J$ doesn't have zero-constraints, $J$ must have at least two constraints. We partition $J$ into two arbitrary non-empty subsets of constraints $J=J_A\sqcup J_B$. Note that by minimality of $J$, $J_A$ and $J_B$ are both satisfiable.

Observe that item 2(a) of the theorem follows from 2(b) by setting $\eps=\Theta(1/n)$. In order to prove the item 2(b), we reduce $\Disj_{m}$ for $m=|J|^{-1}\eps^{-1}$ to $\maxF$ on $n$ variables.  We can assume that $\eps\geq\frac{kq}{n|J|}$, as for smaller $\eps$ the optimal lower bound of $\Omega(n)$ is implied by this setting. We partition the $n$ variables of $\maxF$ into at least $m$ groups of size $kq$. Let $x,y\in\{0,1\}^{m}$ be the inputs of Alice and Bob in the $\Disj_{m}$ problem. If $x_i=1$, then Alice applies the constraints $J_A$ to the $i$th block of $kq$ variables of the formula. Similarly, if $y_i=1$, then Bob applies the constraints $J_B$ to the $i$th block of $kq$ variables. Let $C_A$ and $C_B$ be the sets of constraints produced by Alice and Bob, respectively, and let $\Psi=C_A\cup C_B$. Since $\Delta(x)=\Delta(y)=m/4$,  the total number of constraints in the formula $|\Psi|=|J|m/4$. Note that $\Psi$ is satisfiable if and only if $\Disj(x, y)=1$. Therefore, if $x$ and $y$ are disjoint, then $\val(C_A\cup C_B)=1$, otherwise
\[
\val(\Psi)\leq 1-\frac{4}{|J|m}<1-\eps
\, .
\]

Any streaming algorithm that receives constraints $C_A$ and $C_B$ and solves $(1,1-\eps)$-$\maxF$ with probability $2/3$, also solves the $\Disj_m$ problem. Therefore, by \cref{thm:disj}, such an algorithm must use space $\Omega(m)=\Omega(1/\eps)$.

In order to prove item 2(c), we reduce the $\GHD_{n,t,g}$ problem to $\maxF$ on $nkq+O(1)=O(n)$ variables, where $t=n(1-\gamma)$ and $g\geq1$ will be determined later. We will create two groups of constraints: the first group of constraints $C_A\cup C_B$ will have value $1-O(\Delta(x,y)/n)$, and the second group of constraints will have value close to $\rho_{\min}$. By taking a weighed combination of these two groups we will get a formula whose value is less than $\gamma-\eps$ for $\Delta(x,y)\geq t+g$, and value is at least $\gamma$ for $\Delta(x,y)\leq t$.

Again, we start with a minimal unsatisfiable formula on $kq$ variables. If $|J|=2d$ is even, then we arbitrarily partition $J$ into two sets of $d$ constraints $J_A$ and $J_B$. If $|J|$ is odd, then we add one constraint to $|J|$ as follows. By minimality, there is an assignment that satisfies $|J|-1$ constraints of $J$, let $c$ be one of these constraints. We add another copy of $c$ to $J$, and partition $J$ into two sets of $d$ constraints $J_A$ and $J_B$. Note that while $J_A$ and $J_B$ are satisfiable, only $2d-1$ constraints of $J_A\cup J_B$ can be satisfied simultaneously.

Let $x,y\in\{0,1\}^n$ be the inputs of Alice and Bob in the $\GHD_{n,t,g}$ problem. If $x_i=1$, then Alice applies the constraints $J_A$ to the $i$th block of $kq$ variables of the formula, otherwise Alice applies the constraint $J_B$ to these variables. Similarly, if $y_i=1$ or $y_i=0$, then Bob applies the constraints $J_A$ or $J_B$ to the $i$th block of $kq$ variables. Let $C_A$ and $C_B$ be the sets of constraints produced by Alice and Bob, respectively. Observe that $|C_A|=|C_B|=nd$. The set of constraints added by Alice and Bob when processing their $i$th coordinates is satisfiable if and only if $x_i=y_i$. When $x_i\neq y_i$, then by the construction of $J$, exactly $2d-1$ constraints are satisfiable. Therefore, 
\[
\val(C_A\cup C_B)=1-\frac{\Delta(x,y)}{2dn}
\, .
\]

Let $\gamma'=(\gamma+\rho_{\min})/2<\gamma$. By the definition of $\rho_{\min}(\cF)$, there exists $n_0$ and a formula $\Phi'$ of $\maxf$ such that $\val(\Phi')=\gamma'$. By taking several copies of $\Phi'$ on the same $n_0$ variables, we get an instance $\Phi$ with $D=|\Phi|=\frac{n(2d-1)(1-\gamma)}{\gamma-\gamma'}=\Theta(n)$ constraints and value $\val(\Phi)=\gamma'$.

Now we output an instance $\Psi$ of $\maxF$ on $nkq+n_0$ variables that is simply a union of $C_A \cup C_B$ and $\Phi$ on disjoint sets of variables. By construction.
\[
\val(\Psi)=\frac{(2dn-\Delta(x,y))+\gamma'D}{2dn+D}
\, .
\]
In the case when $\Delta(x,y)\leq t=(1-\gamma)n$, we have 
\[
\val(\Psi) \geq \frac{2dn-(1-\gamma)n+\gamma'D}{2dn+D}= \gamma 
\, .
\]
And for the case of $\Delta(x,y)\geq t+g=(1-\gamma)n+g$, we have that
\[
\val(\Psi) \leq \frac{(2dn-(1-\gamma)n-g)+\gamma'D}{2dn+D}= \gamma -\frac{g}{2dn+D}=\gamma-\eps
\]
for $g=\eps(2dn+D)=\Theta(n\eps)$.

Therefore, any streaming algorithm for $(\gamma,\gamma-\eps)$-$\maxF$ will imply a protocol for the $\GHD_{n,t,g}$ problem.
By \cref{thm:ghd}, such a streaming algorithm must use at least $\Omega(\min\{n, n^2/g^2\})=\Omega(\min\{n, \eps^{-2}\})$ bits of communication.
\end{proof}

\newcommand{\acktext}{
 We are grateful to Lijie Chen, Gillat Kol, Dmitry Paramonov, Raghuvansh Saxena, Zhao Song, and Huacheng Yu, for detecting a fatal error in an earlier version of this paper~\cite{CGSV20} and then for pinpointing the location of the error. As a result the main theorem of the current paper is significantly different than the theorem claimed in the previous version.

Thanks to Johan H\aa stad for many pointers to the work on approximation resistance and answers to many queries.
Thanks to Dmitry Gavinsky, Julia Kempe and Ronald de Wolf for prompt and detailed answers to our queries on outdated versions of their work~\cite{GKKRW}. Thanks to Prasad Raghavendra for answering our questions about the approximation resistance dichotomy from his work~\cite{raghavendra2008optimal}. Thanks to Saugata Basu for the pointers to the algorithms for quantified theory of the reals. Thanks to Jelani Nelson for pointers to $\ell_1$ norm estimation algorithms used in the earlier version of this paper. 
Thanks to Alex Andoni for pointers to $\ell_{1,\infty}$ norm estimation algorithms. Thanks to anonymous referees of many versions of this work for their valuables comments. In particular we thank the referees for clarifying the gap between linear sketching algorithms and dynamic streaming algorithms.}

\ifnum\jacm=1

 \begin{acks}
\acktext 
\end{acks}
  
\bibliographystyle{ACM-Reference-Format}
\else 
\section*{Acknowledgments}
\acktext 
\bibliographystyle{alpha}
\fi 

\bibliography{arXiv_revision2024/mybib}

\end{document}